\documentclass[12pt,a4paper]{article}
\usepackage[a4paper]{geometry}
\usepackage[english]{babel}
\usepackage[cp850]{inputenc}
\usepackage[dvips]{graphicx}
\usepackage{amsmath,amsfonts,amsthm,amssymb}
\usepackage{mathrsfs}
\usepackage[usenames, dvipsnames]{color}
\usepackage{xspace}
\usepackage{adjustbox}
\usepackage[colorlinks=true,citecolor=blue,linkcolor=blue,urlcolor=blue,pdfstartview=FitH]{hyperref}
\usepackage{nameref}
\usepackage{enumitem}
\usepackage{tabularx, booktabs}

\usepackage{algorithm}
\usepackage{algorithmicx}
\usepackage{algpseudocode}
\usepackage{caption}
\usepackage{geometry}
\usepackage{tikz,tkz-tab}
\usetikzlibrary{fit,calc,positioning,trees}
\usetikzlibrary{arrows}
\tikzstyle{level 1}=[level distance=4cm, sibling distance=2.5cm]
\tikzstyle{level 2}=[level distance=8cm, sibling distance=0.6cm]
\tikzstyle{bag} = [text width=4em, text centered]
\tikzstyle{end} = [circle, minimum width=3pt,fill, inner sep=0pt]
\usetikzlibrary{decorations.pathmorphing}
\usetikzlibrary{decorations.pathreplacing}
\usetikzlibrary{decorations.shapes}
\usetikzlibrary{decorations.text}
\usetikzlibrary{decorations.markings}
\usetikzlibrary{decorations.fractals}
\usetikzlibrary{decorations.footprints}

\graphicspath{{images/}}

\DeclareGraphicsExtensions{.pdf,.png}

\makeatletter
\newcommand{\manuallabel}[2]{\def\@currentlabel{#2}\label{#1}}
\makeatother
\newcommand\numberthis{\addtocounter{equation}{1}\tag{\theequation}}

\algnewcommand{\Inputs}[1]{%
  \State \textbf{Inputs:}
  \Statex \hspace*{\algorithmicindent}\parbox[t]{.8\linewidth}{\raggedright #1}
}
\algnewcommand{\Initialize}[1]{%
  \State \textbf{Initialize:}
  \Statex \hspace*{\algorithmicindent}\parbox[t]{.8\linewidth}{\raggedright #1}
}

\makeatletter
\newcommand*{\inlineequation}[2][]{%
  \begingroup
    \refstepcounter{equation}%
    \ifx\\#1\\%
    \else
      \label{#1}%
    \fi
    \relpenalty=10000 %
    \binoppenalty=10000 %
    \ensuremath{%
      #2 %
    }%
    ~\@eqnnum
  \endgroup
}
\makeatother

\usepackage{hyperref}
\usepackage{amsmath}

\DeclareMathAlphabet{\mathpzc}{OT1}{pzc}{m}{it}

\def\Esp{\mathbb{E}}

\def\é{\'{e}}
\def\è{\`{e}}
\def\ê{\^{e}}
\def\à{\`{a}}
\def\ô{\^{o}}
\DeclareMathOperator*{\esssup}{ess\,sup}

\newcolumntype{C}[1]{>{\centering\arraybackslash}p{#1}}
\newcolumntype{L}[1]{>{\raggedleft\arraybackslash}p{#1}}
\newcolumntype{R}[1]{>{\raggedright\arraybackslash}p{#1}}

\newtheorem{theo}{Theorem}
\newtheorem{prop}{Proposition}
\newtheorem{lem}{Lemma}
\newtheorem{Assumption}{Assumption}
\newtheorem{cor}{Corollary}
\newtheorem{defi}{Definition}
\newtheorem{rem}{Remark}

\newtheorem{example}{Example}

\title{From asymptotic properties of general point processes\\
to the ranking of financial agents}
\author{Othmane Mounjid\thanks{\'{E}cole Polytechnique, CMAP}, 
  Mathieu Rosenbaum\footnotemark[1] ~%
  and  %
  Pamela Saliba \footnotemark[1] \thanks{Autorit\'{e} des March\'{e}s Financiers}{}}
\date{\today\\}

\begin{document}
\maketitle

\begin{abstract}
\noindent We propose a general non-linear order book model that is built from the individual behaviours of the agents. Our framework encompasses Markovian and Hawkes based models. Under mild assumptions, we prove original results on the ergodicity and diffusivity of such system. Then we provide closed form formulas for various quantities of interest: stationary distribution of the best bid and ask quantities, spread, liquidity fluctuations and price volatility. These formulas are expressed in terms of individual order flows of market participants. Our approach enables us to establish a ranking methodology for the market makers with respect to the quality of their trading. 
\end{abstract}

\textbf{Keywords:} Market microstructure, limit order book, high-frequency trading, market making, queuing model, Hawkes processes, ergodic properties, volatility, regulation.


\section{Introduction}
In the last two decades, the development of electronic and fragmented markets has lead to a deep disruption in the landscape of market participants. In particular, traditional market making institutions have been largely replaced by high-frequency market makers. Market makers are intermediaries between buyers and sellers. In an electronic limit order book, they provide liquidity to market participants willing to trade immediately by simultaneously posting limit orders on both sides of the book. Market makers undergo different types of risk, mainly adverse selection and inventory risks. To avoid adverse selection risk, they must be able to update very frequently their quotes in response to other order submissions or cancellations. To minimise their inventory risk, they need to use smart algorithms enabling them to hold positions for very
short time periods only, see for example \cite{madhavan2000market}.\\

\noindent High-frequency traders (HFTs) are now the only market participants that are indeed able to play the role of market makers on liquid stocks, see \cite{jones2013we}. This is achieved thanks to an intense use of speed (co-location) and technology. They are supposedly capable to maintain a strong presence at best price limits and control adverse selection at the same time, see \cite{jovanovic2016middlemen}, while operating efficient inventory management in an increasingly fast-moving market, see \cite{ait2017high, baron2018risk}. This is to the extent that HFTs are described as the new market makers in \cite{menkveld2013high}.\\

\noindent Since the arrival of these new market makers, academics, regulators and practitioners aim at understanding whether their activity is harmful or beneficial for markets. On the one hand, some argue that HFTs have a positive impact on markets: the competition between market makers leads to an increase in market depth, to narrower bid-ask spreads which is equivalent to reduced trading costs for other investors, see \cite{hendershott2011does, jovanovic2016middlemen} and to better price discovery,
see \cite{hendershott2011does, riordan2012latency}.  On the other hand, others assert that high-frequency market makers have toxic consequences. For example, they worsen market volatility during flash crashes by aggressively liquidating their long positions, see \cite{kirilenko2017flash, madhavan2012exchange}.\\

\noindent One important common point in most studies analysing the behaviour of HFTs is that they try to measure how HFTs impact the market as a group, without investigating individual behavioural disparities among them. The authors in \cite{megarbane2017behavior,saliba2019information} shed light on the fact that all HFTs do not behave similarly, showing for example that they have very different levels of aggressiveness and liquidity provision. In this paper, we wish to participate to the debate about the role of HFTs on market quality by bringing some new quantitative elements enabling regulators and exchanges to assess the individual effects of each high-frequency market maker operating on the market. In particular, we want to be able to rank market makers according to the quality of their trading.\\

\noindent We use several metrics for market quality such as spread and liquidity fluctuations, but a particular focus is given to the price volatility. This idea of disentangling market participants contribution to volatility is used in \cite{rambaldi2018disentangling}. In this work, the authors nicely model the interactions between the various orders of the different market participants using linear Hawkes processes. This model is very interpretable: an order of type $A$ of Agent $i$ raises the likelihood of an order of type $B$ of Agent $j$ by a certain amount. Consequently, the authors naturally define the contribution of Agent $A$ to the volatility by the weighted sum over all possible types of orders of Agent $A$ of the squared mean price jump triggered by each of these orders, the associated weight being the intensity of the corresponding order type.\\

\noindent Our focus here is on market makers. Thus one crucial element to take into account is the well-known fact that the main market driver of any market making strategy is the state of the limit order book (and not single individual orders of other market participants), see \cite{huang2015simulating,lehalle2017limit,moallemi2016model}. Therefore, in the spirit of the Queue-reactive model of \cite{huang2015simulating}, we assume that the state of the order book, which is a common component, affects the interactions between our high-frequency market participants. However, to get a really accurate modelling of the behaviour of the agents, we also let their individual actions depend on their own past ones and on those of other participants, in the spirit of \cite{rambaldi2018disentangling}. We allow for strong non-linearities in the dependences with the past, leading to a much generalised version of Hawkes-Queue-reactive type order book models, see \cite{morariu2018state,wu2019queue}.\\

\noindent In this extended and non-Markovian framework, we are able to prove the ergodicity and diffusivity of our system, see \cite{huang2017ergodicity} for inspiring ideas. Furthermore, we provide asymptotic expressions for market quantities such as spread, liquidity fluctuations or price volatility in terms of the individual order flows of market participants. This notably enables us to forecast the dynamics of the market in case one market makers leaves it. The idea is that we consider that market makers interact with the market through their algorithms which are specified for example in term of average event size or in term of relative quantities such as the imbalance. If we remove one market participant while the others do not modify their algorithms, we can for instance compute a new volatility. If it is larger (smaller) than the actual one, we can say that the considered market maker has a stabilising (destabilising) effect on the market. This eventually leads us to a ranking of market makers with respect to the quality of their trading.\\

\noindent Let us now give a brief description of our model. Let $n$ be a positive integer representing the index of the $n$-th order book event $e_n$. Each event $e_n$ happens at time $T_n$ and is characterised by a variable $X_n$ that encodes all the needed information to describe $e_n$. For example, $X_n$ contains the order size, the type of the order (limit order, liquidity consuming order such as market order or cancellation), the order posting price and the identity of the agent. A detailed description of the sequences $(T_n)_{n \geq 1}$ and $(X_n)_{n \geq 1}$ is given in Section \ref{subsec:lobdyn}. The order book state is modelled by the process  $U_n = \left(Q^{1}_n, Q^{2}_n, S_n\right)$ with $Q^{1}_n$ the available quantity at the best bid, $Q^{2}_n$ the available quantity at the best ask and $S_n$ the spread at time $T_n$. For a detailed description of the dynamic of $U_n$, see Equation (\ref{eq: dyna P Q}). Here we focus on the first limits to reduce the dimension of the state space and keep a tractable model\footnote{However, we can model deeper limits by enlarging the dimension of the state space.}. Finally, we use a general approach to infer the behaviour of the price process from that of $(U_n)$, in the spirit of \cite{huang2017ergodicity,lehalle2018optimal}, see Section \ref{sec:limitTh} for the detailed formulation. We define the non-linear Hawkes-Markovian arrival rate $\lambda_t(e)$ of an order book event $e$ ($e$ containing the identity of the involved agent) at time $t \in \mathbb{R}_+$  as follows:
\begin{align*}
\lambda_t(e) = \psi\big(e,U_{t^-},t, \sum_{T_i < t} \phi(e,U_{t^-},t-T_i,X_i) \big),
\end{align*}
where $\psi$ is a non-linear function, $U_{t^-}$ is the order book state relative to the last event before $t$ and $\phi$ is the Hawkes kernel representing the influence of past events. The functions $\phi$ and $\psi$ are both $\mathbb{R}_+$-valued. In absence of the kernel $\phi$, the function $\psi$ leads to a classical Markovian approach since the arrival rate of an event $e$ depends essentially on the order book state $U_{t^-}$. When $\phi$ is non-zero, $\psi$ controls the interactions between the past events and the current order book state. Note that we allow $\psi$ to have a polynomial growth while in the literature, it has at most a linear growth, see \cite{bremaud1996stability}. Additionally, we do not impose $\psi$ and $\phi$ to be continuous, which means that a sudden change of regime in the order book dynamic is also incorporated in our modelling. Finally, we propose an agent-based model since market participant identities are contained in the order book events $e$ through the variables $(X_i)_{i \geq 1}$.\\ 

\noindent Our framework is a generalised order book model where the arrival rate of the events follows a non-linear Hawkes-type dynamic that depends on the order book state. This approach covers most existing bid-ask order book models. It is a natural extension of the Poisson intensity models, see \cite{abergel2013mathematical,smith2003statistical}, the Markovian Queue-reactive model introduced in \cite{huang2015simulating} and the Hawkes based models such as  \cite{LOBModHawkes,morariu2018state,rambaldi2018disentangling}. In this setting, under mild assumptions, we provide new ergodic results and limit theorems, expressing all the limiting quantities in terms of the individual flows of market participants. Furthermore, we build an estimation methodology for the intensity functions which turns out to be similar to the one used in the Queue-reactive case, see \cite{huang2017ergodicity}, although the model here is much more general and non-Markovian. These theoretical results for our point processes, which largely extend classical ergodicity properties limited to the Markov case, are the basis for the assessment of the role of the different market participants on market quality as explained above.\\
 
\noindent The paper is organised as follows. First, we introduce in Section \ref{sec:Mktmodel} our order book model and describe how to recover market dynamics from the individual behaviours of each agent. Then, we prove the ergodicity of our system in Section \ref{sec:Ergodicity_father} and its diffusivity in Section \ref{sec:limitTh}. In Section \ref{sec:Formulas}, we provide the needed formulas to compute the order book stationary distribution, the price volatility and the liquidity fluctuations. Finally, numerical results and ranking of market makers on several assets are provided in Section \ref{sec:Application}. Proofs and additional results are relegated to an appendix. 

\section{Market modelling}
\label{sec:Mktmodel}
In this section, we describe the order book model and show how to recover the market dynamics given the agents individual behaviours.
\subsection{Introduction to the  model}
In the order book mechanism buyers and sellers send their orders to a continuous-time double auction system. Market participants orders have a specific size that is measured in average event size (AES)\footnote{AES is the average size of events observed in the limit order book.} and the orders can be sent to different price levels that are separated by a minimum distance which is the tick size. In our model, we only consider the price levels between the best bid and ask prices to reduce the dimension of the state space. Additionally, we assume that the agents can take three elementary decisions:
\begin{itemize}
\item Insert a limit order of a specific size at the best bid or ask price, hoping to get an execution. 
\item Insert a buying or selling limit order of a specific size within the spread. 
\item Send a liquidity consuming order of a specific size at the best bid or ask price. Cancellation and market orders have the same effect on liquidity. Thus, they are aggregated to constitute the liquidity consumption orders.
\end{itemize} 
The size of the orders is not constant in the model. Finally, the mid price moves in a fixed grid separated by the tick. A simple example is to consider the case where the mid price decreases (resp. increases) by one tick when the best bid (resp. ask) is totally depleted. Here, the mid price jumps size may be larger than one tick. In the rest of the article, we take the mid price as our reference price for simplification. The dynamic of the model is illustrated in Figure \ref{fig:lob:flows}.

\begin{figure}[ht!]
\begin{center}
\begin{tikzpicture}[scale = 0.85]
\draw [double][->](2.5,0) -- ++(0,-0.5);
\draw [double][->](2.5,-5.5) -- ++(0,-0.5);
\draw [double][->](7.5,-1) -- ++(0,-0.5);
\draw [double][->](7.5,-5.5) -- ++(0,-0.5);
\draw [double][->](5,-4) -- ++(0,-0.5);
\draw [black,fill=blue!60](2,-1) rectangle (3,-5);
\draw [black,fill=red!60](7,-2) rectangle (8,-5);
\draw [dotted][->](0,-5) -- ++(10,0);
\draw (5,-5) node[sloped]{$|$};
\draw (2.5,-5) node[below]{\small $ Bid $} ;
\draw (7.5,-5) node[below]{\small $ Ask $} ;
\draw (2,-3.15) node[left]{$Q^{1}_t$} ;
\draw (8,-4) node[right]{$Q^{2}_t$} ;
\footnotesize
\draw (5,-5) node[below]{$ P_t $} ;
\normalsize
\draw (10,-5) node[right]{$Price$} ;
\draw (2.5,-0.25) node[right]{$i^1$} ;
\draw (2.5,-5.75) node[right]{$c^1$} ;
\draw (7.5,-1.25) node[right]{$i^2$} ;
\draw (7.5,-5.75) node[right]{$c^2$} ;
\draw (5,-3.5) node{$i^{1(2)}_{\frac{1}{2}}$} ;
\end{tikzpicture}
\end{center}  
  \caption{Diagram of flows affecting our order book model. The quantity $i^1$ (resp. $i^2$) represents the insertion of limit orders at the best bid (resp. ask). The quantity $i^1_{\frac{1}{2}}$ (resp. $i^2_{\frac{1}{2}}$) is associated to buying (resp. selling) limit orders within the spread. The quantities $c^1$ and $c^2$ refer to the orders that consume respectively the liquidity at the best bid and ask.}
  \label{fig:lob:flows}
\end{figure}
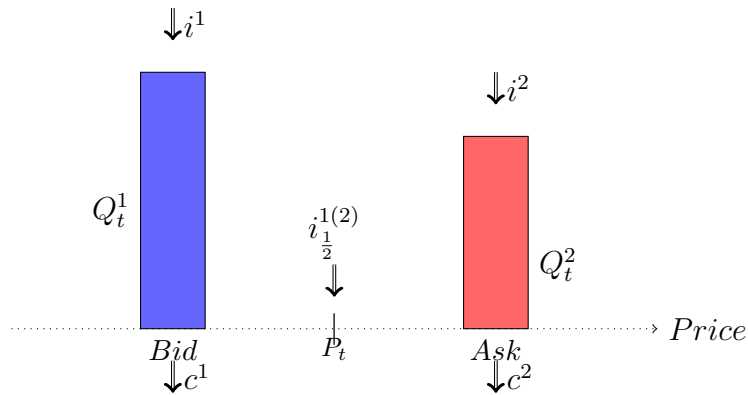
\paragraph{Notations.} We consider the following notations:
\begin{itemize}
\item The current physical time is $t$.
\item The mid price is $P_t$, the best best bid price is $P^1_t$ and the best ask price is $P^2_t$.
\item The spread is $S_t = \cfrac{P^2_t - P^1_t}{2}$ and $\alpha_0$ is the tick size.
\item The available quantity at the best bid (resp. ask) is $Q^1_t$ (resp. $Q^2_t$).
\end{itemize}
\subsection{Order book dynamic}
\manuallabel{subsec:lobdyn}{2.2}
Let $\left(\Omega, \mathcal{F} \right)$ be a measurable space and $(T_n)_{n \geq 1}$ a non-decreasing sequence of random variables such that $T_n < T_{n+1} $ on the event $ \{T_n < \infty \}$. We associate to each $T_n$ a random variable $X_n$ taking its value on a measurable space $(E,\mathcal{E})$. In our case, $T_n$ are the times when events happen in the order book and $X_n$ are variables describing each event. We endow $\Omega$ with the filtration $(\mathcal{F}_t)_{t \geq 0} $ defined such that $\mathcal{F}_{t} =  \sigma ( \{T_n \in C\} \times \{X_n \in B\}, \, C \in \mathcal{B}(\mathbb{R}) \cap (-\infty,t],\, B \in \mathcal{E})$. Each event is characterised by:
\begin{itemize}
\item \textbf{The size of the order:} is an integer representing the minimum quantity that can be inserted in the order book\footnote{In practice, the minimum quantity can be taken as a quarter of the the average event size (AES).}.
\item \textbf{The price of the order:} equals to $k \in \mathbb{N}$ when the order is inserted at the price $P^1 + k \alpha_0 $. 
\item \textbf{The direction of the order:} equals to $ +1$ if it provides liquidity and $-1$ when liquidity is removed.
\item \textbf{The type of the order:} equals to $1$ (resp. $2$) when it modifies the bid (resp. ask)\footnote{A buy (sell) limit order within the spread, a liquidity consumption at the bid (resp. ask) or a limit order at the bid (resp. ask) modify the bid side first.} side.
\item \textbf{The identity of the agent:} is valued in $\{1,\ldots,N\}$ since the market consists in N agents. 
\end{itemize}
Since we track only the first limits, we add the following variables to describe the new order book state when one of these limits is depleted: $\tilde{Q}^{1}$ (resp. $\tilde{Q}^2$) the new bid (resp. ask) queue and $\tilde{S}$ the new spread after a depletion. Note that when there is no depletion, the random vector $(\tilde{Q}^1,\tilde{Q}^2,\tilde{S})$ is arbitrary\footnote{To fix the ideas we can take $(\tilde{Q}^1,\tilde{Q}^2,\tilde{S}) = c $ with $c$ a fixed constant when there is no depletion.} and its values are not used. Finally, we record the order book state after an event to add a dependence between the arrival rate of the events and the past order book states. The order book dynamic is described below. Hence, we consider the following form for $E = \bar{\mathbb{N}} \times \mathbb{T} \times \mathbb{S} \times \mathbb{B} \times \tilde{\mathbb{U}} \times \mathbb{U} \times \mathbb{A}$ with:
\begin{itemize}
\item $\bar{\mathbb{N}} = \mathbb{N}^*$: the set where the orders size is valued.
\item $\mathbb{T} = \mathbb{N}$: the set where the price levels are valued.
\item $\mathbb{S} = \{+1,-1\} $: the set where the orders direction is valued.
\item $\mathbb{B} = \{1,2\}$: the set where the orders type is valued.
\item $\tilde{\mathbb{U}} = \left\{\mathbb{N}^2 \times \alpha_0 \mathbb{N}\right\} \setminus \mathbb{U}^0 $: the set where the order book states after a depletion are valued \footnote{The state where the best bid or ask size is zero is fictitious state that allow us to model the price changes, see Remark \ref{rem:Deterministic_f_0}.}.
\item $\mathbb{U} = \left\{\mathbb{N}^2 \times \alpha_0 \mathbb{N}\right\} \setminus \mathbb{U}^{0}  $: the set where the order book states after an event are valued.
\item $\mathbb{A} = \{1,\ldots,N\}$: the set where the agents identity is valued.
\item $\mathbb{U}^{0} = \left\{0\right\}^2 \times \alpha_0 \mathbb{N} $: the set of unreachable order book states.
\end{itemize} 
\begin{example} We place ourselves in the case where the minimum order size is a quarter of the AES and $(\tilde{Q}^1,\tilde{Q}^2,\tilde{S}) = c$ when there is no depletion with $c$ is a fixed constant. Thus, a buy limit order of size $0.5$ AES inserted at the best bid price $+ 1$ tick by the agent $5$ when the best bid size is $Q^1_i = 1$ AES, the best ask size is $Q^2_i = 3$ AES and the spread $S = 2$ ticks is represented by the event $e = (2,1,+1,1,c,u,5)$ with $u = (2,12,1)$.
\end{example}
\paragraph{Order book dynamic.} The order book state is modelled by the process  $U_t = \left(Q^{1}_t, Q^{2}_t, S_t \right)$  where $Q^{1}_t$ (resp. $Q^{2}_t$) is the best bid (resp. ask) quantity and $S_t$ is the spread. The dynamic of the reference price is going to be deduced from the one of the process $(U_t)_{t \geq 0}$, see Section \ref{sec:limitTh}. The process $U_t$ is defined in the following way: 
\begin{align*}
U_t = \sum_{T_{i} < t } \Delta U_{i}, \quad \Delta U_i = U_i - U_{i-1},
\end{align*}
with $U_{i} =(Q^{1}_{i},Q^{2}_{i},S_{i}) \in \mathbb{U}$ the order book state after the $i$-th event (we write $U_i$ for $U_{T_i}$ when no confusion is possible). Thus, we only need to describe the variables $(U_{i})_{i \geq 1}$. Let $i \geq 1$ and  $X_i = (n_i,t_i,s_i,b_i,\tilde{U}_i ,U_{i},a_i) \in E $ with $n_i \in \bar{\mathbb{N}}$, $t_i \in \mathbb{T} $, $s_i \in \mathbb{S} $, $b_i \in \mathbb{B} $, $\tilde{U}_i = (\tilde{Q}^{1}_{i},\tilde{Q}^{2}_{i},\tilde{S}_{i}) \in \mathbb{U} $, $U_{i} = (Q^{1}_{i},Q^{2}_{i},S_{i})\in \mathbb{U}$ and $a_i \in \mathbb{A}$. The variable $U_{i}$ satisfies 
	\begin{align}\label{eq: dyna P Q}
	\begin{array}{rl}
			S_{i}   = & \mathbf{1}_{\epsilon_i = 0 }S_{i-1} -  (t^1_i + t^2_i) + \mathbf{1}_{\epsilon_i = 1}\tilde{S}_{i}, \\
			Q^1_{i} = & \mathbf{1}_{\epsilon_i = 0 }Q^1_{i-1} +  (n^{1,+}_i - n^{1,-}_i + n^{1,\frac{1}{2}}_i) + \mathbf{1}_{\epsilon_i = 1} \tilde{Q}^{1}_{i},\\
			Q^2_{i} = &  \mathbf{1}_{\epsilon_i = 0 }Q^2_{i-1} + (n^{2,+}_i - n^{2,-}_i + n^{2,\frac{1}{2}}_i) + \mathbf{1}_{\epsilon_i = 1}\tilde{Q}^{2}_{i},
	\end{array}
	\end{align}
where $\epsilon_i$ is a price move indicator (i.e $ \epsilon =  0$ when there is no depletion and $\epsilon = 1$ otherwise), the variable $t^1_i$ (resp. $t^2_i$) is the spread variation when a buy (resp. sell) limit order is inserted within the spread. The variables $n^{1,+}_i$ (resp. $n^{2,+}_i$), $n^{1,-}_i$ (resp. $n^{2,+}_i$) and $n^{1,\frac{1}{2}}_i$ (resp. $n^{2,\frac{1}{2}}_i$) are respectively the best bid (resp. ask) increments when a buy limit order is inserted at the best bid (resp. ask), when a consumption order is sent at the best bid (resp. ask) and when a buy (resp. sell) limit order is inserted within the spread. We now explain how the previous quantities can be written in terms of the state variables:
\begin{equation*}
\begin{array}{lcl}
\epsilon_i &= &\mathbf{1}_{\{s_i = -1\}\cap \big(\{b_i = 1, n_i\geq Q^1_{i-1}\}\cup\{b_i = 2, n_i\geq Q^2_{i-1}\} \big)}, \\
t^1_i & = & \min(t_i \alpha_0, S_{i-1}-\alpha_0)\mathbf{1}_{\{b_i = 1,\,t_i \ne 0\}}, \\
t^2_i & = & (S_{i-1} - t_i \alpha_0)_{+}\mathbf{1}_{\{b_i = 2,\,t_i \ne \frac{S_{i-1}}{\alpha_0}\}}, \\
n^{1(2),+}_i & = & n_i\mathbf{1}_{\{s_i = +1,t_i = 0(\frac{S_{i-1}}{\alpha_0}),b_i = 1(2)\}}, \\
n^{1(2),-}_i & = &  n_i\mathbf{1}_{\{s_i = -1,t_i = 0(\frac{S_{i-1}}{\alpha_0}),b_i = 1(2), n_i < Q^{1(2)}_{i-1}\}}, \\
n^{1(2),1/2}_i & = & n_i\mathbf{1}_{\{s_i = +1,t_i \notin \{ 0,\frac{S_{i-1}}{\alpha_0}\},b_i = 1(2)\}}.
\end{array}
\end{equation*}
We denote by $\lambda_t$ the intensity of the point process $(T_n,X_n)$. For $e \in E $, $\lambda_t(e)$ corresponds to the arrival rate of an event of type $e$ conditional on the past history of the process and it is defined as
$$
\lambda_t(e) = \underset{\delta t \rightarrow 0}{\lim} \cfrac{\mathbb{P} \big[ \#\{T_n \in (t,t+\delta t] , X_n = e \} \geq 1 |  \mathcal{F}_t\big]}{\delta t },
$$
with $\# A$ is the cardinality of the set $A$. We consider the following expression for the intensity:
\begin{equation}
\lambda_t(e) = \psi\big(e,U_{t^-},t, \sum_{T_i < t} \phi(e,U_{t^-},t-T_i,X_i) \big),
\label{Eq:IntensityExpressionHawkes}
\end{equation}
where $\psi$ and $\phi$ are $\mathbb{R}_{+}$-valued functions.
 The individual behaviour of each agent is encoded in the functions $\psi$ and $\phi$ through $e$ and $(X_i)_{i \geq 1}$, see Equation (\ref{Eq:IntensityExpressionHawkes}).\\
 
Note that we can recover the full definition of the intensity of the process $N = (T_n,X_n)$ using the following proposition:
\begin{prop} For any $B \in \mathcal{E}$ and $t \in \mathbb{R}_+$, we have
\begin{equation}
\begin{array}{l}
\underset{\delta t \rightarrow 0}{\lim} \cfrac{\mathbb{P} \big[ \#\{T_n \in (t,t+\delta t] , X_n \in B \} \geq 1 |  \mathcal{F}_t\big]}{\delta t } =  \sum_{e \in B}  \lambda_t(e).
\end{array}
\label{Eq:IntensityExpressionAll}
\end{equation}
\label{Prop1:IntensityExpressionAll}
\end{prop}
The proof of Proposition \ref{Prop1:IntensityExpressionAll} is given in Appendix \ref{sec:MktReconstitution}. The existence and the uniqueness of a probability measure $\mathbb{P}$ on the filtered probability space $(\Omega, \mathcal{F}, \mathcal{F}_t)$ such that (\ref{Eq:IntensityExpressionAll}) is satisfied and $\lambda_t$ verifies Equation (\ref{Eq:IntensityExpressionHawkes}) is ensured as soon as $\sum_{e \in E} \lambda_t(e)$ is locally integrable, see \cite{jacod1975multivariate}. We prove that $\sum_{e \in E}\lambda_t(e)$ is locally integrable in Appendix \ref{sec:ErgValFctHks}.
\subsection{Market reconstitution}
\manuallabel{subsec:MarketReconstit}{2.3}
We can recover the market intensity $\lambda^M_t$ using the corollary below.
\begin{cor} When $\lambda_t$ verifies Equation (\ref{Eq:IntensityExpressionHawkes}), the market intensity $\lambda^M_t(e')$ of an event $e'$ ($e'$ does not contain the identity of the agent) in the exchange is given by  
\begin{equation}
\lambda^M_t (e') = \underset{\delta t \rightarrow 0}{\lim} \cfrac{\mathbb{P} \big[ \#\{T_n \in (t,t+\delta t] , X_n \in (e',\mathbb{A}) \} \geq 1 |  \mathcal{F}_t\big]}{\delta t } = \sum_{a \in \mathbb{A}} \lambda_t((e',a)),
\label{Eq:ReconstituteMkt}
\end{equation}
for any $e' \in E' = \bar{\mathbb{N}} \times \mathbb{T}  \times  \mathbb{S} \times \mathbb{B} \times \tilde{\mathbb{U}} \times \mathbb{U}$.
\label{Prop1:MktReconstit}
\end{cor}
The proof of Corollary \ref{Prop1:IntensityExpressionAll} is a consequence of Proposition \ref{Prop1:IntensityExpressionAll}.
\subsection{Some specific models}
\paragraph{Poisson intensity.} We introduce here a simple version of the Poisson intensity model where 
the variable $X_{n} = (n_n,t^o_n,s_n,b_n,\tilde{U}_n, U_n, a_n)$ with $U_n = (Q^{1}_{n},Q^{2}_{n},S_{n})$ satisfies 
\begin{itemize}
\item the order size \textbf{$n_n = 1$}: all the events have the same size $1$ AES.
\item the price level \textbf{$t^o_n \in \{0,\frac{S_n}{\alpha_0}\}$}: orders are inserted at the best bid or ask.
\item the law of $\tilde{U}_n$ is unchanged: when one limit is depleted, the new state is drawn from the stationary distribution of the order book.
\end{itemize}
For any $e = (n,t^o,s,b,\tilde{u},u,a)\in E$ with $u = (Q^{1},Q^{2},S)$, we can recover Poisson models by taking the following choice of the parameters:
$$
\psi(e,u,t,z) = \tilde{h}(s,b,a)\mathbf{1}_{n =1, t^o \in \{0,\frac{S}{\alpha_0}\}},  \quad  \forall z,t \in \mathbb{R}_{+},
$$
with $\tilde{h}$ a deterministic function valued on $\mathbb{R}_+$. Thus, the expression of the intensity becomes
$$
\lambda_t(e) = \tilde{h}(s,b,a)\mathbf{1}_{n =1, t^o \in \{0,\frac{S}{\alpha_0}\}}.
$$
Such modelling was introduced in \cite{abergel2013mathematical,cont2013price,smith2003statistical}.
\paragraph{Queue-reactive intensity.} In the Queue-reactive model, the arrival rate of the events depends only on the current order book state. For any $e \in E$ and $u\in U$, we take 
$$
\psi(e,u,t,z) = \tilde{h}(e,u), \quad \forall z,t \in \mathbb{R}_{+},
$$
to reproduce the Queue-reactive dynamic with $\tilde{h}$ a deterministic function valued on $\mathbb{R}_+$. Hence, the intensity reads
$$
\lambda_t(e) = \tilde{h}(e,u).
$$
Such modelling was studied in \cite{huang2015simulating,huang2017ergodicity}.
\paragraph{Hawkes Queue-reactive intensity.} In the Hawkes framework, the arrival rate of each event depends fully on all the past market events. For any $e \in E$ and $u\in U$, we generate the Hawkes Queue-reactive dynamic by taking 
$$
\psi(e,u,t,z) = h(e,u,t)  + z,  \quad \forall z,t \in \mathbb{R}_{+}.
$$
Thus intensity has the following expression
$$
\lambda_t(e) =  h(e,U_{t^-},t) + \sum_{ T_i < t} \phi(e,U_{t^-},t-T_i,X_i).
$$
Close modelling was used \cite{LOBModHawkes,citeulike:13497373,jaisson2015market,morariu2018state,rambaldi2018disentangling}.
\paragraph{Quadratic Hawkes process.} The quadratic Hawkes processes generalise the linear Hawkes processes by adding an interaction term between the pairs of past events. In the classical one-dimensional case, the intensity function of a quadratic Hawkes process reads
$$
\lambda_t(e) =  h(t) + \sum_{ T_i < t} \phi(t-T_i) + \sum_{ T_i,T'_i < t}K(t-T_i,t-T'_i),
$$
with $K:\mathbb{R}_+ \times \mathbb{R}_+ \rightarrow \mathbb{R}_+$ the quadratic kernel. We can recover a simple case of the quadratic Hawkes models when $K$ is separable (i.e $K(t,s) = k(t)k(s)$ with $k$ a non negative function) by taking $\psi$ of the following form:
$$
\psi(e,u,t,z) = h(e,u,t) + z^2,  \quad \forall z,t \in \mathbb{R}_{+}.
$$
Hence, the expression of the intensity becomes
$$
\lambda_t(e) =  h(e,U_{t^-},t) +  \sum_{ T_i < t} \phi^2(e,U_{t^-},t-T_i,X_i) + \sum_{ T_i,T'_i < t} \phi(e,U_{t^-},t-T_i,X_i)\phi(e,U_{t^-},t-T'_i,X'_i).
$$
Quadratic Hawkes models were introduced in \cite{blanc2017quadratic,ogata1981lewis}.
\begin{rem}
In our modelling, the linear term is necessarily $\phi^2$. However, to overcome this limitation we can add a new argument to the function $\psi$ which differentiates the linear kernel from the quadratic one. This will not modify the proofs.
\end{rem} 
\section{Ergodicity}
\label{sec:Ergodicity_father}
\subsection{Notations and definitions}
Let $Z_t$ be a process defined on the probability space $(\Omega, \mathcal{F},\mathcal{F}_t,\mathbb{P})$ and valued in $(W_0,\mathcal{W}_0)$. We consider another process $V_t$ defined on $(W_0,\mathcal{W}_0)$ and valued in $(X,\mathcal{X})$ and we denote by $P_t(x,.)$ the probability distribution of $V^{0,x}_t$ starting at $0$ with the initial condition $x \in W_0$. For any measure $\mu$ defined on $(W_0, \mathcal{W}_0)$ viewed as a random starting condition, we denote by $P_t(\mu,.) = \int_{W_0}  P_t(x,.) \mu(dx)$. 
\begin{defi}[Invariant distribution] The measure $\mu$ is invariant if the probability distribution $P_{t}(\mu,.)$ does not depend on the time $t$. 
\label{defi:InvDistrib}
\end{defi}
This definition is consistent with the one given in \cite{bremaud1996stability,hairer2009ergodic,meyn2012markov}. The process $V_t$ starting with the initial distribution $\mu$ is stationary if and only if $\mu$ is invariant. We define the total variation norm between two measures $\pi$ and $\pi'$ such that $||\pi - \pi'||_{TV} = \sup_{A \in \mathcal{X}} | \pi(A) - \pi'(A)  |$.
\begin{defi}[Ergodicity] Let $C \in \mathcal{W}_0$. The process $V_t$ is $C$-ergodic if for any $x \in C$ there exists an invariant measure $\mu$ such that $P_t(x,.) \underset{t \rightarrow \infty}{\rightarrow} P_{0}(\mu,.)$ in total variation.
\label{defi:Ergod}
\end{defi}
\begin{rem}
This definition is consistent with the one given in \cite{meyn2012markov}. Ergodicity is interesting since it ensures the convergence of the order book process $U_t$ towards an invariant probability distribution. Thus the stylized facts observed on market data can be explained by a law of large numbers type phenomenon for this invariant distribution. 
\end{rem}
\begin{rem}
In this Section, we work with a continuous time processes $Z_t$ and $V_t$ with $t \in \mathbb{R}_+$. However, all the definitions are similar for a discrete time processes $Z_n$ and $V_n$ with $n \in \mathbb{N}$. We just have to replace $t$ by $n$ in the definitions above.
\end{rem}
The space $\Omega$ and the filtration $\mathcal{F}_t $ considered here are defined in Section \ref{subsec:lobdyn}, $\mathcal{F} = \mathcal{F}_{\infty}$, the filtered space $W_0$ is the space of sequences indexed by $\mathbb{N}^-$ and valued on $\mathbb{R}_{+}\times E$ , $ X = \mathbb{U} \times (\mathbb{R}_+)^{E}$ and $ \mathcal{X} = \mathcal{U} \times \mathcal{B}(\mathbb{R}_+)^{\otimes E}$ with $\mathcal{U}$ the $\sigma$-algebra generated by the discrete topology on $\mathbb{U}$, $\mathcal{B}(\mathbb{R}_+)^{\otimes E}$ the cylinder $\sigma$-algebra for $(\mathbb{R}_+)^{E}$, $\mathcal{B}(\mathbb{R}_+)$ the borel $\sigma$-algebra of $\mathbb{R}_+$ and $\mathcal{W}_0 = \big(\mathcal{B}(\mathbb{R}_+) \times \mathcal{E} \big)^{\otimes \mathbb{N}^-}$ with $\mathcal{E}$ the $\sigma$-algebra generated by the discrete topology on $E$. We need to work on the functional space $W_0$ since the dynamic of the process depend on its whole past.
\subsection{Ergodicity}
\label{sec:Ergodicity}
In this section, we provide under general assumptions a theoretical result on the ergodicity of the process $\bar{U}_t =(Q^{1}_t,Q^{2}_t,S_t, \lambda_t)$ with $\lambda_t$ the intensity defined by (\ref{Eq:IntensityExpressionHawkes}).\\

We denote by $\lambda^{i,+}_{Q}$ (resp. $\lambda^{i,-}_{Q}$) and $\lambda^{+}_{S}$  (resp. $\lambda^{-}_{S}$) the arrival rate of the events that respectively increase (resp. decrease) the limit $Q^i$ and the spread $S$ for any $i \in \mathbb{B} $. Let $U_t = (Q^1_t,Q^2_t,S_t)$ be the order book process and $ e  \in E$ be a market event, the quantities $\lambda^{i,\pm}_{Q}$ and $\lambda^{\pm}_{S}$ are defined by the following formulas:
\begin{equation}
\lambda^{i,\pm}_{Q}(U_{t^-},n) = \sum_{e \in E^{i,\pm}_{Q}(U_{t^-},n)}\lambda_t(e),  \quad \lambda^{\pm}_{S}(U_{t^-},k) = \sum_{e' \in E^{\pm}_{S}(U_{t^-},k)}\lambda_t(e),
\label{Eq:NumIncreaDecreaQtySpread}
\end{equation}
with $n \in \mathbb{N}$, $k \in \mathbb{N}$ and 
\begin{equation}
\begin{array}{lcl}
E^{i,\pm}_{Q}(U_{t^-},n) & = & \{e \in E; \, s.t \quad \Delta Q^{i}_t = \pm n\}, \\
E^{\pm}_{S}(U_{t^-},k) & = & \{e \in E; \, s.t \quad \Delta S_t = \pm k\}, 
\end{array}
\label{Eq:SetEvtModifyLOB}
\end{equation}
with $\Delta X_t = X_t - X_{t^-}$ for any process $X_t$. For simplicity and since there is no ambiguity, we do not write the dependence of $\lambda^{i,\pm}_Q$ and $\lambda^{\pm}_S$ on the current time $t$. For any $n \in \mathbb{N}^*$, we write 
$$
\mathcal{P}(n) = \{ \mathbf{k}_m = \{k_1,\ldots,k_m\} \in (\mathbb{N}^*)^m; \quad s.t \quad k_1 + \ldots + k_m = n, \quad m \in \mathbb{N}^*\},
$$
for the set containing all the partitions of $n$.
\begin{Assumption}[$\psi$ growth] We assume that there exist $c \geq 0$, $d \geq 0$ and $n_{\psi}\in \mathbb{N}$ such that
\begin{equation*}
\begin{array}{l}
\tilde{\psi}(e,z) \leq c(e) + d(e) z^{n_{\psi}},\\
 \sup_{e \in E}\left\{d(e)\sum_{\mathbf{k}_m \in \mathcal{P}(n_{\psi})} {{n_{\psi}}\choose{\mathbf{k}_m}} \int_{\mathbb{R}_{+}^m} \prod_{i=1}^m {\phi^*}^{k_i}(e,s_i)\,ds_i \right\} < 1,
\end{array}
\label{Eq:BoundednessHksintensity1}
\end{equation*}
with $\tilde{\psi}(e,z) = \sup_{(u,t) \in \mathbb{U} \times \mathbb{R}_+} \psi(e,u,t,z)$, $\phi^*(e,s)  = \sup_{u \in \mathbb{U}}\sum_{x \in E}\phi(e,u,s,x)$ and $ {{n_{\psi}}\choose{\mathbf{k}_m}} = {{n_{\psi}}\choose{k_1,\ldots,k_m}} = \frac{n_{\psi}!}{k_1!\,\ldots\,k_m!} $.\\
\label{Assump:psigrowth}
\end{Assumption}
Assumption \ref{Assump:psigrowth} is natural. To see this, we take a 1-d stationary non-linear Hawkes process $N_t$ with an intensity $\lambda_t$ that verifies
\begin{align*}
\lambda_t = c + d (\sum_{T_i < t} \phi(t-T_i))^{n_{\psi}} = c + d (\int_{- \infty}^{t}\phi(t-s)dN_s)^{n_{\psi}}, \quad \forall t \in \mathbb{R}_+.
\end{align*}
By stationarity, we have 
\begin{align*}
\bar{\lambda} = \Esp[\lambda_t] & = c + d\Esp[(\int_{- \infty}^{t}\phi(t-s)dN_s)^{n_{\psi}}]\\
			    & = c + d \left\{\sum_{\mathbf{k}_m \in \mathcal{P}(n_{\psi})} {{n_{\psi}}\choose{\mathbf{k}_m}} \int_{(-\infty,t)^m} \prod_{i=1}^m {\phi}^{k_i}(t - s_i)\, \Esp[dN_{s_1} \ldots  dN_{s_m} ]\right\},
\end{align*}
with ${{n_{\psi}}\choose{\mathbf{k}_m}}$ an enumeration factor. In fact, if we have $n_{\psi}$ possible events divided in $m$ groups such that the $j$-th group is composed of $k_j$ events, then the quantity ${n_{\psi}}\choose{\mathbf{k}_m}$ counts the number of possible groups. Here each group represents the jumps that happen at the same time. Since the jumps have a unit size, the Brascamp-Lieb inequality ensures that $\Esp[dN_{s_1} \ldots  dN_{s_m} ] \leq \prod_{i=1}^m \Esp[dN_{s_i}^m]^{1/m} = \prod_{i=1}^m \Esp[dN_{s_i}]^{1/m} = \prod_{i=1}^m \Esp[\lambda_{s_i}]^{1/m} = \bar{\lambda}$ which leads to
\begin{align*}
\bar{\lambda} & \leq  c + q \bar{\lambda},
\end{align*}
with $
q = d \sum_{\mathbf{k}_m \in \mathcal{P}(n_{\psi})} {{n_{\psi}}\choose{\mathbf{k}_m}} \int_{(\mathbb{R}_+)^m} \prod_{i=1}^m {\phi}^{k_i}(e,s_i)\, ds_i$. The condition $q < 1$ of Assumption \ref{Assump:psigrowth} guarantees that $\bar{\lambda}$ is finite. 
\begin{rem}	Non linear Hawkes process are studied mainly when the function $\psi$  admits at most a linear growth (i.e $n_{\psi} \leq 1$). When $n_{\psi} = 1$, we recover the classical condition 
$$
 \sup_{e \in E} d(e)\left\{\int_{\mathbb{R}_{+}} \phi^*(e,s)ds \right\} < 1.
$$
When $n_{\psi} = 2$, Assumption \ref{Assump:psigrowth} becomes 
$$
\sup_{e \in E} \, d(e) \left\{ \big(\int_{\mathbb{R}_{+}} \phi^*(e,s)ds\big)^2 + \int_{\mathbb{R}_{+}} {\phi^*(e,s)}^2 ds \right\} < 1
.$$
\end{rem}
\begin{Assumption}[Negative drift] There exist positive constants $C_{bound}$, $z_{0} > 1$ and $\delta$ such that 
\begin{equation}
\begin{array}{lcr}
\sum_{n \geq 0} (z_0^{n} -1) \big(\lambda^{i,+}_{Q}(U_{t^{-}},n) - \lambda^{i,-}_{Q}(U_{t^{-}},n)\frac{1}{z_0^{n}}\big) \leq -\delta, \quad & a.s & \text{when } Q^{i}_{t^{-}} \geq C_{bound},\\
\sum_{k \geq 0} (z_0^{\alpha_0 k} -1)\big(\lambda^{+}_{S}(U_{t^{-}},k) - \lambda^{-}_{S}(U_{t^{-}},k)\frac{1}{z_0^{ \alpha_0 k}}\big) \leq -\delta,  \quad & a.s & \text{when } S_{t^{-}} \geq C_{bound},
\end{array}
\label{Eq:NegDriftQty}
\end{equation}
for any $ i\in \mathbb{B}$ and $U_{t} = (Q^{1}_{t},Q^{2}_{t},S_{t}) \in \mathbb{U}$ where $\alpha_0$ is the tick size.
\label{Assump:Negativeindividualdrift}
\end{Assumption}
Assumption \ref{Assump:Negativeindividualdrift} ensures that both the size of the first limits and the spread tend to decrease when they become too large. Same kind of hypothesis are used in \cite{huang2015simulating,lehalle2018optimal} but when the order book dynamic is Markov.
\begin{rem} In practice, Assumption \ref{Assump:Negativeindividualdrift} is verified when the following conditions are satisfied:
\begin{equation}
\begin{array}{lr}
\sum_{n \geq 0} (z_0^{n} -1)\big(\psi^{i,+}_{Q}(u,n,t,z) - \psi^{i,-}_{Q}(u,n,t,z)\frac{1}{z_0^{n}}) \leq -\delta, & \text{when } q^{i} \geq C_{bound},\\
\sum_{n \geq 0} (z_0^{\alpha_0 k} - 1) \big(\psi^{+}_{S}(u,k,t,z) - \psi^{-}_{S}(u,k,t,z)\frac{1}{z_0^{\alpha_0 k}}) \leq -\delta, & \text{when } s^{i} \geq C_{bound},\\
\phi^{i,+}_{Q}(u,n,t,x) \leq \phi^{i,-}_{Q}(u,n,t,x)  , & \text{when } q^{i} \geq C_{bound},\\
\phi^{+}_{S}(u,k,t,x) \leq \phi^{-}_{S}(u,k,t,x)  , & \text{when } s^{i} \geq C_{bound},\\
\psi(e,u,t,z), \quad  \text{ is non-decreasing in } z, & \text{when } q^{i} \geq C_{bound},\\
\psi(e,u,t,z), \quad \text{ is non-decreasing in } z, & \text{when } s^{i} \geq C_{bound},\\
\end{array}
\label{Eq:NegDriftQtyExample}
\end{equation}
where $u = (q^1,q^2,s) \in \mathbb{U} $, $i \in \mathbb{B} $ and $\psi^{i,\pm}_{Q}$, $\psi^{\pm}_{S}$, $\phi^{i,\pm}_{Q}$ and $\phi^{\pm}_{S}$ are functions defined such that
$$
\begin{array}{ll}
\psi^{i,\pm}_{Q}(u,n,t,z) = \sum_{e \in E^{i,\pm}_{Q}(u,n)} \psi(e,u,t,z), & \phi^{i,+}_{Q(S)}(u,n,t,x) = \sup_{e \in E^{i,+}_{Q(S)}(u,n)} \phi(e,u,t,x),\\
\psi^{\pm}_{S}(u,k,t,z) = \sum_{e \in E^{\pm}_{S}(u,k)} \psi(e,u,t,z), & \phi^{-}_{Q(S)}(u,k,t,x) = \inf_{e \in E^{-}_{Q(S)}(u,k)} \phi(e,u,t,x),
\end{array}
$$
with $(n,k,t,z) \in \mathbb{N}^2 \times \mathbb{R}_+^2$. Although Inequalities (\ref{Eq:NegDriftQty}) and (\ref{Eq:NegDriftQtyExample}) are not equivalent, there is a large panel of functions that satisfy (\ref{Eq:NegDriftQtyExample}). A proof of this result is given Appendix \ref{sec:ProofRemNegativeCondDrift}.
\label{Rem:NegativeCondDrift}
\end{rem}
\begin{Assumption}[Bound on the overall flow] We assume that there exist $z_1 > 1$, $M$ and $ \underline{\psi}> 0$ satisfying
\begin{equation*}
\begin{array}{lcll}
c^* & = & \sum_{e \in E} c(e) < \infty,\\
\lambda^* & = & \sum_{e \in E, \, \mathbf{k}_m \in \mathcal{P}(n_{\psi})} d(e){{n_{\psi}}\choose{\mathbf{k}_m}} \int_{\mathbb{R}_+^m} \prod_{j=1}^m {\phi^*}^{k_j}(e,s_j)\,ds_j < \infty,\\
Q^i_{\infty} & = & \sum_{n \in\mathbb{N}} (z_1^n-1)\Esp_{\mathbf{x}}\big[\lambda^{i,+}_{Q}(u,n) - \frac{\lambda^{i,-}_{Q}(u,n)}{z_1^n}\big]  < M, & \text{ when } q^i \leq C_{bound},\\
S_{\infty} & = & \sum_{k \in\mathbb{N}} (z_1^k-1)\Esp_{\mathbf{x}}\big[\lambda^{+}_{S}(u) - \frac{\lambda^{-}_{S}(u,n)}{z_1^k}\big]  < M, & \text{ when } s \leq C_{bound},\\
\underline{\lambda}_t(e) & = & \sum_{e \in E} \lambda_t(e) \geq  \underline{\psi}, & \, a.s.
\end{array}
\label{Eq:BoundednessHksintensity}
\end{equation*}
with $c(e)$, $d(e)$ and $\phi^*$ defined in Assumption \ref{Assump:psigrowth}, $ i \in \mathbb{B}$, $\mathbf{x} \in W_0$ and $C_{bound}$ defined in Assumption (\ref{Assump:Negativeindividualdrift}). Similar assumptions are considered in \cite{huang2015simulating,lehalle2018optimal} in the Markov case.
\label{Assump:Boundedness}
\end{Assumption}
Assumption \ref{Assump:Boundedness} ensures no explosion in the system since it forces the arrival rate of orders, the size of the limits and the spread to stay bounded.
\begin{rem} In practice, we can find path-wise conditions similar to those used in Remark \ref{Rem:NegativeCondDrift} such that the inequalities $Q^i_{\infty} < M $, $S_{\infty} < M $ and $\underline{\lambda}_t(e) \geq \bar{\psi},\, a.s$ are satisfied.
\label{Rem:Boundedness}
\end{rem}
\begin{theo}[Existence] Under Assumptions \ref{Assump:psigrowth}, \ref{Assump:Negativeindividualdrift} and \ref{Assump:Boundedness}, the process $\bar{U}_t = (Q^1_t,Q^2_t,S_t,\lambda_t)$ admits an invariant distribution.
\label{lem:ExistenceProcHks}
\end{theo}
The proof of this result is given in Appendix \ref{sec:ErgValFctHks}.
\begin{Assumption}[Regularity] We assume that $\psi$ is a c\`{a}dl\`{a}g function continuous with respect to $z$, $\phi$ is a positive c\`{a}dl\`{a}g function and 
there exist $\bar{\psi}: \mathbb{R}_{+} \rightarrow \mathbb{R}_+$ and $n_1 \in \mathbb{N}$ such that 
$$ 
|\psi(e,u,s,x) - \psi(e,u,s,y)|  \leq |\bar{\psi}(x) - \bar{\psi}(y)| , \quad \forall (e,u,s,x,y) \in E \times \mathbb{U} \times \mathbb{R}_+^3,
$$
and
$$
|\bar{\psi}(x) - \bar{\psi}(y)| \leq K|x - y||1 + x^{n_1} + y^{n_1}|, \quad \forall (x,y) \in \mathbb{R}_+^2,
$$
with $K$ a positive constant.
\label{Assump:Regularityparams2}
\end{Assumption}
\begin{rem} Assumption \ref{Assump:Regularityparams2} is satisfied in the special case where $\bar{\psi}$ is a polynomial.
\end{rem}
We have the following result. 
\begin{theo}
[Ergodicity] Under Assumptions \ref{Assump:psigrowth}, \ref{Assump:Negativeindividualdrift}, \ref{Assump:Boundedness} and \ref{Assump:Regularityparams2}, the process $\bar{U}_t $ is $\mathcal{W}_0$-ergodic, which means that there exists an invariant measure $\mu$, see Definition \ref{defi:InvDistrib}, that satisfies
$$
\underset{t \rightarrow \infty}{\lim} P_{t}(x,A) =  P_{0}(\mu,A), \quad \forall x \in W_0,\,A \in \mathcal{X},
$$
where $P_{t}(x,A)$ is the probability that $\bar{U}_t \in A$ starting from the initial condition $x$. Additionally, we have the following speed of convergence:
$$
||P_t(x,.) - P_{0}(\mu,.)||_{TV} \leq K_1 e^{- K_2 t}, \quad \forall x \in W_0,
$$
with  $K_1$, $K_2$ are positive constants and $||.||_{TV}$ the total variation norm.
\label{lem:ErgodicityProcHks}
\end{theo}
The proof of this result is given in Appendix \ref{sec:ErgValFctHks2}. We can construct pathwise the point process $N = (T_n,X_n)$ defined in Section \ref{sec:Mktmodel} using the following algorithm.
\begin{rem}[Pathwise construction of $N$]
Using the thinning algorithm proposed by Lewis in \cite{lewis1979simulation} and Ogata in \cite{ogata1981lewis}, the point process $N = (T_n,X_n)$ defined in Section \ref{sec:Mktmodel} satisfies $N = \underset{m \rightarrow \infty}{\lim} N^m $ where $N^m $ is defined as follows
\begin{equation*}
\begin{array}{lcl}
\lambda^{m+1}_t(e) & = & \psi \big(e,U^m_{t^-},t, \sum_{T^m < t} \phi(e,U^m_{t^-},t- T^{m},X^m) \big)\mathbf{1}_{T^{m} \leq t < T^{m+1} } + \lambda^{m}_t(e)\mathbf{1}_{t < T^{m} }, \\
 N^{m+1}((0,t]\times B) & = & \int_{(T^{m},T^{m+1}]\times B} N^*(dt\times (0,\lambda^{m+1}_t(e)]\times de)\mathbf{1}_{t > T^{m}} + N^m((0,t \wedge T^{m}]\times B),  \\
 T^{m+1} & = & \sup\{ t > T^{m}; \quad \int_{(T^{m},t]\times \mathcal{E}} N^*(dt\times (0,\lambda^{m}_t(e)]\times de) = 0 \},
\end{array}
\label{Eq:IntensityNrecurrence2}
\end{equation*}
with $U^m$ the order book process generated by $N^m$ and described in (\ref{eq: dyna P Q}), $N^* = (T^*_n,R^*_n,X^*_n)$ a Poisson process valued on $\mathbb{R}_{+}^2 \times E$ which admits $ \mathrm{d}t \mathrm{d}z \nu(\mathrm{d}e)$ as an $\mathcal{F}_{t}^{N^*}$ intensity and $\nu = \sum_{e \in E} \delta_{e}$. 
\label{Rem:Rem10}
\end{rem}
This is a well known result that were used in many contexts, see \cite{bremaud1996stability,daley2007introduction,last1993dependent,lewis1979simulation,ogata1981lewis}. The proof of Theorem \ref{lem:ExistenceProcHks} ensures that the above algorithm is well defined.
\section{Limit theorems}
\manuallabel{sec:limitTh}{4}
Let $n$ be the index of the $n$-th jump, $(\eta_n)_{n \geq 0}$ be a process satisfying $\eta_n = f((U_i)_{i \leq n},(Y_i)_{i \leq n})$ with $f$ a measurable function valued on $(\mathbb{R},\mathcal{B}(\mathbb{R}))$, $(Y_i)_{i \geq n}$ is a geometrically ergodic sequence, see 15.7 in \cite{meyn1992stability}, independent of $(U_i)_{i \geq n}$. Here, we write $\mu$ for the invariant measure of the joint process $(U,Y)$, $V_n = \sum_{k=1}^n \eta_k $ and $S_n = \sum_{k=1}^n (\eta_k - \Esp_{\mu} [\eta_k])$. We denote by
\begin{equation*}
X_n(t) = \cfrac{S_{\lfloor nt \rfloor}}{\sqrt{n}}, \qquad \forall t \geq 0.
\end{equation*}
\begin{Assumption} Under the invariant measure $\mu$, the sequence $(\eta_i)_{i\geq 0}$ is stationary and $\Esp_{\mu} [|\eta_0|] < \infty$.
\label{Assump:LimitThMrkv}
\end{Assumption}
\begin{Assumption} Under the invariant measure $\mu$, we have $ \Esp_{\mu} [(\eta_0 - \Esp_{\mu} [\eta_0] )^2] < 1$.
\label{Assump:LimitThMrkv2}
\end{Assumption}
\begin{prop} Under Assumption \ref{Assump:LimitThMrkv}, we have 
\begin{equation}
\frac{V_n}{n} \underset{n \rightarrow \infty }{\longrightarrow}  \Esp_{\mu} [\eta_0], \quad a.s.
\label{Eq:CvgResult1}
\end{equation}
Moreover when both Assumptions \ref{Assump:LimitThMrkv} and \ref{Assump:LimitThMrkv2} are verified, the quantity $X_n(t)$ satisfies
\begin{equation}
X_n(t)\overset{\mathcal{L}}{\longrightarrow}   \sigma W_t,
\label{Eq:CvgResult2}
\end{equation}
with $\sigma^2 = \Esp_{\mu} [\eta_0^2] + 2 \sum_{k \geq 1} \Esp_{\mu} [\eta_0\eta_k]$ and $\mu$ the invariant measure of $(U_i,Y_i)$ and $W_t$ a standard brownian motion. 
\label{Prop:LimitThMrkv}
\end{prop}
Note that $\sigma^2 < \infty $ under Assumption \ref{Assump:LimitThMrkv2}. The proof of this result is given in Appendix \ref{sec:Prooflth}.
\begin{rem} The leading term in the expression of $\sigma^2$ is $\Esp_{\mu} [\eta_0^2]$. Numerically, it can be computed as soon as we have an estimate of the stationary distribution of $\eta_0$, see Proposition \ref{Prop:GeneralProbStat}.
\end{rem}
Proposition \ref{Prop:LimitThMrkv} ensures that the large scale limit of $S$ in event time is a brownian motion. However, it is more relevant to study the large scale limit of the process $S$ in calendar time. Thus we now consider the process
\begin{equation*}
\tilde{X}_n(t) = \cfrac{S_{N(nt)}}{\sqrt{n}}, \qquad \forall t \geq 0.
\end{equation*}
The following proposition provides the large scale limit of the process $S_{N(nt)}$.
\begin{prop} Under Assumption \ref{Assump:LimitThMrkv}, we have 
\begin{equation}
\frac{V_{N(nt)}}{n} \underset{n \rightarrow \infty }{\longrightarrow}  \cfrac{\Esp_{\mu} [\eta_0]}{\Esp_{\mu}[\Delta T_1]}, \quad a.s.
\label{Eq:CvgResult1_1}
\end{equation}
Moreover when both Assumptions \ref{Assump:LimitThMrkv} and \ref{Assump:LimitThMrkv2} are verified, the quantity $\tilde{X}_n(t)$ satisfies
\begin{equation}
\tilde{X}_n(t)\overset{\mathcal{L}}{\longrightarrow}   \frac{\sigma}{\sqrt{\Esp_{\mu}[\Delta T_1]}} W_t,
\label{Eq:CvgResult2_1}
\end{equation}
with $\sigma^2 = \Esp_{\mu} [\eta_0^2] + 2 \sum_{k \geq 1} \Esp_{\mu} [\eta_0\eta_k]$, $\mu$ the invariant measure of $(U_i,Y_i)$, $\Delta T_n = T_n - T_{n-1}$ the inter-arrival time between the $n$-th and $(n-1)$-th jump and $W_t$ a standard brownian motion. 
\label{Prop:LimitThMrkv2_1}
\end{prop}
The proof of this result is given in Appendix \ref{sec:Prooflth}.
\begin{rem}
The mid price after $n$ jumps $P_n$ satisfies $P_n = P_0 + \sum_{i=1}^n \Delta P_i $ with $\Delta P_i = (P_{i} - P_{i-1}) = \eta_i$. When $(\eta_i)_{i \geq 0}$ verifies Assumptions \ref{Assump:LimitThMrkv} and \ref{Assump:LimitThMrkv2}, the rescaled price process $\tilde{P}_n(t) = \cfrac{P_{N( nt)}}{\sqrt{n}}$ converges towards a Brownian diffusion.
\label{Rem:PriceApply}
\end{rem}
\section{Formulas}
\label{sec:Formulas}
In this section, we provide a calibration methodology for the intensities and computation formulas for the quantities of interest: the stationary distribution of the order book, the price volatility and the fluctuations of liquidity.
\subsection{Stationary probability computation}
\manuallabel{subsec:Stationaryprob}{5.1}
In this section, we denote by $\mu$ the invariant measure of $\bar{U} = (Q^1,Q^2,S,\lambda)$ defined on $(W_0,\mathcal{W}_0)$. Let $\zeta_t = f((U_i)_{T_i \leq t})$ be a stationary process under $\mu$ with $f$ a measurable function valued in $(Z,\mathcal{Z})$, $Z$ a countable space and $\pi$ the stationary distribution of $\zeta_t$. The proposition below provides a fixed point formula satisfied by $\pi$.
\begin{prop}
The stationary distribution $\pi$ satisfies
\begin{equation}
\begin{array}{l}
\pi Q = 0\\
\pi \mathbf{1} = 1.
\end{array}
\label{Eq:StatDistrib}
\end{equation}
where the infinite dimensional matrix $Q$ verifies 
\begin{equation}
Q(z,z') = \sum_{e \in E(z,z')} \Esp_{\mu}[\lambda(e)|\zeta_0 = z],
\label{Eq:ComputeQestim}
\end{equation}
with $E(z,z')$ the set of events directly leading to $z'$ from $z$.
\label{Prop:GeneralProbStat}
\end{prop}
The proof of this result is provided in Appendix \ref{sec:statdistrib}.
\begin{rem}
When $\zeta_t = U_t = (Q^1_t,Q^2_t,S_t)$, Proposition \ref{Prop:GeneralProbStat} provides a fixed point equation for the computation of the stationary distribution $\pi$ of the order book.
\end{rem}
\begin{rem} The operator $Q$ is the infinitesimal generator of the process $\zeta$ defined such that $Q(z,z') = \underset{\delta \rightarrow 0}{\lim} \cfrac{\mathbb{P}_{\mu} [ \zeta_{\delta} = z' | \zeta_{0} = z ]}{\delta}$ for any $z \ne z'$. The proof of this result is given in Equation (\ref{Eq:proportQtildeQ}) of Appendix \ref{sec:statdistrib}. 
\end{rem}
\subsubsection{Markov framework}
\manuallabel{subsec:Markovframework}{5.1.1}
In the Markov case, it is a well known result that $Q$ satisfies (\ref{Eq:StatDistrib}), see \cite{citeulike:1400630}. In this case, the coefficients of $Q$ are parameters of the model and can be estimated using (\ref{Eq:Ntdivt}).

\subsubsection{General case}
Let us take $z$ and $z'$ two states such that $z \ne z'$, $N_t^{z,z'} = \sum_{T_i < t} \delta^i_{z,z'}$ with $\delta^i_{z,z'} = \mathbf{1}_{\{\zeta_{T_{i-1}} = z,\,\zeta_{T_i} = z'\}}$ and $t^{z} = \sum_{T_i < t} \Delta T_i \mathbf{1}_{\{\zeta_{T_{i-1}} = z\}} $ with $\Delta T_i = T_i - T_{i-1}$. We have the following results:
\begin{prop} When $(\delta^i_{z,z'})_{i \geq 1}$ satisfies Assumption \ref{Assump:LimitThMrkv}, we have 
\begin{align*}
\hat{Q}(z,z') = \frac{N_t^{z,z'}}{t^{z}} \underset{t \rightarrow \infty}{\rightarrow} Q(z,z'), \quad a.s.\numberthis \label{Eq:Ntdivt}
\end{align*}
\label{Prop:Estimateintens}
\end{prop}
The proof of this result is given in Appendix \ref{sec:AppendixLawLargeNumbersNtdivt}.
\begin{rem}[Confidence interval] We can compute a confidence interval for the estimator $\hat{Q}(z,z')$, see Appendix \ref{sec:AppendixLawLargeNumbersNtdivt} for the details.
\end{rem}
\begin{rem} When $\zeta_t = U_t = (Q^1_t,Q^2_t,S_t)$, Proposition \ref{Prop:Estimateintens} provides an estimator for the operator $Q(u,u')$ with $u,u'\in \mathbb{U}$ and $u \ne u'$. 
\end{rem}
\begin{rem} In the Markov case and $\zeta_t = U_t$, see \cite{huang2017ergodicity}, the authors used the estimator presented in Proposition \ref{Prop:Estimateintens} to evaluate $Q(u,u')$. 
\end{rem}
\begin{rem} Let $(z,z')\in \mathbb{U}^2$ such that $z \ne z'$ and $a \in \mathbb{A} $, we consider the quantity $Q(z,z',a) = \sum_{e \in E(z,z')\cap E(a)} \Esp[\lambda(e)|\zeta_0 = z]$ with $E(a)$ the set of events generated by the agent $a$. This quantity represents the infinitesimal probability that agent $a$ sends an order that moves $\zeta$ from $z$ to $z'$. It can be estimated by $ \hat{Q}(z,z',a) = \frac{N_t^{z,z',a}}{t^{z}} $ which satisfies 
\begin{align*}
\hat{Q}(z,z',a) = \frac{N_t^{z,z',a}}{t^{z}}  \underset{t \rightarrow \infty}{\rightarrow} Q(z,z',a), \quad a.s,\numberthis \label{Eq:Ntdivt_agent}
\end{align*}
with $N_t^{z,z',a} = \sum_{T_i < t} \delta^i_{z,z',a}$, $\delta^i_{z,z',a} = \mathbf{1}_{\{\zeta_{T_{i-1}} = z,\,\zeta_{T_i} = z',A_i = a\}}$ where $A_i$ is the identity of the agent causing the $i$-th event. The quantity $Q(z,z',a)$ allows us to infer the market dynamic (i.e the operator $Q$) for a specific combination of the agents, see Equation (\ref{Eq:ComputeQestim}).
\end{rem}
\subsection{Spread computation}
Since the process $U_t$ is ergodic the spread $S_t$ has a stationary distribution. Then, we can compute  $\Esp_{\pi} [S_{\infty} ] $ where $ \pi$ is the stationary distribution of $U$. The computation formula for $\pi$ is detailed in Proposition \ref{Prop:GeneralProbStat} and the estimation methodology of $Q$ is described in Proposition \ref{Prop:Estimateintens}.

\subsection{Price volatility computation}
\manuallabel{subsec:VolLimitProcess}{5.3}
We place ourselves in the case of Remark \ref{Rem:PriceApply} and assume that the mid price moves $(\eta_i)_{i \geq 0} $ are valued in $\zeta = \alpha_0 \mathbb{Z}$ with $\alpha_0$ the tick size. In such situation, the limit theorem of Section \ref{sec:limitTh} ensures the convergence of $\bar{P}_n(t)$ towards 
\begin{equation*}
\bar{P}_n(t)\overset{\mathcal{L}}{\longrightarrow}   \sigma W_t,
\end{equation*}
with $\sigma^2 = \Esp_{\mu} [\eta_0^2] + 2 \sum_{k \geq 1} \Esp_{\mu} [\eta_0\eta_k]$ and $\mu$ the invariant measure of $\bar{U}$. The quantity of interest is $\sigma^2 $. To compute $\sigma^2$, we need to evaluate $\Esp_{\mu} [\eta_0\eta_k]$ for all $k \geq 0$. We have 
\begin{align}
\begin{array}{lr}
\Esp_{\mu} [\eta_0^2] = \sum_{\eta \in \zeta} \pi_{\eta_0}(\eta)\eta^2, & \\
\Esp_{\mu} [\eta_0\eta_k] = \sum_{\eta \in \zeta} \pi_{\eta_0}(\eta)\eta\Esp_{\mu} [\eta_k | \eta_0 = \eta], & \forall k \geq 1,
\end{array}
\label{Eq:Computation_mkt_vol}
\end{align}
with $\pi_{\eta_0}(\eta) = \mathbb{P}_{\mu}[\eta_0 = \eta]$. Thus we need to estimate $\pi_{\eta_0}$ and $\Esp_{\mu} [\eta_k | \eta_0 = \eta]$ to evaluate $\sigma^2$. The computation of the leading term $\Esp_{\mu} [\eta_0^2]$ requires only the knowledge of the stationary distribution $\pi_{\eta_0}$. The latter is evaluated using Proposition \ref{Prop:GeneralProbStat}. To estimate $\Esp_{\mu} [\eta_k | \eta_0 = \eta]$ with $k \geq 1$, we use the following proposition.
\begin{prop} Let us take $k \geq 1$, $\eta \in \zeta$, $N_n^{\eta,(k)} = \sum_{j \leq n}  \eta_j\delta^{j\,(k)}_{\eta}$ with $\delta^{j\,(k)}_{\eta} =  \mathbf{1}_{\{\eta_{j-k} = \eta\}}$ and $n^{\eta} = \sum_{j \leq n}  \delta^{j\,(k)}_{\eta}$. When both $( \eta_i\delta^{i\,(k)}_{\eta})_{i \geq 1}$ and $(\delta^{i\,(k)}_{\eta})_{i \geq 1}$ satisfy Assumption \ref{Assump:LimitThMrkv}, we have 
\begin{align*}
\hat{E}(\eta_0,k) = \frac{N_n^{\eta_0\,(k)}}{n^{\eta}} \underset{n \rightarrow \infty}{\rightarrow} \Esp_{\mu} [\eta_k | \eta_0 = \eta], \quad a.s.\numberthis \label{Eq:Ntdivt_ee}
\end{align*}
\label{prop:DirectComputEeta}
\end{prop}
The proof of this result is similar to the one of Proposition \ref{Prop:Estimateintens}.
\begin{rem}[Markov case]When the dynamic of $U$ is Markov and $\eta_i = f_0(U_i)$ for any $i \geq 0$ with $f_0$ a deterministic function, see Remark \ref{rem:Deterministic_f_0}. We have 
\begin{align}
\Esp_{\pi} [\eta_0\eta_k] = \sum_{u \in \mathbb{U}} \pi(u) \eta_{0}(u) \Esp_{u} [\eta_k],
\label{Eq:Computation_mkt_vol_mark}
\end{align}
where $\pi$ is the stationary distribution of $U$ that can be computed using Proposition \ref{Prop:GeneralProbStat} and $\Esp_{u} [\eta_k ] = (P^k * \eta_0)_{u}= \sum_{u' \in \mathbb{U}} P^k_{u,u'} \eta_0(u')$ with $P^k$ the $k$-th power of the Markov chain $P$ associated to the process $U$ and which satisfies 
\begin{align*}
P_{u,u'} & = \left\{\begin{array}{ll}
- Q_{u,u'}/Q_{u,u} & \text{ if } u \ne u' \text{ and }  Q_{u,u} \ne 0,\\
0 & \text{ if } u \ne u' \text{ and }  Q_{u,u} = 0,
\end{array}
\right. \\
P_{u,u} & = \left\{\begin{array}{ll}
0 & \text{ if }  Q_{u,u} \ne 0,\\
1 & \text{ if }  Q_{u,u} = 0,
\end{array}
\right. \numberthis \label{Eq:PtransMarkov}
\end{align*}
where the quantity $P_{u,u'}$ represents $P_{u,u'} = \mathbb{P}[U_1 = u'| U_0 = u]$ with $U_1$ the state of the order book after one jump. 
\label{rem:EstimAveragePchanges}
\end{rem}

\begin{rem} In Section \ref{sec:Application}, for any $u = (q^1, q^2, s)$, we consider the following function: 
\begin{align*}
f_{0}(u) = \left\{
\begin{array}{ll}
-1 & \text{ if } q^1 = 0 \text{ and } q^2 > 0,\\
+1 & \text{ if } q^2 = 0 \text{ and } q^1 > 0,\\
0  & \text{ otherwise ,}
\end{array}
\right.
\end{align*}
for the numerical simulations. Note that the states where $q^1 = 0$ or $q^2 = 0$ are fictitious states that are not observable in practice. These states are introduced to handle the price changes. Indeed, the states where $q^1 = 0$ (resp. $q^2 = 0$) correspond to a price decrease (resp. increase) by one tick and the states where both $q^1 = 0$ and $q^2 = 0$ are unreachable.
\label{rem:Deterministic_f_0}
\end{rem}
\subsection{An alternative measure of market stability}
Another way to look at market stability is to investigate the behaviour of the disequilibrium between offer and demand. This equilibrium can be for example measured through the cumulative imbalance $N_t = V^{b}_t - V^{a}_t$ where $V^{b}_t$ (resp. $V^a_t$) is the net number of inserted limit orders at the bid (resp. ask). From no arbitrage argument, we know that the dynamic of $N_t$ is closely related to that of the price \cite{jaisson2015market,jusselin2018no}. Consequently, it is natural to view the long term volatility of this object as an alternative measure of market stability.\\

In this section, we follow the same methodology of Section \ref{subsec:VolLimitProcess}. The cumulative imbalance after $n$ jumps $N_n$ satisfies $N_n = N_0 + \sum_{i=1}^n \Delta N_i $ where $\Delta N_i = N_{i} - N_{i-1} = n_i$. Hence, when $(n_i)_{i \geq 0}$ satisfies Assumptions \ref{Assump:LimitThMrkv} and \ref{Assump:LimitThMrkv2}, we have the following convergence result:
\begin{equation*}
X^{N}_n = \cfrac{\sum_{k=1}^n (n_k - \Esp_{u} [n_k])}{\sqrt{n}} \overset{\mathcal{L}}{\longrightarrow}   \tilde{\sigma} W_t,
\label{Eq:CvgResult2}
\end{equation*}
with $\tilde{\sigma}^2 = \Esp_{\mu} [n_0^2] + 2 \sum_{k \geq 1} \Esp_{\mu} [n_0n_k]$ and $\mu$ the stationary distribution of $\bar{U}$ given by proposition \ref{Prop:GeneralProbStat}. The quantity $\Esp_{\mu} [n_0n_k]$ can be computed using the same methodology of Section \ref{subsec:VolLimitProcess}.

\section{Numerical experiments}
\label{sec:Application}
In this section, we propose a ranking of the market makers for four different assets, based on their impact on volatility. For each asset, we compute first the liquidity provision and consumption intensities relative to the whole market using Equation (\ref{Eq:Ntdivt})\footnote{A liquidity provision (resp. consumption) event is assimilated to an increase (resp. decrease) of the best bid or ask size by $1$ unit. To fix the ideas, one (AES) is our unit here.}. Then, we estimate the stationary measure of the order book, see Equation (\ref{Eq:StatDistrib}), and use it to compute the two following estimators of the market volatility:
\begin{align*}
\begin{array}{lcl}
\sigma^{2,G} & = & \Esp_{\mu}[\eta_0^2], \\
\sigma^{2,M}_k & = & \Esp_{\pi}[\eta_0^2] + 2 \sum_{j=1}^k\Esp_{\pi}[\eta_0 \eta_j],
\end{array}
\end{align*}
where $\mu$ is the invariant measure of $\bar{U}$ given by Theorem \ref{lem:ErgodicityProcHks}, $\pi$ is the stationary distribution of $U$ when both the order book dynamic is Markov and $\eta_i = f_0(U_i)$ with $f_0$ defined in Remark \ref{rem:Deterministic_f_0}. The estimator $\sigma^{2,G}$ is computed by applying Equation \eqref{Eq:Computation_mkt_vol} and $\sigma^{2,M}_k$ is evaluated using Remark \ref{rem:EstimAveragePchanges}. Thereafter, for each market maker, we compute  its own intensities using Equation \eqref{Eq:Ntdivt_agent}. After that, we estimate the new market intensities in a situation where we suppose that he withdraws from the exchange by subtracting the agent intensity from the market one, see Corollary \ref{Prop1:MktReconstit}. We finally compute the new market volatility estimators $\sigma^{2,G}$ and $\sigma^{2,M}_k$ corresponding to this new scenario using Equation \eqref{Eq:Computation_mkt_vol} and Remark \ref{rem:EstimAveragePchanges} again. 
\begin{rem} In the simple case where the order book dynamic is Markov and the queues are independent, see Section 2.3.3 in \cite{huang2015simulating}, minimizing the first order approximation of the price volatility $\sigma^2 \sim \Esp_{\pi} [\eta_0^2]$ is similar to selecting the agent with the highest ratio insertion/consumption $\frac{\lambda^{1(2),+}_Q}{\lambda^{1(2),-}_Q}$. This condition is a well-known result which means that the new agent needs to have an insertion/consumption ratio greater than the one of the market. The proof of this result is given in Section \ref{sec:proofRemVol_ratio_cons_insert}.
\label{Rem:Vol_ratio_cons_insert}
\end{rem}
\begin{rem}The reconstruction methodology of the market assumes that other participants will not modify their behaviours when an agent leaves the market. In practice, this assumption is satisfied since agents react to global variables such as the imbalance and not to a specific agent-based information. Additionally, when an agent leaves the market, the other participants do not have enough order flow history to calibrate all the parameters of their models.
\end{rem}
\begin{rem}
The reconstruction methodology of the market takes into account the volume exchanged by each agent since this information is included in the estimated intensities. Indeed, the intensity of an agent who trades a large volume is high because he either interacts frequently with the market or generates significant changes in the order book state.
\end{rem}	

\subsection{Database description.} 

We study four large tick European stocks: Air Liquid, EssilorLuxottica, Michelin and Orange, on Euronext, over a year period: from January 2017 till December 2017. The data under study are provided by the French Regulator Autorit\'e des march\'es financiers. For each of these assets, we have access to the trades and orders data. Using both data, we rebuild the Limit Order book (LOB) up to the first limit of both sides, whenever an event (an order insertion, an order cancellation or an aggressive order) happens on one of these limits. Note that we remove market data corresponding to the first and last hour of trading, as these periods have usually specific features because of the opening/closing auction phases. We present in Table \ref{Pre_stats} some preliminary statistics on the different considered assets. \\

\begin{table}[H]
\resizebox{\linewidth}{!}{%
\begin{tabular}{|R{2.8 cm}|C{2. cm}|C{2.5 cm}|C{2.5 cm}|C{2.7 cm}|C{1.8 cm}|}
\hline
Asset & Number of insertion orders (in millions of orders) & Number of cancellation orders (in millions of orders) & Number of aggressive orders (in millions of orders)& Ratio of cancellation orders number over aggressive orders number& Average spread (in ticks)\\
\hline
Air Liquide & 2.36 & 2.40 &  0.21 &11.4& 1.07 \\
\hline 
EssilorLuxottica &  3.90 & 3.96 &  0.34 & 11.6 &1.11 \\
\hline 
Michelin & 3.81 & 4.01 & 0.32 &12.5& 1.14 \\
\hline 
Orange & 6.60 & 6.66 & 0.47 &14.1& 1.14 \\
\hline 
\end{tabular}}
\caption[Table \ref{Pre_stats}]{Preliminary statistics on the assets.}
\label{Pre_stats}
\end{table}

Table \ref{Pre_stats} shows that the number of insertion orders is lower than that of cancellation orders. A priori, this seems contradictory, but what happens in practice is that some agents insert orders that they cancel  partially and progressively at a later stage by sending multiple cancellation orders, which leads to a number of cancellation orders higher than that of insertion orders. \\

The considered market makers, that we aim at ranking, are the Supplemental Liquidity Providers (SLP) members. The SLP programme imposes a market making activity on programme members, including order book presence time at competitive prices. In return, they get favorable pricing and rebates in the form of a maker-taker fees model directly comparable to those of the major competing platforms. This programme includes 9 members. Some of them have at the same time SLP activity and other activities, such like proprietary or agency activity. In our analysis, we only analyse the SLP flow of these members. We denote the market makers by MM1 to MM9. 

\subsection{Computation of the intensities and the stationary measure} 

We compute the liquidity consumption and provision intensities at the first limit relative to the whole market according to the queue size, the corresponding stationary measure and the long term volatility for Air Liquide. Results relative to EssilorLuxottica, Michelin and Orange are relegated to Appendix \ref{sec: appendix stationary measure}. The estimation methodology of the intensities is based on Proposition \ref{Prop:Estimateintens}. To apply this proposition, we record, for every event occurring in the LOB at the best limits (best ask and  bid), the type of this order (insertion or consumption), the waiting time (in number of seconds) between this event and the preceding one occurring at the same limit and the queue size before the event. The queue size is then approximated by the smaller integer that is larger than or equal to the volume available at the queue, divided by the stock average event size (AES) computed for each limit on a daily basis.  In practice, the spread cannot be equal to one tick all the time. This is why we exclude from our analysis all the events that occur when the spread is higher than one tick.

\begin{figure}[H]
\centering
    \hfill (a) Intensity of the market \hfill (b) Stationary measure $Q^1$ \hfill ~\\
        \includegraphics[width=0.42\linewidth]{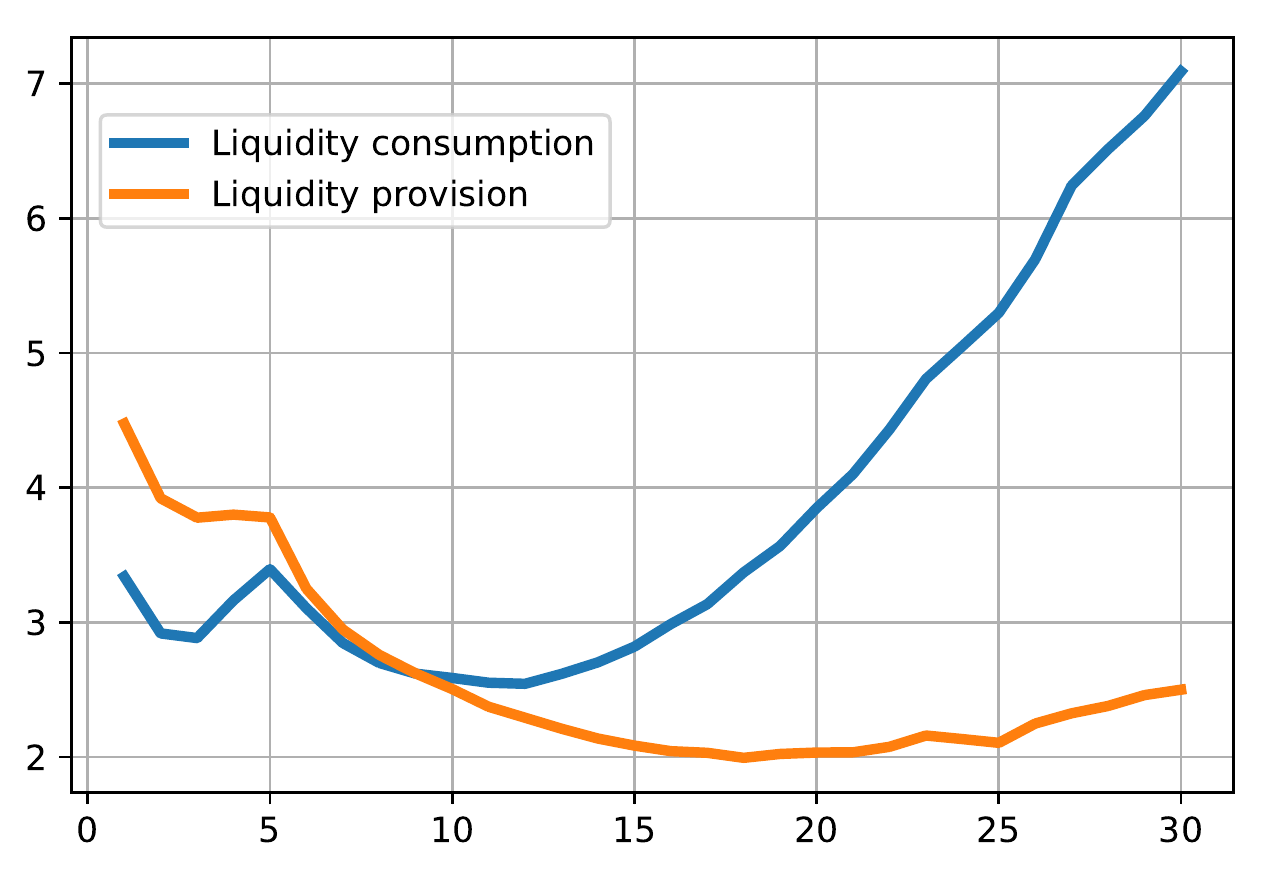}
	    \includegraphics[width=0.47\linewidth]{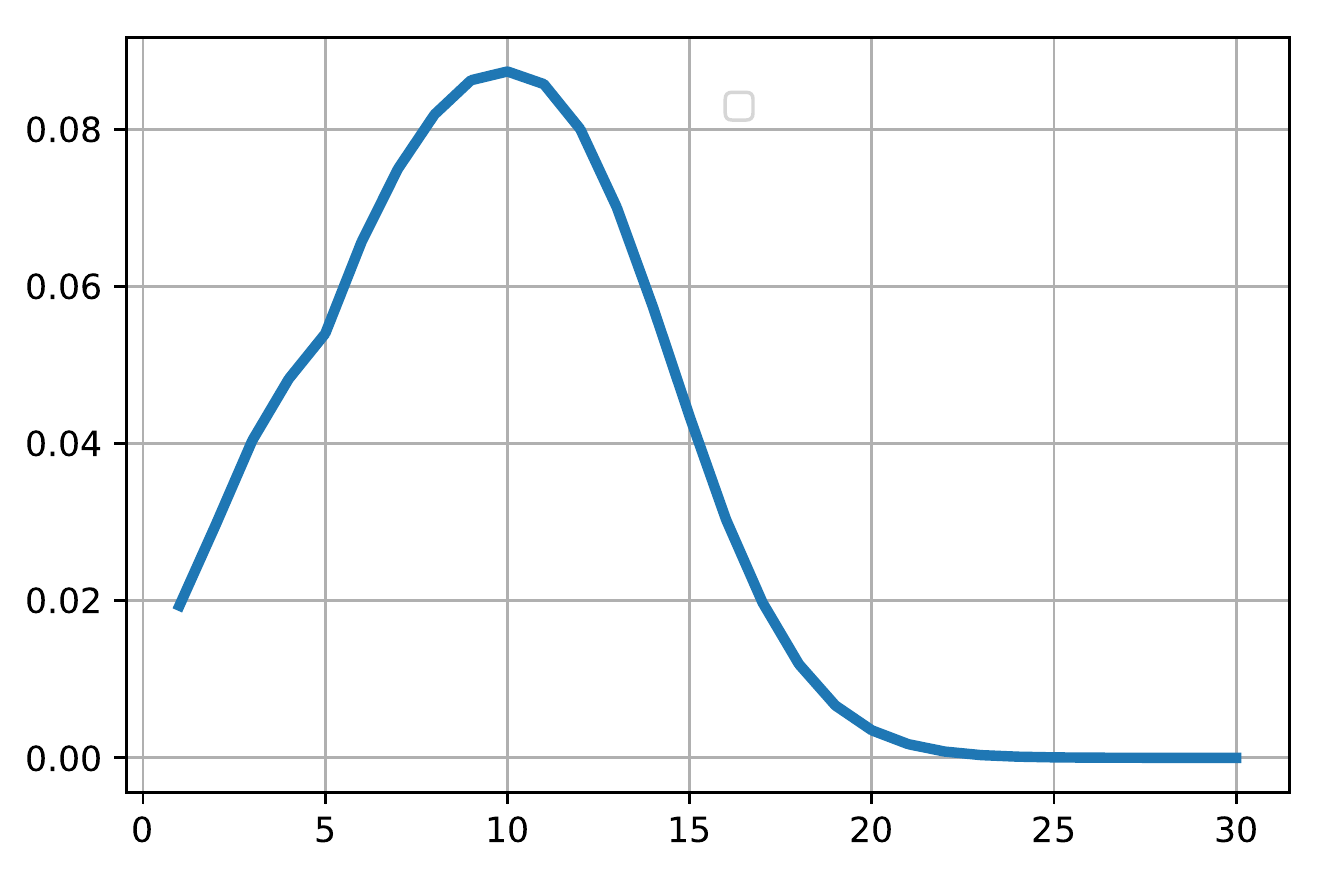}\\
	    \hfill Long term price volatility $\sigma^{2,G} = 0.035$, $\sigma^{2,M}_{10} = 0.227$. \hfill ~\\
	    \caption{(a) Liquidity insertion and consumption intensities (in orders per second) with respect to the queue size (in AES) and (b) the corresponding stationary distribution of $Q^1$ with respect to the queue size (in AES), proper to Air Liquide.}
	     \label{AirLiquide}
\end{figure}

We can see that for all these assets, the liquidity provision  intensity is approximately a decreasing function of the queue size. This result reveals a quite common strategy used in practice: posting orders when the queue is small to seize priority (for further details about the priority value, see \cite{huang2019glosten}). For all assets, the consumption intensity is an increasing function when the queue size is large. For small queue sizes, we notice a slight decrease of this intensity, see Figure \ref{AirLiquide}. Indeed, the increasing aspect corresponding to large queue sizes is explained by market participants waiting for better price when liquidity is abundant. The decreasing aspect associated to small queue sizes is due to aggressive orders sent by agents to get the last remaining quantities available at the first limits: market participants rushing for liquidity when it is rare. The lower the ratio of cancellation orders number over aggressive orders number is, the clearer the decreasing shape for small queue sizes stands out, see Table \ref{Pre_stats} and Figures \ref{AirLiquide}, \ref{Exilor}, \ref{Michelin} and \ref{Orange}.

\subsection{Ranking of the market makers}
For each of the assets and for each one of the market makers, we compute the liquidity consumption and provision intensities, and the corresponding price volatility $\sigma^{2,M}_{10}$ that we would obtain in a situation where the studied market maker withdraws from the market. Since the estimators $\sigma^{2,G}$ and $\sigma^{2,M}_{10}$ give the same ranking, we choose to show the values for $\sigma^{2,M}_{10}$ alone. We show next the results relative to Air Liquide; those of EssilorLuxottica, Michelin and Orange are relegated to Appendix \ref{sec: appendix stationary measure}.

\begin{figure}[H]
\centering
    \hfill Intensities and $\sigma^{2,M}_{10}$ when one market maker leaves the market\hfill ~\\
      \adjustbox{max height=\dimexpr\textheight-.5cm\relax,
           max width=\dimexpr\textwidth-1.cm\relax}{
        \includegraphics[width=1\linewidth]{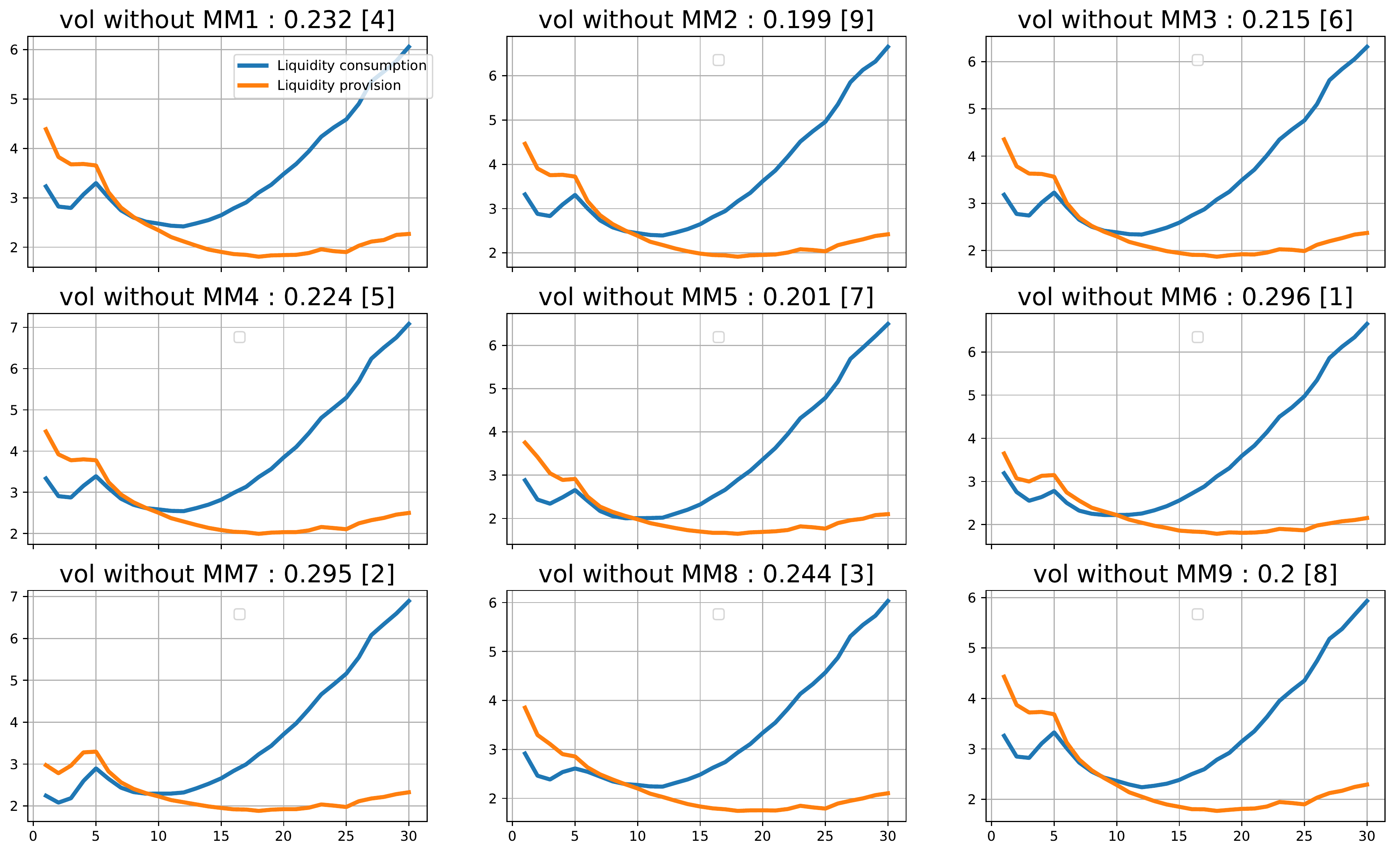}
        }
	    \caption{Liquidity insertion and consumption intensities (in orders per second) with respect to the queue size (in AES) and $\sigma^{2,M}_{10}$ when one market maker is ejected from the market for the stock Air Liquide.}
\end{figure}

Based on the previous results, we carry out for each asset the ranking of the different market makers according to their contribution to volatility. To do so, we compare the expected volatility when removing each market maker from the market to the actual one when all the market makers in the market: if the expected volatility is higher (resp. lower) than the actual one, this means that the market maker into question decreases (resp. increases) market volatility. The market maker who decreases\footnote{The expected volatility of the new market without this market maker is the highest.} (resp. increases\footnote{The expected volatility of the new market without this market maker is the lowest.}) volatility the most is ranked first (resp. last). In the following table, we add a star next to  market makers deceasing volatility: a zero star (resp. a four stars ) means that the market maker increases (resp. decreases) the market volatility of the 4 studied assets. 

\begin{table}[H]
\resizebox{\linewidth}{!}{%
\begin{tabular}{|R{2 cm}|C{1.5 cm}|C{1.5 cm}|C{1.5 cm}|C{1.5 cm}|C{1.5 cm}|C{1.5 cm}|C{1.5 cm}|C{1.5 cm}|}
\hline
Market maker & Ranking Air Liquide &Market share Air Liquide& Ranking ExilorLuxottica & Market share ExilorLuxottica & Ranking Michelin &Market share Michelin & Ranking Orange& Market share Orange\\
\hline
MM1*** & 4 & 4\% & 3 & 3\% &  3 & 4\% & 3 & 3\% \\
\hline 
MM2 &  9 & 1\%& 9 & 1\%&  9 & 1\% & 7 & 1\%\\
\hline 
MM3 & 6 & 5\% & 6 &5\%& 7 & 4\%& 5& 4\% \\
\hline 
MM4 & 5 & 1\% & 4 &1\% & 4 & 0\% & 4 & 1\% \\
\hline 
MM5 & 7 & 5\% & 8 & 5\%& 8 & 5\% & 9 & 5\% \\
\hline 
MM6**** & 1 &3\% & 2 &3\% & 1 & 3\% & 1 & 4\% \\
\hline 
MM7**** & 2 & 7\% & 1 & 12\% & 2 & 9\% & 2 & 7\% \\
\hline 
MM8* & 3 & 9\% & 5 & 5\% & 5 & 5\% & 6 & 4\% \\
\hline 
MM9 & 8 & 2\% & 7 & 2\% & 6 & 2\% & 8 & 2\% \\
\hline 
\end{tabular}}
\caption[Table \ref{Ranking}]{Market share and ranking of markets makers }
\label{Ranking}
\end{table}

\clearpage
\appendix 
\section{Market reconstitution}
\label{sec:MktReconstitution}
\begin{proof}[Proof of Proposition \ref{Prop1:MktReconstit}.]Let $t \geq 0$ be the current time. For any $B \in \mathcal{E}$, we denote by $T^{t,e}$ the first time greater than $t$ when an event $e \in B$ happens given $\mathcal{F}_{t}$ and $T^{t,B} = \min_{e \in B} T^{t,e}$ the next market event. Thus, we have 
\begin{align*}
\lambda_t (B) & = \underset{\delta t \rightarrow 0}{\lim} \cfrac{\mathbb{P} \big[ \#\{T_n \in (t,t+\delta t] , X_n \in B \} \geq 1 |  \mathcal{F}_t\big]}{\delta t } \\
				& = \underset{\delta t \rightarrow 0}{\lim} \cfrac{\mathbb{P} \big[ \{T^{t,B} \in (t,t+\delta t]\} |  \mathcal{F}_t\big]}{\delta t }.
\end{align*}
We write $f^{t,e}$ for the density function of $T^{t,e}$ and $F^{t,e}_B(s) = \mathbb{P}[\big(\min_{\tilde{e} \in B \setminus \{e\}} T^{t,\tilde{e}}\big) \geq s|T^{t,e} \leq s] $ for any $s \geq 0$. Using the monotone convergence theorem, we have 
 \begin{align*}
\underset{\delta t \rightarrow 0}{\lim} \cfrac{\mathbb{P} \big[ \{T^{t,B} \in (t,t+\delta t]\} |  \mathcal{F}_t\big]}{\delta t }  & = \underset{\delta t \rightarrow 0}{\lim} \cfrac{\sum_{e \in B} \int_{t}^{t + \delta t} f^{t,e}(s) F^{t,e}_B(s) \, ds}{\delta t}\\
							   & = \sum_{e \in B} \underset{\delta t \rightarrow 0}{\lim} \cfrac{ \int_{t}^{t + \delta t} f^{t,e}(s) F^{t,e}_B(s) \, ds}{\delta t} \\
							   & = \sum_{e \in B}  f^{t,e}(t) F^{t,e}_B(t) = \sum_{e \in B} \lambda_t(e),
\end{align*}
since $f^{t,e'}_a(t) = \lambda_t((e',a))$ using Equation (\ref{Eq:IntensityExpressionHawkes}) and $F^{t,e'}_a(t) = 1$ by definition.
 This completes the proof. 
\end{proof}

\section{Proof of Remark \ref{Rem:NegativeCondDrift}}
\label{sec:ProofRemNegativeCondDrift}
\begin{proof}[Proof of Remark \ref{Rem:NegativeCondDrift}]Let $ N =(T_n,X_n)$ be the point process defined in Section \ref{sec:Mktmodel} and $ i \in \mathbb{B} = \{1,2\}$. We define $\phi^{i,\pm,n}_Q$ in the following way:
\begin{align*}
\begin{array}{lcl}
\phi^{i,+,n}_Q & = &\sup_{e \in E^{i\,+}_Q(u,n)}\sum_{T_i< t}\phi(e,U_{t^{-}},n,t-T_i,X_i), \\
\phi^{i,-,n}_Q & = &\inf_{e \in E^{i\,-}_Q(u,n)}\sum_{T_i< t}\phi(e,U_{t^{-}},n,t-T_i,X_i),
\end{array}
\end{align*}
with $U_t = (Q^1_t,Q^2_t,S_t)$. When $Q^i_{t^-} \geq C_{bound}$, using that $\psi$ is non-decreasing in $z$, we have 
\begin{align*}
\sum_{n \geq 0} (z_0^{n} -1) \big(\lambda^{i,+}_{Q}(U_{t^{-}},n) - \lambda^{i,-}_{Q}(U_{t^{-}},n)&\frac{1}{z_0^{n}}\big)  \leq \sum_{n \geq 0} (z_0^{n} -1) \big(\psi^{i,+}_{Q}(U_{t^{-}},n,t,\phi^{i,+,n}_Q) - \lambda^{i,-}_{Q}(U_{t^{-}},n)\frac{1}{z_0^{n}}\big) \\
					  & = \sum_{n \geq 0} (z_0^{n} -1)  \big(\psi^{i,+}_{Q}(U_{t^{-}},n,t,\phi^{i,+,n}_Q) - \psi^{i,-}_{Q}(U_{t^{-}},n,t,\phi^{i,+,n}_Q)\frac{1}{z_0^{n}}\big)\\
					  & + \sum_{n \geq 0} (1 -\frac{1}{z_0^{n}}) \big(\psi^{i,-}_{Q}(U_{t^{-}},n,t,\phi^{i,+,n}_Q) - \lambda^{i,-}_{Q}(U_{t^{-}},n)\big) = (i) + (ii).
\end{align*}
Using Equation (\ref{Eq:NegDriftQtyExample}), we have 
\begin{equation}
(i) = \sum_{n \geq 0} (z_0^{n} -1) \big(\psi^{i,+}_{Q}(U_{t^{-}},n,t,\phi^{i,+,n}_Q) - \psi^{i,-}_{Q}(U_{t^{-}},n,t,\phi^{i,+,n}_Q)\frac{1}{z_0^{n}}\big) \leq - \delta,\quad a.s,
\label{Eq:RemNegDriftProof1}
\end{equation}
when $Q^i_{t^-} \geq C_{bound}$. Moreover, using that $\psi$ is non-decreasing in $z$, we have 
\begin{align*}
(ii) & = \sum_{n \geq 0} (1 -\frac{1}{z_0^{n}}) \sum_{e \in E^{i\,-}_Q(u,n)}\big( \psi(e,U_{t^{-}},n,t,\phi^{i,+,n}) - \psi(e,U_{t^{-}},n,t,\sum_{T_i< t}\phi(e,U_{t^{-}},n,t-T_i,X_i))\big) \\
     & \leq \sum_{n \geq 0} (1 -\frac{1}{z_0^{n}}) \sum_{e \in E^{i\,-}_Q(u,n)}\big( \psi(e,U_{t^{-}},n,t,\phi^{i,+,n}) - \psi(e,U_{t^{-}},n,t,\phi^{i,-,n})\big), \quad a.s,
\label{Eq:RemNegDriftProof2}
\end{align*}
when $Q^i_{t^-} \geq C_{bound}$. Since Equation (\ref{Eq:NegDriftQtyExample}) ensures that $\phi^{i,+,n} \leq \phi^{i,-,n}, \quad a.s$ and $\psi$ is non-decreasing in $z$, we deduce that  
\begin{equation}
(ii) = \sum_{n \geq 0} (1 -\frac{1}{z_0^{n}}) \big(\psi^{i,-}_{Q}(U_{t^{-}},n,t,\phi^{i,+,n}_Q) - \lambda^{i,-}_{Q}(U_{t^{-}},n)\big) \leq 0, \quad a.s,
\label{Eq:RemNegDriftProof2}
\end{equation} 
when $Q^i_{t^-} \geq C_{bound}$. Using Equations (\ref{Eq:RemNegDriftProof1}) and (\ref{Eq:RemNegDriftProof2}), we get
$$
\sum_{n \geq 0} (z_0^{n} -1) \big(\lambda^{i,+}_{Q}(U_{t^{-}},n) - \lambda^{i,-}_{Q}(U_{t^{-}},n)\frac{1}{z_0^{n}}\big) \leq - \delta\quad a.s,
$$
when $Q^i_{t^-} \geq C_{bound}$. By following the same methodology, we also get
$$
\sum_{n \geq 0} (z_0^{n} -1) \big(\lambda^{i,+}_{S}(U_{t^{-}},n) - \lambda^{i,-}_{S}(U_{t^{-}},n)\frac{1}{z_0^{n}}\big) \leq - \delta, \quad a.s,
$$
when $S_{t^-} \geq C_{bound}$. This completes the proof.
\end{proof}

\section{Proof of Theorem \ref{lem:ExistenceProcHks}}
\label{sec:ErgValFctHks}
\subsection{Preliminary results}
For any $k \geq 1$, we denote by $T_{n+1}(e)$, $T^{i\pm}_{Q_{n+1}}(k)$ and $T^{i\pm}_{S_{n+1}}(k)$ respectively the arrival time of the first event $e$, $e^{i\pm}_{Q}(k) \in E^{i,\pm}_{Q}$ and $e^{\pm}_{S}(k) \in E^{\pm}_{S}$ greater than $T_{n}$. The sets $E^{i,\pm}_{Q}$ and $E^{\pm}_{S}$ are defined in Equation \eqref{Eq:SetEvtModifyLOB}. They contain the events that increase or decrease the best bid, best ask and spread by $k$. 
\begin{lem} Let $n\geq 0$ and $i \in \mathbb{B}$. The order book increments satisfy the following formulas:
\begin{align*}
 \mathbb{P}[\Delta Q^{i}_{n+1} = \pm k] & = \Esp\big[\int_{\mathbf{R}_+} \lambda^{i,\pm}_{Q_n}(t,k) Z_{n}(t)\, dt\big], \\
 \mathbb{P}[\Delta S_{n+1} = \pm k] & = \Esp\big[\int_{\mathbf{R}_+} \lambda^{\pm}_{S_n}(t,k) Z_{n}(t)\, dt\big],\\
\end{align*}
with $\Delta Q^{i}_{n+1} = Q^{i}_{n+1} - Q^{i}_{n}$, $\Delta S_{n+1} = S_{n+1} - S_{n} $ and 
\begin{align*}
Z_n(t) = e^{-[\sum_{e}\int_{0}^{t}\lambda_n(e,s + T_n)\, ds]}, & \qquad \lambda^{i,\pm}_{Q_n}(t,k) = \sum_{e \in E^{i\pm}_{Q}(k)} \lambda_n(e,t + T_n), \\
\lambda^{\pm}_{S_n}(t,k) = \sum_{e \in E^{i\pm}_{S}(k)} \lambda_n(e,t + T_n), & \qquad \lambda_n(e,t) = \psi\big(e,U_{T_n},t, \sum_{T_i \leq T_n} \phi(e,U_{T_n},t-T_i,X_i) \big), \quad \forall t \geq 0.
\end{align*}
\label{Lem:ComputeDeltaU}
\end{lem}
\begin{proof}[Proof of Lemma \ref{Lem:ComputeDeltaU}] We write $\Delta T_{n+1}(e) = T_{n+1}(e) - T_{n}$ for any event $e \in E$ and $\Delta T^{i\pm}_{Q_{n+1}}(k) = T^{i\pm}_{Q_{n+1}}(k) - T_{n}$. Using Remark \ref{Rem:Rem10}, the increments $(\Delta T_{n+1})_{n \geq 0}$ are independent given $\mathcal{F}_{n}$ and $\Delta T_{n+1}(e)|\mathcal{F}_{n}$ follows a non homogeneous exponential distribution with an intensity $\lambda_{n}(e,.)$. Thus, we have
\begin{align*}
 \mathbb{P}[\Delta Q^{i}_{n+1} = \pm k] & =  \Esp[\mathbb{P}[\Delta T^{i\pm}_{Q_{n+1}}(k) < \Delta T_{n+1}(e), \quad \forall e \notin E^{i\pm}_Q(k) | \mathcal{F}_{n}]] \\
 						 & = \Esp\big[\int_{\mathbf{R}_+} \lambda^{i,\pm}_{Q_n}(t,k) \,e^{-\int_{0}^{t}\lambda^{i,\pm}_{Q_n}(s,k)\,ds}\, dt \prod_{e \notin E^{i\pm}_Q(k)} \big(\int_{\mathbb{R}_+}\mathbf{1}_{t < t_{e}}\lambda_n(e,t_e) e^{- \int_{0}^{t_e}\lambda_n(e,s)\, ds}dt_{e}\big)\big]\\
 						 & = \Esp\big[\int_{\mathbf{R}_+} \lambda^{i,\pm}_{Q_n}(t,k) e^{-[\sum_{e}\int_{0}^{t}\lambda_n(e,s)\, ds]}\, dt\big] = \Esp\big[\int_{\mathbf{R}_+} \lambda^{i,\pm}_{Q_n}(t,k) Z_{n}(t)\, dt\big]. \numberthis \label{Ineq:ProbaEvent}
\end{align*}
By following the same methodology used in Equation \eqref{Ineq:ProbaEvent}, we get
\begin{align*}
\mathbb{P}[\Delta S_{n+1} = \pm k] & = \Esp\big[\int_{\mathbf{R}_+} \lambda^{\pm}_{S_n}(t,k) Z_{n}(t)\, dt\big],
\end{align*}
which completes the proof.
\end{proof}
Let $ \tau_{\mathcal{O}}$ be the first entrance period of $N^i = (T_{i+j},X_{i+j})_{j \leq 0}$ to the set $\mathcal{O} \in \mathcal{W}_0$, $C_{bound}$ defined in Assumption \ref{Assump:Negativeindividualdrift} and  $1 < z \leq \min(z_0,z_1)$ with $z_0$ and $z_1$ are respectively defined in Assumptions \ref{Assump:Negativeindividualdrift} and \ref{Assump:Boundedness}.
\begin{lem}[Drift condition] Under Assumptions \ref{Assump:Negativeindividualdrift} and \ref{Assump:Boundedness}, the process $U_n = (Q^1_n,Q^2_n,S_n)$ satisfies the following drift condition:
\begin{align*}
\Esp\big[ z^{Q^i_{n+1} - C_{bound}}\mathbf{1}_{\tau_{\mathcal{O}} \geq n+1}\big] &\leq \lambda \Esp\big[ z^{Q^i_{n}- C_{bound}}\mathbf{1}_{\tau_{\mathcal{O}} \geq n+1}\big] + B \Esp\big[ \mathbf{1}_{\tau_{\mathcal{O}} \geq n+1}\big],\\
\Esp\big[ z^{S_{n+1}- C_{bound}}\mathbf{1}_{\tau_{\mathcal{O}} \geq n+1}\big] & \leq \lambda \Esp\big[ z^{S_{n}- C_{bound}}\mathbf{1}_{\tau_{\mathcal{O}} \geq n+1}\big] + B \Esp\big[ \mathbf{1}_{\tau_{\mathcal{O}} \geq n+1}\big],&  \quad \forall n \in \mathbb{N}, \forall i \in \mathbb{B},
\end{align*}
with $\lambda < 1$ and $B$ two constants.
\label{Lem:negDriftCondition}
\end{lem}
\begin{rem} We define
\begin{equation}
V_{C_{bound}}(u) = \sum_{i\in \{1,2\}} z^{q^i-C_{bound}} + z^{s-C_{bound}}, \quad \forall u \in \mathbb{U}.
\label{Eq:V}
\end{equation}
Using Lemma \ref{Lem:negDriftCondition}, we deduce that 
\begin{equation*}
\Esp\big[ V_{C_{bound}}(U_{n+1})\mathbf{1}_{\tau_{\mathcal{O}} \geq n+1}\big] \leq \lambda \Esp\big[ V_{C_{bound}}(U_{n})\mathbf{1}_{\tau_{\mathcal{O}} \geq n+1}\big] + 3B \Esp\big[ \mathbf{1}_{\tau_{\mathcal{O}} \geq n+1}\big], \quad \forall n \in \mathbb{N},
\end{equation*} 
\label{rem:NegDriftV_0}
\end{rem}
\begin{proof}[Proof of Lemma \ref{Lem:negDriftCondition}]
We write $ \tilde{\Esp}\big[X\big] = \Esp\big[X\mathbf{1}_{\tau_{\mathcal{O}} \geq n+1}\big]$  for any random variable $X$ to simplify the notations and $V$ instead of $V_{C_{bound}} $ since there is no possible confusion. We have
\begin{align*}
\tilde{\Esp}\big[z^{Q^i_{n+1}}| \mathcal{F}_n\big] & =\tilde{\Esp}\big[z^{Q^i_{n}}| \mathcal{F}_n\big] + \sum_{u'\ne U_n} \tilde{\mathbb{P}}[Q^i_{n+1} = q'| \mathcal{F}_n]\left[ z^{q'} - z^{Q^i_{n}} \right].
\end{align*}
Using Lemma \ref{Lem:ComputeDeltaU}, we get
\begin{align*}
\mathbb{P}[\Delta Q^{i}_{n+1} = \pm k] & = \Esp\big[\int_{\mathbf{R}_+} \lambda^{i,\pm}_{Q_n}(t,k) Z_{n}(t)\, dt\big],
\end{align*}
which leads to 
\begin{align*}
\tilde{\Esp}\big[z^{Q^i_{n+1}}\big] & =\tilde{\Esp}\big[z^{Q^i_{n}}\big] + \tilde{\Esp}\left[\int_{\mathbb{R}_+} Z_n(t) \left\{\sum_{k \geq 1}\lambda^{i,+}_{Q_n}(t,k)\left[ z^{Q^i_{n}+k} - z^{Q^i_{n}} \right] + \sum_{k \geq 1}\lambda^{i,-}_{Q_n}(t,k)\left[ z^{Q^i_{n}-k} - z^{Q^i_{n}} \right]  \right\}\, dt\right], \\
      	   	     &  = \tilde{\Esp}\big[z^{Q^i_{n}}\big]  +\tilde{\Esp}\big[\int_{\mathbb{R}_+} Z_n(t) \left\{ \mathcal{Q}_{u}(t,U_{n})\right\}\, dt\big] , \numberthis \label{Eq:DecompV_0}
\end{align*}
with $\mathcal{Q}_{u}(t,U_{n}) = \sum_{k \geq 1}\lambda^{i,+}_{Q_n}(t,k)\left[ z^{Q^i_{n}+k} - z^{Q^i_{n}} \right] + \sum_{k \geq 1}\lambda^{i,-}_{Q_n}(t,k)\left[ z^{Q^i_{n}-k} - z^{Q^i_{n}} \right] $. By rearranging the above terms, we get
\begin{align*}
\mathcal{Q}_u (t,U_{n}) & = z^{Q^i_n - C_{bound}} \sum_{1\leq k } (z^k -1) \left[\lambda^{i,+}_{Q_n}(t,k)  - \lambda^{i,-}_{Q_n}(t,k) \frac{1}{z^k} \right] .
\end{align*}
We write $\tilde{\Esp}\big[\int_{\mathbb{R}_+} Z_n(t) \left\{ \mathcal{Q}_{u}(t,U_{n})\right\}\, dt\big] = T_1 + T_2 $ with 
\begin{align*}
T_1 & = \tilde{\Esp}\big[\int_{\mathbb{R}_+} Z_n(t) \mathbf{1}_{Q^i_n \leq C^{bound}}\left\{ \mathcal{Q}_{u}(t,U_{n})\right\}\, dt\big], \\
T_2 & = \tilde{\Esp}\big[\int_{\mathbb{R}_+} Z_n(t) \mathbf{1}_{Q^i_n > C^{bound}}\left\{ \mathcal{Q}_{u}(t,U_{n})\right\}\, dt\big]. 
\end{align*}
We first handle the term $T_1$. When $Q^i_n \leq C^{bound}$, the quantity $z^{Q^i_n  - C_{bound}}<1$ is bounded. Additionally, we have $\sum_{e \in E} \lambda_n(e,s + T_n) \geq \underline{\psi} > 0$ under Assumption \ref{Assump:Boundedness}. This ensures that $Z_n(t) \leq e^{- \psi t}$, $a.s$. Thus, there exist $c^1> 0$ and $d^1>0$ such that
\begin{align*}
T_1  \leq \int_{\mathbb{R}_+} e^{- \psi t} \tilde{\Esp}\big[\left\{ \mathcal{Q}_{u}(t,U_{n})\right\}\, dt\big] \leq  -c^1\tilde{\Esp}\left[z^{Q^i_n  - C_{bound}}\mathbf{1}_{Q^i_n \leq C^{bound}}\right]  + d^1.
\numberthis \label{Eq:QVM}
\end{align*} 
In the last inequality we used Assumption \ref{Assump:Boundedness} again. For the term $T_2$, we use Assumption \ref{Assump:Negativeindividualdrift} and $Z_n(t) \leq e^{- \psi t}$, $a.s$, to deduce that 
\begin{equation}
T_2  \leq - \cfrac{\delta}{\psi}  \tilde{\Esp}\left[z^{Q^i_n - C_{bound}}\mathbf{1}_{Q^i_n > C^{bound}}\right].
\label{Eq:QVMC}
\end{equation}
By combining Inequalities (\ref{Eq:QVM}) and (\ref{Eq:QVMC}), we have
\begin{equation*}
\tilde{\Esp}\big[\int_{\mathbb{R}_+} Z_n(t) \left\{ \mathcal{Q}_{u}(t,U_{n})\right\}\, dt\big] \leq - c\tilde{\Esp}\big[z^{Q^i_n - C_{bound}} \big] +\tilde{\Esp}\big[d\big],
\end{equation*}
with $c = \min(c^1,\frac{\delta}{\psi})$ and $d = d^1$ which proves the first inequality of Lemma \ref{Lem:negDriftCondition}. By following the same steps, we also prove the second inequality. This completes the proof.
\end{proof}

\subsection{Outline of the proof}
To prove the existence of an invariant distribution, we first construct $N$ as a limiting process of the sequence $N^m$ defined in Remark \ref{Rem:Rem10}. This construction is based on the thinning algorithm. After that, we show, in Steps (ii) and (iii), that $N$ is well defined. Then, we introduce the process $\bar{U}^{\infty} = \esssup_{t \geq 0} \bar{U}_t$ which dominates $\bar{U}_t$ and prove that is does not explode in Step (iv). This ensures the tightness of the family $ \cup_{t \geq 0}\bar{U}_t$. Additionally, the process $\bar{U}$ satisfies the Feller property since $E$ is a countable space and $\Esp[\|\bar{U}_t\|]$ is uniformly bounded. Thus, we deduce that $\bar{U}$ admits an invariant distribution and complete the proof.
\subsection{Proof}
\begin{proof}[Proof of Theorem \ref{lem:ExistenceProcHks}] Let us take $N^*$ and $U^*$ the processes  described in Remark \ref{Rem:Rem10} with $\nu = \sum_{e \in E} \delta_{e}$. For clarity, we forget the dependence of $\Esp_{\mathbf{x}}[.] $ on the initial condition $\mathbf{x} \in W_0$. 
\paragraph{Step (i):} In this step, we prove that the process $N$, defined by Equation \eqref{Eq:IntensityExpressionAll}, exists as a limiting process of the sequence $N^m$. To do so, we first introduce some notations. We define recursively the processes $\lambda^m$ and $N^m$ as in Remark \ref{Rem:Rem10}. Note that $U^{m} = ({Q^{m}}^1, {Q^{m}}^2,S^{m})$ can be decomposed in the following way:
\begin{equation}
\begin{array}{ll}
Q^{m\,i}_t = Q^{m\,i,+}_t - Q^{m\,i,-}_t, & {S^{m}_t} =  {S^{m}_t}^{+} - {S^{m}_t}^{-},\\
\end{array}
\label{Eq:LOBdecomp}
\end{equation}
with 
\begin{equation*}
\begin{array}{ll}
Q^{m\,i,+}_t = \sum_{T^{m} < t} \Delta Q^{m\,i}_t \mathbf{1}_{\Delta Q^{m\,i}_t > 0} , & Q^{m\,i,-}_t = \sum_{T^{m} < t} \Delta Q^{m\,i}_t \mathbf{1}_{\Delta Q^{m\,i}_t < 0},\\
\\
{S^{m}}^{+}_t = \sum_{T^{m} < t} \Delta {S^{m}}_t \mathbf{1}_{\Delta {S^{m}}_t > 0} , & {S^{m}}^{-}_t = \sum_{T^{m} < t} \Delta {S^{m}}_t \mathbf{1}_{\Delta {S^{m}}_t < 0}, \\
\end{array}
\end{equation*} 
with $i \in \mathbb{B} $ and $\Delta Z_t = Z_t - Z_{t^-}$ for any process $Z$. For all $\omega \in \Omega$, each one of the processes $N^m$, $\lambda^{m}$, $Q^{m\, i,\pm}$ and $S^{m\,\pm}$ is non decreasing with $m$ by induction. Hence, they admit limiting processes $N$, $\lambda$, ${Q^{1(2),\pm}}$ and ${S^{\pm}}$. This implies that $U^m $ converges towards $U$. To ensure that $N$ admits $\lambda$ as an intensity, we need to prove that $\sum_{e \in E}\lambda_t(e)$ and $U$ are both finite $a.s$, see Steps (ii)-(iii).  
\paragraph{Step (ii):} In this step, we prove by induction on $m$ that $\sup_{t} \Esp [\sum_{e \in E}\lambda^{m}_t(e)] $ is uniformly bounded which ensures that $\sup_{t} \Esp [\sum_{e \in E} \lambda_t(e)] $ is finite and that  $\sum_{e \in E} \lambda_t(e)$ does not explode. We write $\lambda^{m}_n(e,t) = \lambda^{m}_t(e)\mathbf{1}_{T^m_n < t \leq T^m_{n+1}}$. For $m=0$, we have $\Esp [ \lambda^{m}_n(t,e)] = 0$ since $\lambda^{m}_t(e)=0$ for any $t \geq 0$. We have by construction
\begin{equation*}
\begin{array}{ll}
\Esp [ \lambda^{m+1}_n(e,t)] = \Esp [ \lambda^{m}_n(e,t)], & \text{ when } n \leq m,\\
\Esp [ \lambda^{m+1}_{n}(e,t)] = 0, & \text{ when } n > m+1.
\end{array}
\end{equation*}
for any $t \geq 0$. Thus, we only need to study the case $n = m+1 $. Using Remark \ref{Rem:Rem10} and Assumption \ref{Assump:psigrowth}, we have
\begin{align*}
\sup_t\Esp [ \lambda^{m+1}_{m+1}(e,t)] & \leq c(e) + d(e) \sup_{t}\Esp [ \big(\sum_{T^m_i < t} \bar{\phi}(e,t- T^{m}_i,X^m_i)\big)^{n_{\psi}}] \\
						& = c(e) + d(e) \sup_t \sum_{\{\mathbf{k}_m\} \in \mathcal{P}(n_{\psi})} {{n_{\psi}}\choose{\mathbf{k}_k}} \sum_{\mathbf{x} \in E^k}\int_{(- \infty,t)^k} \prod_{i=1}^k\bar{\phi}^{k_i}(e,t-s_i ,x_i)\Esp [dN^m_{s_1}\,\ldots \,dN^m_{s_k}],
\end{align*}
with  $\bar{\phi}(e,t,x) = \sup_{u \in \mathbb{U}}\phi(e,u,t,x) $ and ${{n_{\psi}}\choose{\mathbf{k}_k}} = \frac{n_{\psi}!}{k_1!\ldots,k_k!} $. Using the above equation and the Brascamp-Lieb inequality, we have
\begin{align*}
\sup_{t }\Esp [ \lambda^{m+1}_{m+1}(e,t)] & \leq  c(e) + d(e) \sup_t \sum_{\mathbf{k}_k \in \mathcal{P}(n_{\psi})} {{n_{\psi}}\choose{\mathbf{k}_k}} \sum_{\mathbf{x} \in E^k}\int_{(- \infty,t)^k} \prod_{i=1}^k\bar{\phi}^{k_i}(e,t-s_i ,x_i)(\sup_{t,n}\Esp [\lambda^{m}_{n}(x_i,t)])^{1/k}\,ds_i, \\
						  & =  c(e) + \bar{\lambda}^m d(e) \sum_{\mathbf{k}_k \in \mathcal{P}(n_{\psi})} {{n_{\psi}}\choose{\mathbf{k}_k}} \sum_{\mathbf{x}  \in E^k}\int_{\mathbb{R}_+^k} \prod_{i=1}^k\bar{\phi}^{k_i}(e,s_i ,x_i)\,ds_i,\\
						  & \leq c(e) + q\bar{\lambda}^m,
\numberthis \label{Eq:BornIneq}
\end{align*}
with $\bar{\lambda}^m = \sup_{t,e,n} \Esp [ \lambda^{m}_n(e,t)] $ and $
q =  \sup_{e}\{d(e)\sum_{\mathbf{k}_k \in \mathcal{P}(n_{\psi})} {{n_{\psi}}\choose{\mathbf{k}_k}} \int_{\mathbb{R}_+^k} \prod_{i=1}^k {\phi^*}^{k_i}(e,s_i)\,ds_i\}
$ where $\phi^*$ is defined in Assumption \ref{Assump:psigrowth}.
Using (\ref{Eq:BornIneq}), we deduce that 
\begin{equation*}
\bar{\lambda}^{m+1} \leq \frac{c}{1-q} + q^{m+1} \bar{\lambda}^0 = \bar{x}.
\end{equation*}
Since $q < 1$ under Assumption \ref{Assump:psigrowth}, it ensures that $\bar{\lambda} =  \sup_{m} \bar{\lambda}^m$ is finite. To complete the proof, we use (\ref{Eq:BornIneq}) and Assumption \ref{Assump:Boundedness}, to get the following inequality:
\begin{equation*}
\sup_{t}\Esp [ \sum_{e \in E}\lambda^{m}_t(e)] \leq c^* + \bar{\lambda} \sum_{e\in E, \,\mathbf{k}_k \in \mathcal{P}(n_{\psi})} d(e){{n_{\psi}}\choose{\mathbf{k}_k}} \sum_{\mathbf{x}  \in E^k}\int_{\mathbb{R}_+^k} \prod_{i=1}^k\bar{\phi}^{k_i}(e,t-s_i ,x_i)\,ds_i < \infty.
\end{equation*}
\paragraph{Step (iii):} We write $U^{m}_n = (Q^{m \,1}_{n},Q^{m \,2}_{n},S^{m}_{n}) = U^{m}_{T_n} $. We prove here that $\Esp[{Q^{m\,i}_n}]$ and $\Esp [ {S^m_n}]  $ are uniformly bounded for all $m \geq 0$ and $n \geq 0$ to ensure that ${S}$ and ${Q^{i}} $ do not explode. Let us prove that 
\begin{equation}
\Esp\big[z^{Q^{m\,i}_{n+1}}\big]  \leq \lambda \Esp\big[z^{Q^{m\,i}_{n}}\big] + B, \quad \forall n \leq m-1,\, m \geq 1 .
\label{Eq:BoundQ}
\end{equation}
with $z \leq \min(z_0, z_1)$ and $z_0$ and $z_1$ are respectively defined in Assumption \ref{Assump:Negativeindividualdrift} and \ref{Assump:Boundedness}, $\lambda < 1$ and $B \geq 0$. Let $m \geq 1 $, we have by construction
\begin{equation*}
\begin{array}{ll}
\Esp [ z^{Q^{m+1\,i}_{n+1}}] = \Esp [ z^{Q^{m\,i}_{n+1}}], & \text{ when } n \leq m-1.
\end{array}
\end{equation*}
Thus, we only need to investigate the case $n = m $. This is proved in Lemma \ref{Lem:negDriftCondition}. Using Inequality (\ref{Eq:BoundQ}), we get  
\begin{align*}
\Esp\big[z^{Q^{m\,i}_{n}}\big]  \leq \frac{B}{1-\lambda} + \lambda^{n}z^{Q^{m \, i}_{0}}, \quad \forall n \leq m, \numberthis 
\end{align*}
with $z^{Q^{m \,i}_{0}}$ fixed. Thus, $\Esp\big[Q^{m\,i}_{n}\big]$ is uniformly bounded. Using similar lines of argument, we also have $\Esp [ S^m_n]$ uniformly bounded. Hence, the limiting processes $U$ does not explode. 
\paragraph{Step (iv):} First, note that the process $N$ is well defined since $\lambda_t$ is locally integrable, see Step (ii)-(iii) and \cite{jacod1975multivariate}. Additionally, we can construct it pathwise using the thinning algorithm, see Remark \ref{Rem:Rem10}.\\

Let $\bar{U}_s$ be the process described in Theorem \ref{lem:ExistenceProcHks} and for which we just proved the existence. This process is dominated by the process $\bar{U}^{\infty} = (U^{\infty},\lambda^{\infty}) = \esssup_{s \geq 0}\,\bar{U}_s$. In this part, we prove that both $\Esp[U^{\infty}]$ and $ \Esp[\lambda^\infty]$ are finite.\\

First, we prove that $\Esp[U^\infty] < \infty$. Let $\lambda < \rho < \rho^1< 1$, $C > 0$, $\mathcal{S}$ the set $\mathcal{S} = \{u \in \mathbb{U}; \, u > C, \, c.w.\}$ where $c.w$ means component-wise and $\mathpzc{S}$ a set $ \mathpzc{S} \in \mathcal{U} \subset \mathcal{S}$. Since $U_n\mathbf{1}_{U_n \in \mathpzc{S}^c}$ is bounded $a.s$, we only need to show $\Esp[U^{\infty,\,\mathpzc{S}}]$ is finite with $U^{\infty,\,\mathpzc{S}} = \esssup_{n \in \mathbb{N}}\,U_n^{\mathpzc{S}} $ and $ U_n^{\mathpzc{S}} = U_n \mathbf{1}_{U_n \in \mathpzc{S}}$. Using the Doob's decomposition, we have $U_n^{\mathpzc{S}} = M_n^{\mathpzc{S}} + A_n^{\mathpzc{S}}$ with $M_n^{\mathpzc{S}}$ a martingale and $A_n^{\mathpzc{S}} = \sum_{k = 1}^n \big(\Esp[U^{\mathpzc{S}}_k|\mathcal{F}_{k-1}] - U^{\mathpzc{S}}_{k-1} \big)$ a predictable process.  Thus, we get 
\begin{align*}
\Esp[U^{\infty,\mathpzc{S}}] \leq \Esp[\esssup_{n \geq 0} M_n^{\mathpzc{S}}] + \Esp[\esssup_{n \geq 0} A_n^{\mathpzc{S}} ],\quad c.w.
\end{align*} 
The Doob's inequality and Fatou's Lemma ensure that $\Esp[\sup_{n \geq 0} M_n^{\mathpzc{S}}] \leq 2 \underset{n\rightarrow \infty}{\lim} \Esp[{M_n^{\mathpzc{S}}}^{2}]^{\frac{1}{2}},\, c.w$. Using the martingale property of $M_n^{\mathpzc{S}}$ and the Doob's decomposition of $U_n^{\mathpzc{S}}$, we find 
\begin{align*}
\Esp[(M_n^{\mathpzc{S}})^{2}] - \Esp[(M_0^{\mathpzc{S}})^{2}] = \sum_{k=1}^n \Esp[(M^{\mathpzc{S}}_k - M_{k-1}^{\mathpzc{S}})^2], \qquad M^{\mathpzc{S}}_k - M^{\mathpzc{S}}_{k-1} = U^{\mathpzc{S}}_k - \Esp[U^{\mathpzc{S}}_k|\mathcal{F}_{k-1}],\, c.w.
\end{align*} 
We have 
\begin{align*}
\Esp[(M^{\mathpzc{S}}_k - M^{\mathpzc{S}}_{k-1})^2]  = \Esp[(U^{\mathpzc{S}}_k - \Esp[U^{\mathpzc{S}}_k|\mathcal{F}_{k-1}])^2] &  \leq 2 \big(\Esp[(U^{\mathpzc{S}}_k)^2] + \Esp[\Esp[(U^{\mathpzc{S}}_k)^2|\mathcal{F}_{k-1}]]\big) \leq 4 \Esp[(U^{\mathpzc{S}}_k)^2],\, c.w.
\end{align*}
Let us prove that $\sum_{k \geq 0} \Esp[(U^{\mathpzc{S}}_k)^2] < \infty $. Using Lemma \ref{Lem:negDriftCondition} and by taking $\mathcal{O} = \{(T_j,X_j)_{j \leq 0} \in W_0; \, X_j = (n_j,t_j,b_j,\tilde{u}_j,u_j,a_j) \in E \text{ and } u_0 \geq C, \, c.w.\}$, we have 
\begin{align*}
\Esp\big[V_C(U_{n+1})\mathbf{1}_{U_{n+1} \in \mathpzc{S}, U_{n} \in\mathpzc{S}} \big] & \leq \Esp\big[V_C(U_{n+1})\mathbf{1}_{U_{n} \in\mathpzc{S}} \big] \leq \lambda \Esp\big[V_C(U_{n})\mathbf{1}_{U_{n} \in \mathpzc{S}}\big] + \Esp\big[B\mathbf{1}_{U_{n} \in \mathpzc{S}}\big]. \numberthis \label{Eq:NoIdeaName_0}
\end{align*}
By following the same lines of arguments used to prove \eqref{Eq:DecompV_0} in Lemma \ref{Lem:negDriftCondition} and basic approximations, we have the following inequality:
\begin{align*}
\Esp\big[&z^{Q^i_{n+1}-C} \mathbf{1}_{\{U_{n+1} \in \mathpzc{S}, U_{n} \in\mathpzc{S}^c\}}\big]  \leq \Esp\big[z^{Q^i_{n}-C}\mathbf{1}_{\{U_{n+1} \in \mathpzc{S}, U_{n} \in\mathpzc{S}^c\}}\big]  +\Esp\big[\int_{\mathbb{R}_+} Z_n(t) \mathbf{1}_{\{U_{n+1} \in\mathpzc{S} ,U_{n} \in\mathpzc{S}^c \}} \left\{ \mathcal{Q}_{u}(t,U_{n})\right\}\, dt\big],
\end{align*}
In the set $\{U_{n} \in\mathpzc{S}^c\}$, we have $Q^i_n \leq C$ which implies $z^{Q^i_n  - C}<1$. Additionally, we have $\sum_{e \in E} \lambda_n(e,s + T_n) \geq \underline{\psi} > 0$ under Assumption \ref{Assump:Boundedness}. This ensures that $Z_n(t) \leq e^{- \psi t}$, $a.s$. Thus, using Assumption \ref{Assump:Boundedness}, there exists $B^1$ such that 
\begin{align}
\Esp\big[z^{Q^i_{n}-C}\mathbf{1}_{\{U_{n+1} \in \mathpzc{S}, U_{n} \in\mathpzc{S}^c\}}\big]  +\Esp\big[\int_{\mathbb{R}_+} Z_n(t) \mathbf{1}_{\{U_{n+1} \in\mathpzc{S} ,U_{n} \in\mathpzc{S}^c \}} \left\{ \mathcal{Q}_{u}(t,U_{n})\right\}\, dt\big] \leq \Esp[B^1\mathbf{1}_{\{U_{n+1} \in\mathpzc{S} ,U_{n} \in\mathpzc{S}^c \}}], \label{Eq:NoIdeaName_1}
\end{align}
We take $ C \geq C^* = \max(\log(\frac{2B}{\rho - \lambda} +1), \log(\frac{B^1}{1 - \rho^1}),C_{bound})$ to ensure that 
\begin{align*}
\left\{
\begin{array}{ll}
\big[B  - (\rho - \lambda)V_C(U_n) \big] \mathbf{1}_{U_{n} \in \mathpzc{S}} < 0, & \quad a.s.\\
\big[B^1 - (1 -\rho^1) V_C(U_n+1)  \big] \mathbf{1}_{U_{n+1} \in \mathpzc{S}} < 0, & \quad a.s.
\end{array}
\right.
\end{align*}
By combining Inequalities \eqref{Eq:NoIdeaName_0} and \eqref{Eq:NoIdeaName_1} and taking $C \geq C^* $, we deduce that 
\begin{align*}
\rho^1 \Esp\big[V_C(U_{n+1})\mathbf{1}_{U_{n+1} \in \mathpzc{S}} \big] & \leq \rho \Esp\big[V_C(U_{n})\mathbf{1}_{U_{n} \in \mathpzc{S}}\big] + \Esp\left[\big(B  - (\rho - \lambda)V_C(U_n) \big) \mathbf{1}_{U_{n} \in \mathpzc{S}}\right]  \\
                                 & \qquad \qquad + \Esp\left[\big(B^1 - (1 -\rho^1) V_C(U_{n+1})  \big) \mathbf{1}_{U_{n+1} \in \mathpzc{S}}\right], \\
         & \leq \rho \Esp\big[V_C(U_{n})\mathbf{1}_{U_{n} \in \mathpzc{S}}\big], 
\end{align*}
which ensures that $\Esp\big[V_C(U_{n+1})\mathbf{1}_{U_{n+1} \in \mathpzc{S}} \big] \leq r\Esp\big[V_C(U_{n})\mathbf{1}_{U_{n} \in \mathpzc{S}}\big]$ with $ r = \frac{\rho}{\rho^1} < 1$.
Since $(U^{\mathpzc{S}}_k)^2 \leq c_1 V_C(U^{\mathpzc{S}}_k) $, this proves that $\sum_{k \geq 0} \Esp[(U^{\mathpzc{S}}_k)^2] < c_1 \sum_{k \geq 0} \Esp[V_C(U^{\mathpzc{S}}_k)] \leq \frac{c_1}{1-\rho} < \infty $. Hence, we get $\Esp[\esssup_{n \geq 0} M^{\mathpzc{S}}_n ] \leq \big(\frac{c_1 }{1 - \rho}\big)^{\frac{1}{2}}, \, c.w$ .\\
We also have 
\begin{align*}
A^{\mathpzc{S}}_n \leq \tilde{A}^{\mathpzc{S}}_n  = \sum_{k=1}^n |\Esp[U^{\mathpzc{S}}_k|\mathcal{F}_{k-1}] - U^{\mathpzc{S}}_{k-1}| \leq 2 \sum_{k=1}^n \Esp[|U^{\mathpzc{S}}_k|],\, c.w,
\end{align*}
with $\tilde{A}^{\mathpzc{S}}_n$ a component-wise non-decreasing process. Since $\Esp[|U^{\mathpzc{S}}_k|] \leq \big(\Esp[(U^{\mathpzc{S}}_k)^2]\big)^{\frac{1}{2}} $, we get $\Esp[\tilde{A}_n ] \leq  \big(\frac{c_1 }{1 - \rho}\big)^{\frac{1}{2}}$. Hence, we deduce that $\Esp[\esssup_{n \geq 0} A^{\mathpzc{S}}_n ] \leq \big(\frac{c_1 }{1 - \rho}\big)^{\frac{1}{2}}, \, c.w$ which ensures that $\Esp[U^{\infty,\mathpzc{S}}] < \infty$.\\
Second, we prove that $\Esp[\lambda^{\infty}]$ is finite. Let $t\geq 0$ and $\mathcal{T}= \{t_0 = 0 < t_1 < \ldots < t_n = t \}$ be a partition of $[0,t]$. Using the monotone convergence theorem, we have  
\begin{align*}
\Esp[\sum_{k = 1}^n|\lambda_{t_k} - \lambda_{t_{k-1}}|] \leq  \Esp[\frac{\sum_{k = 1}^n(t_k - t_{k-1})}{t}|\tilde{\lambda}_{t_k} - \tilde{\lambda}_{t_{k-1}}|] = \Esp[\cfrac{\int_{0}^t f^{\mathcal{T}} \, \mathrm{d}s}{t} ] \leq \cfrac{\int_{0}^t \Esp[f^{\mathcal{T}}] \, \mathrm{d}s}{t},
\end{align*}
with $f^{\mathcal{T}} = \sum_{k=1}^n |\lambda_{t_k} - \lambda_{t_{k-1}}| \mathbf{1}_{t_{k-1} \leq t < t_{k}}$. Since $ \Esp[|\lambda_{t_k} - \lambda_{t_{k-1}}|] \leq 2 \sup_{t} \Esp [|\lambda_{t}|] \leq \frac{c}{1-q} < \infty$, we get 
\begin{align*}
\Esp[\sum_{k = 1}^n|\lambda_{t_k} - \lambda_{t_{k-1}}|]  \leq \frac{c}{1-q} < \infty.
\end{align*} 
We can then apply Bichteler-Dellacherie theorem to write $\lambda_{t} = M_s + A_s$ with $M_s$ a martingale and $A_s$ a predictable process with almost surely finite variation over finite time intervals such that 
\begin{align*}
\Esp [\mathbf{var}_t(\lambda) ] = \Esp[\mathbf{var}_t(M) ] + \Esp[\mathbf{var}_t(A)],
\end{align*}
where $\mathbf{var}_t(Z)$ is the variation of the process $Z$ over the interval $[0,t]$. Since 
\begin{align*}
\Esp[\lambda^{\infty}] \leq \Esp[\esssup_{s} M_s ] + \Esp[\esssup_{s} A_s ],\quad \esssup_{s \leq t} M_s \leq \mathbf{var}_t(M),\quad \esssup_{s\leq t} A_s \leq  \mathbf{var}_t(A),
\end{align*} 
and $\sup_{t} \Esp [\mathbf{var}_t(\lambda) ] < \infty$, we deduce that $ \Esp[\lambda^{\infty}] < \infty$. 
Finally, we have $\Esp [\|\bar{U}_t\|] \leq \Esp [\|\bar{U}^{\infty}\|] < \infty$, for all $t \geq 0$. Thus, the family $\cup_{t \geq 0} \bar{U}_t$ is tight. Moreover, the process $ \bar{U}_t$ satisfies the Feller property since $\mathbb{U}$ and $E$ are countable states and $\Esp[\|\bar{U}_t\|]$ is uniformly bounded. Thus the process $\bar{U}$ admits an invariant distribution $\mu$ which completes the proof. 
\end{proof}

\section{Proof of Theorem \ref{lem:ErgodicityProcHks}}
\label{sec:ErgValFctHks2}

\subsection{Preliminary result}
\begin{lem}
Let $(\mathcal{F}_n)_{n \geq 0}$ be a sequence of $\sigma$-algebras such that $\mathcal{F}_{n} \underset{n \rightarrow \infty}{\rightarrow} \mathcal{F}_{\infty}$ with $\mathcal{F}_{\infty}$ a $\sigma$-algebra and $(X_n)_{n \geq 0}$ be a sequence of random variables valued in $\mathbb{R}$ such that $X_n \underset{n \rightarrow \infty}{\rightarrow} X,$ $a.s$, $X_n$ is $\mathcal{F}_\infty$-measurable, $X$ is $\mathcal{F}_{\infty}$-measurable and $\sup_n\Esp[X_n^2] < \infty$. Then, we have 
$$
\Esp\big[X_n|\mathcal{F}_{n}\big] \underset{n \rightarrow \infty}{\rightarrow} X, \quad a.s.
$$
\label{lem:ConvergenceCondExpect}
\end{lem}
\begin{rem} In the above Lemma \ref{lem:ConvergenceCondExpect}, we can replace the condition $\sup_n\Esp[X_n^2] < \infty$ by the condition $\Esp[\sup_n X_n] < \infty$ and recover the same result.
\end{rem}
\begin{proof}[Proof of Lemma \ref{lem:ConvergenceCondExpect}] Let $m$ and $n$ be two positive integers. We write $X_n^m = \Esp[X_m|\mathcal{F}_n]$.
\paragraph{Step (i):} Since $\sup_n\Esp[X_n^2] < \infty$, we can apply a conditional dominated convergence theorem to show that $X_n^m \underset{m \rightarrow \infty}{\rightarrow} X_n = \Esp[X|\mathcal{F}_n],$ $a.s$. 
\paragraph{Step (ii):} Since $\mathcal{F}_{\infty} = \underset{n \rightarrow \infty}{\lim} \mathcal{F}_{n}$, there exists a sequence $(A_n)_{n \geq 0}$ such that $A_n \in \mathcal{F}_{n}$ and $A_n \underset{n \rightarrow \infty}{\rightarrow} A$. By definition, we have 
$$
\Esp[X_n \mathbf{1}_{A_n}] = \Esp[X \mathbf{1}_{A_n}].
$$
Note that the family $(X_n)_{n \geq 0}$ is tight. Indeed, using Doob's and Jensen's inequalities, we have 
\begin{align*}
\Esp \big[ \sup_{i \leq n} |X_i| \big] \leq \Esp \big[ \big(\sup_{i \leq n} |X_i|\big)^{2} \big]^{\frac{1}{2}} & = \Esp \big[ \sup_{i \leq n} X_i^{2} \big]^{\frac{1}{2}} \leq 2 \Esp\big[ X_n^{2} \big]^{\frac{1}{2}}.
\end{align*}
Then, using Fatou's Lemma, we get $\Esp \big[ \sup_{i \leq n} X_i\big] \leq 2 (\sup_n\Esp\big[ X_n^{2} \big])^{\frac{1}{2}} < \infty  $ which ensures that $(X_n)_{n\geq 0}$ is tight. Thus, we can extract a sub sequence $(X_{n_k})_{k \geq 0} $ such that $ X_{n_k} \underset{k \rightarrow \infty}{\rightarrow} Z$ $a.s$. Since $\sup_n\Esp[X_n^2] < \infty$, we can use the dominated convergence theorem to get
\begin{align*}
\Esp[Z\mathbf{1}_{A}] = \underset{k \rightarrow \infty}{\lim} \Esp[X_{n_k}\mathbf{1}_{A_{n_k}}] = \underset{k \rightarrow \infty}{\lim} \Esp[X\mathbf{1}_{A_{n_k}}] = \Esp[X\mathbf{1}_{A}].
\end{align*}
Thus, we have $Z = X,$ $\mathcal{F}_{\infty}-a.s.$ Since all the variables $X_k$ are $\mathcal{F}_{\infty}$-measurable, the variable $Z$ is also $\mathcal{F}_{\infty}$-measurable for any $n\geq 0$. Given that $Z$ and $X$ are both $\mathcal{F}_{\infty}$-measurable, we deduce that every accumulation point $Z$ of $(X_n)_{n \geq 0}$ satisfies $Z = X,$ $a.s$. Finally, we get $\lim_{\substack{m \rightarrow \infty\\n \rightarrow \infty}} X^m_n = X,$ $a.s.$ and we can use a composition argument, to deduce that $\Esp\big[X_n|\mathcal{F}_{n}\big] \underset{n \rightarrow \infty}{\rightarrow} X,$ $a.s.$
\end{proof}
We borrow the following definition from \cite{bremaud1996stability}.
\begin{defi}[Coupling] Two point processes $N$ and $N'$ couple if and only if 
$$
\underset{t  \rightarrow \infty }{\lim } \mathbb{P}\big[ N_s = N'_s, \quad \forall s \in (t, \infty) \big] = 1.
$$
\end{defi}
\begin{lem} Let $N$ be a point process and $\lambda$ its intensity. We have 
$$
\mathbb{P} [ N_s - N_t = 0, \quad \forall s \in (t, \infty)  | \mathcal{F}_t] = \Esp[e^{- \int_{t}^{\infty} \lambda_u \mathbf{1}_{A_u}  \, ds}| \mathcal{F}_t],
$$ 
with $A_u = \{N_u - N_t = 0 \} $ for all $u \geq t$.
\label{Lem:ComputPP}
\end{lem}
\begin{proof}
See Lemma 1 in \cite{bremaud1996stability}.
\end{proof}
\begin{lem}Two point processes $N$ and $N'$ which admit respectively $\lambda$ and $\lambda'$ as intensities couple if and only if 
$$
 \int_{0}^{\infty} \sup_{e \in E}\Esp\big[ |\lambda_s(e) - \lambda'_s(e)|\big]\, ds  < \infty.
$$
\label{Lem:CondIntensity}
\end{lem} 
\begin{proof} Let $\mathcal{F}_t = \mathcal{F}^N_t \vee \mathcal{F}^{N'}_t $. Using the canonical coupling, the point process $|N-N'|$ admits $|\lambda_t - \lambda'_t|$ as an $\mathcal{F}_t$-intensity. Using Lemma (\ref{Lem:ComputPP}) and Jensen's Inequality, we have 
$$
 \mathbb{P}\big[ \sup_{e}|N_s(e)-N'_s(e)| = 0, \quad \forall s \in (t, \infty) \big] \geq \Esp [e^{- \int_{t}^{\infty} \sup_{e}|\lambda_s(e) - \lambda'_s(e)|  \, ds}] \geq e^{- \int_{t}^{\infty} \sup_{e}\Esp [|\lambda_s(e) - \lambda'_s(e)|]  \, ds} .
$$
Since $\int_{0}^{\infty}\sup_{e} \Esp \big[ |\lambda_s(e) - \lambda'_s(e)|\big]\, ds   < \infty$, we have $ \int_{t}^{\infty} \sup_{e}\Esp \big[ |\lambda_s(e) - \lambda'_s(e)|\big]\, ds   \underset{t \rightarrow \infty}{\rightarrow} 0$ 	which implies that 
$$
\mathbb{P}\big[ \sup_{e}|N_s(e)-N'_s(e)| = 0, \quad \forall s \in (t, \infty) \big] \underset{t \rightarrow \infty}{\rightarrow} 1.
$$
This completes the proof.
\end{proof}
\subsection{Uniqueness}
\subsubsection{Outline of the proof}
Let $N^{\infty}= (T^{\infty}_i,X^{\infty}_i)$ be the stationary process constructed in Theorem \ref{lem:ExistenceProcHks} and $N = (T_i,X_i)$  be a point process whose intensity satisfies (\ref{Eq:IntensityExpressionHawkes}). We write $\lambda$ (resp. $\lambda^{\infty}$) for the intensity of $N$ (resp. $N^{\infty}$). To prove the uniqueness of the invariant distribution, we only need to show that $\int_{0}^{\infty} \sup_{e \in E}\Esp\big[ |\lambda_s(e) - \lambda^{\infty}_s(e)|\big]\, ds  < \infty$, see Lemma \ref{Lem:CondIntensity}. To do so, we first show that $(U_n)_{n \geq 0}$ is $f$-geometrically ergodic, see Lemma \ref{lem:BoundednesSum}. The proof of this result requires Lemmas \ref{lem:FiniteSumU_n} and \ref{lem:Weak_dependence_past_1}. Using this result, we prove, in Lemma \ref{Lemma:ErgodicityIncreasing}, that $f(t) =  \sup_e\Esp \big[|\lambda_t(e) - \lambda_t^{\infty}(e) | \big]$ satisfies the following inequality:
\begin{equation*}
f(t) \leq u(t) + c_3 G\big(\int_{0}^t \bar{h}(t-s)f(s) \, ds\big),
\end{equation*}
with $u(t) = c_2 \Esp[|| U_{t} - U^{\infty}_t ||] + c_1 \Esp \big[ \|U_{t} - U^{\infty}_{t}\|^{\beta p}\big]^{\frac{1}{\beta p}}$, $G(t) = t^{\frac{1}{\beta}}$ and $\bar{h}(t) = \sup_{e,u,x} \phi\big(e,u,t,x\big)$ with $c_1$, $c_2$, $c_3$, $\beta > 1$ and $p > 1$ positive constants. Then, we use Theorem 3 in \cite{beesack1984some} and the above inequality, to show that $\int_{\mathbb{R}_+} f(t)\, dt < \infty$ which ensures the uniqueness. 
\subsubsection{Proof}
Let $\lambda < 1$ given by Lemma \ref{Lem:negDriftCondition} and $\lambda < \rho < 1 $. We denote by $s = \{(T_j,X_j)_{j \leq 0} \in W_0; \, X_j = (n_j,t_j,b_j,\tilde{u}_j,u_j,a_j) \in E \text{ and } V(u_0) \leq \frac{2B}{\rho - \lambda} + 1\}$ and by $\alpha$ a set $\alpha \in \mathcal{W}_0 \subset s$. We have the following lemma.
\begin{lem}
Under Assumptions \ref{Assump:Negativeindividualdrift} and \ref{Assump:Boundedness}, the function $f = V +1$ with $V$ defined in Equation (\ref{Eq:V}) and $r> 1$ such that
\begin{equation*}
\sup_{\mathbf{x} \in W}\Esp_{\mathbf{x}} \big[ \sum_{n = 1}^{\tau_{\alpha}} f(U_n)r^n \big] <  \infty.
\end{equation*}
\label{lem:FiniteSumU_n}
\end{lem}
\begin{proof} The proof is similar to Theorem 6.3 in \cite{meyn1992stability}.  
\end{proof}
Let $\mathcal{F}_n$ and $\mathcal{F}_{l \leq j\leq n}$ be respectively defined in the following way $\mathcal{F}_n = \sigma\big(T_j \times X_j, \, \forall j \leq n\big)$,  $\mathcal{F}_{l \leq j\leq n} = \sigma\big(T_j \times X_j, \, \forall l \leq j \leq n\big)$. We also write $p^n_k$ as follows:
\begin{align*}
p^n_k(u) = |\mathbb{P}[U_{n} = u|\mathcal{F}_{k\leq j \leq n-1}\big] - \mathbb{P}[ U_{n} = u|\mathcal{F}_{j \leq n-1}]|, \quad \forall n \in \mathbb{N},\forall k \leq n-1, \forall u \in \mathbb{U}.
\end{align*}
\begin{lem} Under Assumptions  \ref{Assump:psigrowth}, \ref{Assump:Boundedness} and \ref{Assump:Regularityparams2}, we have 
\begin{align*}
p^n_k = \sup_{u \in \mathbb{U}} p^n_k(u) \underset{k \rightarrow \infty}{\rightarrow} 0,\quad a.s. \numberthis \label{Ineq:Weak_dependence_past_1}
\end{align*}
\label{lem:Weak_dependence_past_1}
\end{lem}
\begin{proof} Using Lemma \ref{Lem:ComputeDeltaU}, we have 
\begin{align*}
p^n_k(u) & =  | \Esp[ \int_{\mathbb{R}_+}Z_n(t) \lambda_{n} (u,t) \, dt|\mathcal{F}_{k\leq j \leq n-1}] - \Esp[ \int_{\mathbb{R}_+}Z_n(t) \lambda_{n} (u,t) \, dt |\mathcal{F}_{j \leq n-1}]| \\
	     & = | \Esp[\int_{\mathbb{R}_+}Z_n(t) \lambda_{n} (u,t) \, dt|\mathcal{F}_{k\leq j \leq n-1}] - \int_{\mathbb{R}_+}Z_n(t) \lambda_{n} (u,t) \, dt|,
\end{align*}
with $\lambda_{n}(u,t) = \sum_{e \in E(U_{n-1},u)} \lambda_n(e,t)$, $\lambda_n(e,t) = \psi(e,U_{n-1},t + T_{n-1}, r_n(t))$, $r_n(t) = \sum_{j\leq n-1} \phi(e,U_{n-1},t+ T_{n-1}-T_j,X_j)$ and $Z_n(t) = e^{-\big[\int_{0}^{t}\sum_e \lambda_n(e,s) \, ds\big]} $.\\
Since $p_k = \sup_{ u \in \mathbb{U}} p_k^n(u)$, we can construct a sequence $(u_j)_{j\geq 0}$ such that $p^{n}_k(u_j) \underset{j \rightarrow \infty}{\rightarrow} p^n_k,$ $a.s.$ We write $u_j = (q^1_j,q^2_j,s_j)$. Without loss of generality, we can consider that $(q^1_j)_{j\geq 0}$ is monotonic by taking a sub-sequence of $(q^1_j)_{j\geq 0}$. Hence, there exists a limiting process $q^1_{\infty} $ such that $q^1_j \underset{j \rightarrow \infty}{\rightarrow} q^1_{\infty},$ $a.s$. By repeating this argument several times, we can always construct $(u_j)$ such that 
\begin{align*}
p^{n}_k(u_j) \underset{j \rightarrow \infty}{\rightarrow} p_k, \quad u_j\underset{j \rightarrow \infty}{\rightarrow} u_{\infty}, \qquad a.s.
\end{align*}
Let us prove that $\lambda_{n}(u_j,t)\underset{j \rightarrow \infty}{\rightarrow} \lambda_{n}(u_{\infty},t) $, $a.s$. To do so, we distinguish two sets $A_1 = \{w \in \Omega ; \,u_{\infty}(w)< \infty\} $ and $A_2 = \{w \in \Omega ; \, u_{\infty}(w) = \infty\} $. When $u_{\infty}< \infty$, we have $u_j = u_{\infty}$ for $j$ large enough since $\mathbb{U}$ is countable. This ensures that $E(U_{n-1},u_j) = E(U_{n-1},u_\infty),$ $a.s$ for $j$ large enough. Thus, we get
\begin{align*}
\lambda_{n}(u_j,t) \mathbf{1}_{ A_1}\underset{j \rightarrow \infty}{\rightarrow} \sum_{e \in E(U_{n-1},u_\infty)} \psi(e,U_{n-1},t + T_{n-1}, r_{n}(t))\mathbf{1}_{ A_1}, \quad a.s.
\end{align*} 
When $u_{\infty} = \infty$, we have $\sum_{e \in E(U_{n-1},u_\infty)} \lambda_{n}(e,t) = 0$ since $E(U_{n_{\infty}-1},u_\infty) = \varnothing $. Using $\sum_{e \in E}\lambda_{n_{j}}(e,t) < \infty,$ $a.s$, see Step (ii) in the proof of Theorem \ref{lem:ExistenceProcHks}, we deduce that $\sum_{e \in E(U_{n_{j}-1},C^c)} \lambda_{n_{j}}(e,t) \underset{c \rightarrow \infty}{\rightarrow}  0$, $a.s$ with $C^c = \{u \in \mathbb{U}; \, u > c,\, c.w \} $, $c> 0$ and $c.w$ means component-wise. Since $ E(U_{n_{j}-1},u_j) \subset E(U_{n_{j}-1},C^c) $ for $j$ large enough, we get $\sum_{e \in E(U_{n_{j}-1},u_j)} \lambda_{n_{j}}(e,t) \underset{j \rightarrow \infty}{\rightarrow}  0,$ $a.s$ which means that 
\begin{align*}
\lambda_{n}(u_j,t) \mathbf{1}_{A_2}\underset{j \rightarrow \infty}{\rightarrow} \sum_{e \in E(U_{n-1},u_\infty)} \psi(e,U_{n-1},t + T_{n-1}, r_{n}(t))\mathbf{1}_{ A_2} = 0, \quad a.s,
\end{align*} 
and proves $\lambda_{n}(u_j,t)\underset{j \rightarrow \infty}{\rightarrow} \lambda_{n}(u_{\infty},t) $, $a.s$.\\
Additionally, we have $\Esp[\sup_{n,s} \sum_{e}\lambda_{n}(e,s)] < \infty$, see Step (iv) in the proof of Theorem \ref{lem:ExistenceProcHks}. Thus, we get $\Esp[\sup_{n,u,s}\lambda_{n}(u,s)] < \infty$. Since $\sum_{e}\lambda_n(e,s)\geq \underline{\psi}$ under Assumption \ref{Assump:Boundedness}, we have $ Z_n(t)\leq e^{- \underline{\psi}t},$ $a.s.$ Then, we can apply the dominated convergence theorem to show that 
$$
\int_{\mathbb{R}_+}Z_{n}(t) \lambda_{n} (t,u_j) \, dt \underset{j \rightarrow \infty}{\rightarrow} \int_{\mathbb{R}_+}Z_{n}(t) \lambda_{n} (t,u_{\infty}) \, dt,\quad a.s.
$$
Furthermore, we have
\begin{align*}
\Esp \big[ \sup_{j}\int_{\mathbb{R}_+}Z_{n}(t) \lambda_{n} (u_j,t) \, dt \big] & \leq 
\Esp \big[\int_{\mathbb{R}_+} \sup_{j}  e^{-\underline{\psi} t}\lambda_{n} (u_j,t) \, dt \big]   \overset{\text{Fubini}}{=} \int_{\mathbb{R}_+}  e^{-\underline{\psi} t}\Esp \big[\sup_{j}\lambda_{n} (u_j,t)\big] \, dt,
\end{align*}
with $\Esp \big[\sup_{j}\lambda_{n} (u_j,t)\big] < \infty$. Hence, we can use the conditional dominated convergence to show 
\begin{align*}
\Esp \big[\int_{\mathbb{R}_+}Z_{n}(t) \lambda_{n} (u_j,t) \, dt |\mathcal{F}_{k\leq r \leq n-1}\big] \underset{j \rightarrow \infty}{\rightarrow} \Esp \big[ \int_{\mathbb{R}_+}Z_{n}(t) \lambda_{n} (u_\infty,t) \, dt |\mathcal{F}_{k\leq r \leq n-1}\big].
\end{align*}
Finally, since $\mathcal{F}_{k\leq r \leq n-1} \underset{k \rightarrow \infty}{\rightarrow} \mathcal{F}_{r \leq n-1}$, we can apply Lemma \ref{lem:ConvergenceCondExpect} to deduce that 
\begin{align*}
p^n_k \underset{k \rightarrow \infty}{\rightarrow} 0,\, a.s.
\end{align*}
 This completes the proof.
\end{proof}
Let $\Delta T_n = T_n - T_{n-1}$ be the inter-arrival time between $n$-th jump and the $n-1$-th jump with $T_n$ the time of the $n$-th event. Let $N^{\infty}= (T^{\infty}_i,X^{\infty}_i)$ be the stationary process constructed in Lemma \ref{lem:ExistenceProcHks} and $N = (T_i,X_i)$  be a point process whose intensity satisfies (\ref{Eq:IntensityExpressionHawkes}). We write $U^{\infty} = ({Q^1}^{\infty},{Q^2}^{\infty},S^{\infty})$ (resp. $U = (Q^1,Q^2,S)$) for the order book state associated to $N^{\infty}$ (resp. $N$). We denote by $\lambda^{\infty}$ (resp. $\lambda$) the intensity of $N^{\infty}$ (resp. $N$). We have the following result.
\begin{lem} Under Assumptions  \ref{Assump:psigrowth}, \ref{Assump:Negativeindividualdrift}, \ref{Assump:Boundedness} and \ref{Assump:Regularityparams2}, the process $(U_n)_{n \geq 0}$ is $f$-geometrically ergodic, see 15.7 in \cite{meyn2012markov}, in the sense that there exists $r >1$ such that
$$
\sup_{\mathbf{x} \in W_0}\sum_{n \geq 1} \Esp_{\mathbf{x}} \big[\|f(U_n) - f(U^{\infty}_n)\|\, r^n \big] < \infty.
$$
\label{lem:BoundednesSum}
\end{lem}
\begin{proof}
Let $P^n(\mathbf{x},A)$ be the probability of being in the set $A = \{(t_k,x_k)_{k \leq 0} \in W_0; \, x_k = (n_k,t_k,b_k,s_k,\tilde{u}_k,u_k,a_k),\, u_0 \in a\}$, $ a\in \mathcal{U}$, with $\mathcal{U}$ the $\sigma$-algebra generated by the discrete topology on $\mathbb{U}$, after $n$ jumps conditional on $\mathbf{x} = (t_k,x_k)_{k \leq 0} \in W_0 = (\mathbb{R}_+ \times E)^{\mathbb{N}^-}$. Let $\mathbf{y} \in W_0$. We write $\pi$ for the stationary distribution of the process $U^{\infty}_n = ({Q^1}^{\infty}_n, {Q^2}^{\infty}_n, {S}^{\infty}_n)$ and $ \tau_{\alpha^k}$ for the first entrance time of $U$ to the set $\alpha^k = \{\mathbf{z} \in W_0; \, \mathbf{z}_{-k+1\leq j \leq 0} = \mathbf{y}_{-k+1\leq j \leq 0}\}$. Using the first-entrance last-exit decomposition of $P^n(\mathbf{x},A)$, see Section 8.2 in \cite{meyn2012markov}, we have

\begin{align*}
P^n(\mathbf{x},A) &= \, _{\alpha^k}P^{n}(\mathbf{x},A) + \sum_{j= 1}^n  \sum_{i=1}^j\big[ \int_{\mathbb{U}^{j - i}_{\mathbf{u},\alpha^k}}\int_{\mathbb{U}^{i}_{\mathbf{x},\alpha^k}} \, _{\alpha^k} P^{i}(\mathbf{x},\,du)P^{j-i}(\mathbf{u}, dv)\,_{\alpha^k}P^{n-j}(\mathbf{v},A) \big]\\
				& = \, _{\alpha^k}P^{n}(\mathbf{x},A) + \sum_{j= 1}^n \sum_{i=1}^j\big[ \int_{\mathbb{U}^{j - i}_{\mathbf{u},\alpha^k}}\int_{\mathbb{U}^{i}_{\mathbf{x},\alpha^k}} \, _{\alpha^k} P^{i}(\mathbf{x},\,du)P^{j-i}(\mathbf{u}, dv)\,_{\alpha^k}P^{n-j}(\mathbf{y},A) \big] \\
				& + \sum_{j= 1}^n \sum_{i=1}^j\big[ \int_{\mathbb{U}^{j - i}_{\mathbf{u},\alpha^k}}\int_{\mathbb{U}^{i}_{\mathbf{x},\alpha^k}} \, _{\alpha^k} P^{i}(\mathbf{x},\,du)P^{j-i}(\mathbf{u}, dv)\,|_{\alpha^k}P^{n-j}(\mathbf{v},A) - \,_{\alpha^k}P^{n-j}(\mathbf{y},A)|\big] .\numberthis \label{Eq:ProbDecomposition1}
\end{align*}

with $_{\alpha^k}P^{n}(\mathbf{x} ,A) = \mathbb{P}[(T_k,X_k)_{k \leq 0} =\mathbf{x},\, U_n \in A, \, \tau_{\alpha^k} \geq n ]$ and $\mathbb{U}^{i}_{\mathbf{x},\alpha^k} = \{\mathbf{z}\in \alpha^k\, ; \, (\mathbf{z}_{k})_{k \leq i} = \mathbf{x}\}$. Using $\Esp_{\mathbf{x}} [\tau_{\mathbf{\alpha^k}} ] < \infty$ for all $\mathbf{x} \in S$ and the arguments used in the proof of Theorem 10.2.1 in \cite{meyn2012markov}, we deduce that the stationary distribution admits the following representation:
\begin{align*}
\pi(A) & = \Esp_{\mathbf{y}} [\tau_{\mathbf{\alpha^k}} ]^{-1} \Esp_{\mathbf{y} } [\sum_{j=1}^{\tau_{\alpha^k}} \mathbf{1}_{\bar{U}_n \in A} ] = \pi(\alpha^k) \sum_{j=1}^{\infty}  {_{\alpha^k}}P^{j}(\mathbf{y} ,A). \numberthis \label{Eq:StationaryDistrib1}
\end{align*}
By combining (\ref{Eq:ProbDecomposition1}) and (\ref{Eq:StationaryDistrib1}), we get
\begin{align*}
P(\mathbf{x},A) - \pi(A) & = \, _{\alpha^k}P^{n}(\mathbf{x},A) + \bigg[\big(\,_{\alpha^k} P(\mathbf{x}) * P(\alpha^k) -  \pi(\alpha^k)\big)*\, _{\alpha^k} P(\mathbf{y})\bigg]_n(A) + \pi(\alpha^k) \sum_{j\geq n+1} {_{\alpha^k}}P^{j}(\mathbf{y},A) \\
								& + \big(\, _{\alpha^k} P(\mathbf{x}) * P(\alpha^k)\big)* \big(\,_{\alpha^k}P(\alpha^k) - \,_{\alpha^k} P(\mathbf{y})\big)_{n}(A). \numberthis \label{Eq:StationaryDistribErr1}
\end{align*}
with $*$  the integrated Cauchy product between two sequences which is defined as follows:
\begin{align*}
[u(B)*v(C)]_n(A) = \sum_{i = 1}^n\int_{\mathbb{U}^{i}_{B,C}} \, u_i(B,\,du)v_{n-i}(\mathbf{u}, A), \quad \forall (B,C,A) \in (\mathcal{W}_0)^3
\end{align*}
with $(u_n)_{n\geq 0}$ and $(v_n)_{n\geq 0}$ two sequences such that $u_n, v_n:(\mathcal{W}_0)^2 \rightarrow \mathbb{R} $. Let $f$ be the function defined in Lemma \ref{lem:FiniteSumU_n}, $\pi(f) =  \int_{\mathcal{U}} \pi(du)f(u) < \infty$, $\Esp_{\mathbf{x}}[f(U_n)] = \int_{\mathcal{U}} P^n(\mathbf{x},du)f(u)$ and $| P^n(\mathbf{x},.) - \pi|_f  = |\Esp_{\mathbf{x}}[f(U_n)] - \pi(f) | $. Using (\ref{Eq:StationaryDistribErr1}), we have 
\begin{align*}
| P^n(\mathbf{x},.) - \pi|_f  & \leq \Esp_{\mathbf{x}}[f(U_n) \mathbf{1}_{\tau_{\alpha^k} \geq n}] +  [ _{\alpha^k} P(\mathbf{x})* P(\alpha^k) -  \pi(\alpha^k)]* t^f_n \\
								  & + \pi(\alpha^k) \sum_{j \geq n+1}t^f_j + | _{\alpha^k} P(\mathbf{x}) * P(\alpha^k)* \Delta t^f_n |, \numberthis \label{Eq:ErrorTot1}
\end{align*}
with $ t^f_n = \Esp_{\mathbf{y}}[f(U_n) \mathbf{1}_{\tau_{\alpha^k} \geq n}]$ and $\Delta t^f_n(\mathbf{v}) = (\Esp_{\mathbf{v}}[f(U_n) \mathbf{1}_{\tau_{\alpha^k} \geq n}] -  t^f_n)$. To prove geometric ergodicity  we have to show 
\begin{align*}
\sup_{\mathbf{x}} \sum_{n \geq 1}| P^n(\mathbf{x},.) - \pi|_f r^n < \infty, \numberthis \label{Eq:GeomErgod}
\end{align*}
with $r > 1$. Let us take $\bar{n} \in \mathbb{N}^*$ and the delay $k(\bar{n}) \in \mathbb{N}$ associated to $\alpha^k$ depending on $\bar{n}$. Using (\ref{Eq:ErrorTot1}), we have 
\begin{align*}
\sum_{n\geq 1}^{\bar{n}}| P^n(\mathbf{x},.) - \pi|_f r^n & \leq \sum_{n\geq 1}^{\bar{n}} \Esp_{\mathbf{x}}[f(U_n) \mathbf{1}_{\tau_{\alpha^k} \geq n}]r^n + \sum_{n\geq 1}^{\bar{n}} |\big( _{\alpha^k} P(\mathbf{x})* P(\alpha^k) -  \pi(\alpha^k) \big)* t^f_n r^n | \\
										      & \quad + \pi(\alpha^k)\sum_{n\geq 1}^{\bar{n}}\sum_{j \geq n+1}t^f_jr^n + \sum_{n\geq 1}^{\bar{n}} [ _{\alpha^k} P(\mathbf{x}) * P(\alpha^k)]* \Delta t^f_n r^n = \text{ (i) + (ii) + (iii) + (iv) }.
\end{align*}
The error term (i) can be dominated by 
\begin{align*}
\sum_{n \geq 1}^{\bar{n}} \Esp_{\mathbf{x}}[f(U_n) \mathbf{1}_{\tau_{\alpha^k} \geq n}]r^n \leq \sum_{n \geq 1} \Esp_{\mathbf{x}}[f(U_n) \mathbf{1}_{\tau_{\alpha^k} \geq n}]r^n = \Esp_{\mathbf{x}} \big[\sum_{n = 1}^{\tau_{\alpha^k}} f(U_n) r^n \big].\numberthis \label{Ineq:ErrTermi}
\end{align*} 
The error term (iii) can be bounded by 
\begin{align*}
\pi(\alpha^k)\sum_{n \geq 1}^{\bar{n}}\sum_{j \geq n+1}t^f_j r^n \leq \pi(\alpha^k)\sum_{n \geq 1} \sum_{j \geq n+1}t^f_j r^n \leq \frac{\pi(\alpha^k)}{r-1} \sup_{\mathbf{v}}\Esp_{\mathbf{v}} \big[\sum_{n = 1}^{\tau_{\alpha^k}} f(U_n) r^n\big]. \numberthis \label{Ineq:ErrTermii}
\end{align*}
Now we move to the error term (iv). We have
\begin{align*}
\big( _{\alpha^k} P(\mathbf{x}) * P(\alpha^k) \big)* \Delta t^f_n \leq \sum_{j\leq n,\, i \leq j}[ \int_{\mathbb{U}^{j - i}_{\mathbf{u},\alpha^k}\times \mathbb{U}^{i}_{\mathbf{x},\alpha^k}\times W_0} \, _{\alpha^k} P^{i}(\mathbf{x},\,du)P^{j-i}(\mathbf{u}, dv)\big]\,\Delta_{\alpha^k}P^{n-j}(dw)f(w),
\end{align*}
with $\Delta_{\alpha^k}P^{n-j}(dw) = |\,_{\alpha^k}P^{n-j}(\mathbf{v},dw) - \,_{\alpha^k}P^{n-j}(\mathbf{y},dw)|$. Using Equations \eqref{Eq:ProbDecomposition1} and \eqref{Eq:StationaryDistrib1}, we get
\begin{align*}
\left\{
\begin{array}{l}
\Delta_{\alpha^k}P^{n-j}(dw) \leq \,_{\alpha^k}P^{n-j}(\mathbf{v},dw) + \,_{\alpha^k}P^{n-j}(\mathbf{y},dw) ,\\
\\
\sum_{j\leq n,\, i \leq j}[ \int_{\mathbb{U}^{j - i}_{\mathbf{u},\alpha^k}\times \mathbb{U}^{i}_{\mathbf{x},\alpha^k}\times\mathcal{W}} \, _{\alpha^k} P^{i}(\mathbf{x},\,du)P^{j-i}(\mathbf{u}, dv)\big]\,_{\alpha^k}P^{n-j}(\mathbf{v},dw)f(w) \leq  \Esp_{\mathbf{x}}\big[ f(U_n) r^n\big] < \infty,\\
\\
\sum_{j\leq n,\, i \leq j}[ \int_{\mathbb{U}^{j - i}_{\mathbf{u},\alpha^k}\times \mathbb{U}^{i}_{\mathbf{x},\alpha^k}\times\mathcal{W}} \, _{\alpha^k} P^{i}(\mathbf{x},\,du)P^{j-i}(\mathbf{u}, dv)\big]\,_{\alpha^k}P^{n-j}(\mathbf{y},dw)f(w) \leq  \Esp_{\mathbf{y}}\big[ f(U_{n-j}) r^{n}\big] < \infty,
\end{array}
\right.
\end{align*}
Since $\Delta_{\alpha^k}P^{n-j} \underset{k \rightarrow \infty}{\rightarrow} 0$, see Lemma \ref{lem:Weak_dependence_past_1}, the dominated convergence theorem ensures that 
\begin{align*}
\big( _{\alpha^k} P(\mathbf{x}) * P(\alpha^k) \big)* \Delta t^f_n \underset{k \rightarrow \infty}{\rightarrow} 0.
\end{align*}
Thus, there exists $\bar{k}(\bar{n})$ such that $\big( _{\alpha} P(\mathbf{x}) * P(\alpha) \big)* \Delta t^f_n \leq \epsilon(\bar{n})$  for any $k \geq  \bar{k}(\bar{n})$. Hence the error term (iv) can be majorated by 
\begin{align*}
\sum_{n\geq 1}^{\bar{n}} [ _{\alpha^k} P(\mathbf{x}) * P(\alpha^k)]* \Delta t^f_n r^n \leq \epsilon(\bar{n})\frac{r^{\bar{n}+1} - 1}{r - 1}, \numberthis \label{Ineq:Errorbarn}
\end{align*}
which means that we have to choose $\epsilon(\bar{n}) < c_1\frac{r - 1}{r^{\bar{n}+1} - 1}$ with $c_1$ a positive constant. Finally, using the property 
\begin{align*}
\underset{n \rightarrow \infty}{\lim} (u * v)_n = \underset{n \rightarrow \infty}{\lim} u_n \times \underset{n \rightarrow \infty}{\lim} v_n , \numberthis \label{Eq:PropCProd}
\end{align*}
we dominate the error term (ii) by 
\begin{align*}
 |\sum_{n \geq 1}^{\bar{n}}  \big(\,_{\alpha^k}P(\mathbf{x})*P(\alpha^k) -  \pi(\alpha^k)\big)* t^f_n r^n |& \leq   \big(\sum_{n \geq 1} |\,\big[_{\alpha^k}P(\mathbf{x}) * P(\alpha^k)\big]_n(\alpha^k) -  \pi(\alpha^k)|r^n \big)\sup_{\mathbf{v}}\Esp_{\mathbf{v}} \big[\sum_{n = 1}^{\tau_{\alpha^k}} f(U_n) r^n \big]. \numberthis \label{Ineq:ErrTermiii}
\end{align*}
Additionally, we have 
\begin{align*}
|\big[_{\alpha^k} P(\mathbf{x})* P(\alpha^k)\big]_{n}(\alpha^k) -  \pi(\alpha^k)| & = |\big[_{\alpha^k} P(\mathbf{x}) * (P(\alpha^k) -  \pi(\alpha^k))\big]_n(\alpha^k) - \pi(\alpha^k) \sum_{i \geq n+1} \,_{\alpha^k}P^{i}(\mathbf{x},\alpha^k)|\\
																		   & = |\big[\,_{\alpha^k} P(\mathbf{x}) * (P(\alpha^k) -  P(\mathbf{y}))\big]_n(\alpha^k) + \big[_{\alpha^k} P(\mathbf{x}) * (P(\mathbf{y}) -  \pi(\alpha^k))\big]_n(\alpha^k) \\
																		   & \quad \quad - \pi(\alpha^k) \sum_{i \geq n+1} \,_{\alpha^k}P^{i}(\mathbf{x},\alpha^k)| \\
																		   & \leq \, \big[_{\alpha^k} P(\mathbf{x}) * |P(\mathbf{y}) -  P(\alpha^k)|\big]_n(\alpha^k) + \big[\, _{\alpha^k} P(\mathbf{x})* |P(\mathbf{y}) - \pi(\alpha^k)|\big]_n(\alpha^k) \\
																		   & \quad + \pi(\alpha^k) \sum_{i \geq n+1} \,_{\alpha^k}P^{i}(\mathbf{x}^u,\alpha^k),
\end{align*}
for any $n \in \mathbb{N}$. Using Equation \eqref{Eq:PropCProd}, we get
\begin{align*}
\sum_{n \geq 1}^{\bar{n}} \big[|\,_{\alpha^k}P(\mathbf{x}) * P(\alpha^k) -  \pi(\alpha^k)&|\big]_n(\alpha^k)r^n  \leq \sum_{n \geq 1}^{\bar{n}}\big[\, _{\alpha^k} P(\mathbf{x})* |P^{\mathbf{y}}(\alpha^k) - P(\alpha^k)|\big]_n(\alpha^k) r^n \\
						& + \sum_{n \geq 1}^{\bar{n}}\big[\, _{\alpha^k} P(\mathbf{x})* |P(\mathbf{y}) - \pi(\alpha^k)|\big]_n(\alpha^k) r^n + \pi(\alpha^k) \sum_{\substack{n \geq 1\\i \geq n+1}} \,_{\alpha^k}P^{i}(\mathbf{x},\alpha^k)r^n \\
						& \leq \big(\sum_{n \geq 1}\, _{\alpha^k}P^{n}(\mathbf{x},\alpha^k) r^n \big) \big( \sum_{n \geq 1}\sup_{\mathbf{w} \in \alpha^k}|P^{n}(\mathbf{y},\alpha^k) -  P^n(\mathbf{w},\alpha^k)|r^n \big) \\ 
						&  + \big(\sum_{n \geq 1}\, _{\alpha^k}P^{n}(\mathbf{x},\alpha^k) r^n \big) \big( \sum_{n \geq 1}|P^{n}(\mathbf{y},\alpha^k) -  \pi(\alpha^k)|r^n \big) \\
						& + \pi(\alpha^k) \sum_{n \geq 1}\sum_{i \geq n+1} \,_{\alpha^k}P^{i}(\mathbf{x},\alpha^k)r^n = (1) + (2) + (3).
\end{align*}
The term $(2)$ is bounded by 
\begin{align*}				
(2)			& \leq \, \Esp_{\mathbf{x}} \big[r^{\tau_{\alpha}} \big] \sup_{\mathbf{y}}\big( \sum_{n \geq 1}|P^{n}(\mathbf{y},\alpha) -  \pi(\alpha)|r^n \big).
\end{align*}
Since the Kendall theorem ensures that $E_{\mathbf{x}} \big[r^{\tau_{\alpha^k}} \big] < \infty$ and $ \sum_{n \geq 1} |P^{n}(\mathbf{x},\alpha^k) -  \pi(\alpha^k)|r^n < \infty $ are equivalent, the quantity $(1)$ is finite if and only if $\sup_{\mathbf{v}}E_{\mathbf{v}} \big[r^{\tau_{\alpha^k}} \big] < \infty$. The term $(1)$ is majorated by 
\begin{align*}
(1)			& \leq \, \Esp_{\mathbf{x}} \big[r^{\tau_{\alpha^k}} \big] \big( \sum_{n \geq 1}\sup_{\mathbf{w}\in \alpha^k}|P^{n}(\mathbf{w},\alpha^k) -  P^n(\mathbf{y},\alpha^k)|r^n \big).
\end{align*}
To ensure that the sequence $v(\bar{n}) =  \sum_{n \geq 1}^{\bar{n}}\sup_{\mathbf{w}\in \alpha^k}|P^{n}(\mathbf{w},\alpha^k) -  P^n(\mathbf{y},\alpha^k)|r^n$ is bounded, the put a dependence $k$ and $\bar{n}$. Let $\epsilon^1(\bar{n}) > 0 $. By following the same arguments used in the proof of Inequality (\ref{Ineq:Errorbarn}), there exists $\bar{k}^1(\bar{n})$ such that for any $k \geq  \bar{k}^1(\bar{n})$, we have
\begin{align*}
\sum_{n \geq 1}^{\bar{n}}\sup_{\mathbf{w}\in \alpha^k}|P^{n}(\mathbf{y},\alpha^k) -  P^n(\mathbf{w},\alpha^k)|r^n \leq \epsilon^1(\bar{n}) \frac{r^{\bar{n} + 1} - 1}{r - 1}.
\end{align*} 
By taking $\epsilon^1(\bar{n}) \leq c_1\frac{ r - 1}{r^{\bar{n} + 1} - 1} $, we get $(1) \leq c_1\sup_{\mathbf{x}}\Esp_{\mathbf{x}} \big[r^{\tau_{\alpha^k}} \big]$. Furthermore, the term $(3)$ can be dominated by 
$(3) \leq \Esp_{\mathbf{x}} \big[r^{\tau_{\alpha^k}}\big]$. Thus, we deduce that 
\begin{align*}
(ii) \leq 
c_1 \Esp_{\mathbf{x}} \big[r^{\tau_{\alpha}} \big] (1 + \sup_{\mathbf{v}}\sum_{n \geq 1} |P^{n}(\mathbf{v},\alpha) -  \pi(\alpha)|r^n \big)\sup_{\mathbf{v}}\Esp_{\mathbf{v}} \big[\sum_{n = 1}^{\tau_{\alpha^k}} f(U_n) r^n \big]. \numberthis \label{Ineq:Errtermiii}
\end{align*}
By combining Inequalities (\ref{Ineq:ErrTermi}), (\ref{Ineq:ErrTermii}), (\ref{Ineq:ErrTermiii}) and (\ref{Ineq:Errtermiii}), we have (\ref{Eq:GeomErgod}) when $\sup_{\mathbf{x}}E_{\mathbf{x}} \big[\sum_{n = 1}^{\tau_{\alpha^k}} f(U_n) r^n \big]$ and $\sup_{\mathbf{x}}E_{\mathbf{x}} \big[r^{\tau_{\alpha^k}} \big] $ are both finite. Since $E_{\mathbf{x}} \big[\sum_{n = 1}^{\tau_{\alpha^k}} f(U_n) r^n \big] < \infty$ implies $E_{\mathbf{x}} \big[r^{\tau_{\alpha^k}} \big] < \infty $, we only need to prove
\begin{align*}
\Esp_{\mathbf{x}} \big[ \sum_{n = 1}^{\tau_{\alpha^k}} f(U_n)r^n \big] <  \infty.
\end{align*}
This last inequality is satisfied thanks to Lemma \ref{lem:FiniteSumU_n}.
\end{proof}
\begin{lem} Under Assumptions  \ref{Assump:psigrowth}, \ref{Assump:Boundedness} and \ref{Assump:Regularityparams2}, the process $\bar{U}$ is ergodic. 
\label{Lemma:ErgodicityIncreasing}
\end{lem}
\begin{proof}[Proof of Lemma \ref{Lemma:ErgodicityIncreasing}] For simplicity, we write $c_1$, $c_2$ and $c_3$  for positive constants and forget the dependence of $\Esp_{\mathbf{x}}[X] $ on the initial state $\mathbf{x}$ for any random variable $X$. Let $N^{\infty}= (T^{\infty}_i,X^{\infty}_i)$ be the stationary process constructed in Lemma \ref{lem:ExistenceProcHks} and $N = (T_i,X_i)$  be a point process whose intensity satisfies (\ref{Eq:IntensityExpressionHawkes}). We write $U^{\infty} = ({Q^1}^{\infty},{Q^2}^{\infty},S^{\infty})$ (resp. $U = (Q^1,Q^2,S)$) for the order book state associated to $N^{\infty}$ (resp. $N$). We denote by $\lambda^{\infty}$ (resp. $\lambda$) the intensity of $N^{\infty}$ (resp. $N$). To prove the uniqueness, we need to show that $N$ and $N^{\infty}$ couple which is satisfied when
\begin{equation*}
\int_{0}^{\infty} \sup_e\Esp \big[ |\lambda_t(e) - \lambda_t^{\infty}(e) |\big]\, dt  < \infty,
\end{equation*}
thanks to Lemma \ref{Lem:CondIntensity}. We write $f(t) =  \sup_e\Esp \big[|\lambda_t(e) - \lambda_t^{\infty}(e) | \big]$ for any $t \geq 0$.
\paragraph{Step (i):} For any $ \gamma = \frac{p}{q} >1$ with $p,q \in \mathbb{N}^*$ and $\beta$ such that $\frac{1}{\beta} + \frac{1}{\gamma} = 1$. Let us first prove that 
\begin{equation}
f(t) \leq u(t) + g_{1}(t) G\big(\int_{0}^t \bar{h}(t-s)f(s) \, ds\big),
\label{Eq:IneqTermF3}
\end{equation}
with $u(t) = c_3 \Esp[|| U_{t} - U^{\infty}_t ||] + c_1 \Esp \big[ \|U_{t} - U^{\infty}_{t}\|^{\beta p}\big]^{\frac{1}{\beta p}}[1 + 2B(t)],$  $g_{1}(t) =  c_2(1 + 2B(t))$, $G(t) = t^{\frac{1}{\beta}}$, $\bar{h}(t) = \sup_{e,u,x} \phi\big(e,u,t,x\big)$ and $B(t) = \sup_{0 \leq k \leq n_{\psi}-1}\bigg[ B_{k}(t)\bigg]^{\frac{1}{p\gamma}}$ with $ B_{k}(t)$ defined in Equation (\ref{Ineq:Termii2}). The quantities $c_1 $, $c_2$ and $c_3$ are positive constants. We have
\begin{align*}
 f(t) & = \Esp \big[|\psi(e,U_t,t,r_t) - \psi(e,U^{\infty}_t, t, r^{\infty}_t)| \big] \\
 				& \leq  \Esp \big[|\psi(e,U_t, t,r_t) - \psi(e,U^{\infty}_t, t,r_t) |\big] + \Esp \big[|\psi(e,U^{\infty}_t, t,r_t) - \psi(e,U^{\infty}_t, t,r^{\infty}_t) |\big] \\
 				& = (1) + (2),
\end{align*}
with $r_t = \int_{0}^t \phi(e, U_t, t- s, X_s) \, dN_s $ and $r^{\infty}_t = \int_{0}^t \phi(e, U^{\infty}_t, t- s, X^{\infty}_s) \, dN^{\infty}_s  $. Let us first handle the term (2). Using Assumption \ref{Assump:Regularityparams2}, we have
\begin{align*}
\Esp \big[ |\psi(e,U^{\infty}_t,t,r_t) - \psi(e,U^{\infty}_t, t, r^{\infty}_t) | \big]  & \leq  \Esp  \big[ |\bar{\psi}(r_t) - \bar{\psi}(r^{\infty}_t) | \big] \\
 					& \leq K\Esp \big[ |r_t - r^{\infty}_t | | 1 + r^{n_1}_t + r^{\infty^{n_1}}_t | \big] \\
 					& \leq K\overbrace{\Esp \big[ |r_t - r^{\infty}_t |^{\beta} \big]^{\frac{1}{\beta}} }^{(i)} \overbrace{\Esp \big[ |1 + r^{n_1}_t + r^{\infty^{n_1}}_t |^{\gamma} \big]^{\frac{1}{\gamma}}}^{(ii)}.
\end{align*}
The term (i) can be dominated by 
\begin{align*}
\Esp \big[ |r_t - r^{\infty}_t |^{\beta} \big]^{\frac{1}{\beta}} & \leq  \Esp \big[ |\int_{0}^{t} \phi(e,U_{t},t-s,X_s)dN_s -  \phi(e,U^{\infty}_{t},t-s, X_s) dN^{\infty}_s|^{\beta} \big]^{\frac{1}{\beta}} \\
			& \leq  2^{\frac{\beta - 1}{\beta}}\Esp \big[ |\int_{0}^{t} \phi(e,U_{t},t-s,X_s)dN_s -  \phi(e,U^{\infty}_{t},t-s, X_s) dN_s|^{\beta} \big]^{\frac{1}{\beta}} \\
			& +  2^{\frac{\beta - 1}{\beta}}\Esp \bigg[ \bigg(\int_{0}^{t} \bar{h}(t-s)\big|dN_s - dN^{\infty}_s\big|\bigg)^{\beta} \bigg]^{\frac{1}{\beta}}\\
			& \leq 2^{\frac{\beta - 1}{\beta}}\Esp \big[ \|U_{t} - U^{\infty}_{t}\|^{\beta}|\int_{0}^{t} \tilde{h}(e,t-s,X_s)dN_s|^{\beta} \big]^{\frac{1}{\beta}} + 2^{\frac{\beta - 1}{\beta}}\big[\int_{0}^t \bar{h}(t-s)f(s)\, ds\big]^{\frac{1}{\beta}} \\
			& \leq 2^{\frac{\beta - 1}{\beta}}\Esp \big[ \|U_{t} - U^{\infty}_{t}\|^{\beta p} \big]^{\frac{1}{\beta p}} \Esp \big[|\int_{0}^{t} \tilde{h}(e,t-s,X_s)dN_s|^{\beta q} \big]^{\frac{1}{\beta q}} + 2^{\frac{\beta - 1}{\beta}}\big[\int_{0}^t \bar{h}(t-s)f(s)\, ds\big]^{\frac{1}{\beta}} \\
			&  = c_1 \Esp \big[ \|U_{t} - U^{\infty}_{t}\|^{\beta p}\big]^{\frac{1}{\beta p}} + c_2 \big[\int_{0}^t \bar{h}(t-s)f(s)\, ds\big]^{\frac{1}{\beta}}, \numberthis \label{Ineq:Termi1}
\end{align*}
with $\bar{h}(s) = \sup_{e,u,x} \phi(e,u,s,x)$, $\tilde{h}(e,s,x) = \frac{2}{\min(\alpha_0,1)} \sup_{u} \phi(e,u,s,x)$ and $\min(\alpha_0,1)$ represents the minimum distance between two elements in the countable space $\mathbb{U}$. The quantity $\Esp \big[|\int_{0}^{t} \tilde{h}(e,t-s,X_s)dN_s|^{\beta q} \big]$ is bounded since 
\begin{align*}
\Esp \big[|\int_{0}^{t} \tilde{h}(e,t-s,X_s)dN_s|^{\beta q} \big] &\leq \Esp \big[|\int_{0}^{t} \tilde{h}(e,t-s,X_s)dN_s|^{ q} \big]^{\frac{1}{\beta}} \\
					& \leq \bigg\{\sum_{\mathbf{k}_{m} \in \mathcal{P}(q)} \sum_{\bar{x} \in E^m}{q\choose{\mathbf{k}_m}} \times\int_{(- \infty,t)^{m}}\Esp\big[ \prod_{i=1}^m  \tilde{h}(e,t-s_i,x_{i}) dN_{s_i}\big]\bigg\}^{\frac{1}{\beta}} \\
					& \leq \bigg\{\sum_{\mathbf{k}_{m} \in \mathcal{P}(q)} \sum_{\bar{x} \in E^m}{q\choose{\mathbf{k}_m}} \times\int_{(- \infty,t)^{m}} \prod_{i=1}^m  \tilde{h}(e,t-s_i,x_{i}) \Esp\big[ \lambda_{s_i}\big] ds_i\bigg\}^{\frac{1}{q}} < \infty. 
\end{align*} 
The term (ii) satisfies
\begin{align*}
\Esp \big[|1 + r^{n_1}_t + r^{\infty^{n_1}}|^{\gamma} \big]^{\frac{1}{\gamma}} & \leq  3^{\frac{\gamma - 1}{\gamma}} \big( 1 + \Esp[|r^{n_1}_t|^{\gamma}]^{\frac{1}{\gamma}} + \Esp[|r^{\infty^{n_1}}_t|^{\gamma}]^{\frac{1}{\gamma}}\big) \numberthis \label{Ineq:Termii1},
\end{align*}
with $\gamma = \frac{p}{q} $ and $p, q \in \mathbb{N}^*$. We have 
\begin{align*}
\Esp[|r^{n_1}_t|^{\frac{p}{q}}] & \leq \Esp[|r^{n_1}_t|^{p}]^{\frac{1}{q}} \\
															  & =  \Esp[\big(\int_{0}^{t} \phi(e,U_{t},t-s,X_s)dN_s\big)^{n_1p}]^{\frac{1}{q}} \\
								& = \left\{\sum_{\mathbf{k}_{m} \in \mathcal{P}(\bar{p})} \sum_{\bar{x} \in E^m}{\bar{p}\choose{\mathbf{k}_m}} \times\int_{(- \infty,t)^{m}}\Esp\left[ \prod_{i=1}^m  \bar{\phi}(t-s_i,x_{i}) dN_{s_i}\right]\right\}^{\frac{1}{q}},\numberthis \label{Ineq:Termii1a}
\end{align*}
with $\bar{\phi}(t,x) = \sup_{e,u} \phi(e,u,t,x) $ and $\bar{p} = n_1 p $. Using (\ref{Ineq:Termii1a}) and the Brascamp-Lieb inequality, we have
\begin{align*}
\Esp[|r^{n_1}_t|^{\frac{p}{q}}] & \leq  \left[ \sum_{\mathbf{k}_{m} \in \mathcal{P}(\bar{p})} \sum_{\bar{x} \in E^m}{\bar{p}\choose{\mathbf{k}_m}} \times \int_{(- \infty,t)^{m}}\prod_{i=1}^m  \bar{\phi}(t-s_i,x_{i}) \Esp\left[ \lambda_{s_i}\right]^{\frac{1}{m}}ds_{i} \right]^{\frac{1}{q}}\\
								& = \bigg[ \sum_{\mathbf{k}_{m} \in \mathcal{P}(\bar{p})} \sum_{\bar{x} \in E^m}{\bar{p}\choose{\mathbf{k}_m}} R_{m}(t)\bigg]^{\frac{1}{q}} = B_{k}(t)^{\frac{1}{q}}, \numberthis \label{Ineq:Termii2}
\end{align*}
with $R_{m}(t) = \int_{(- \infty,t)^{m}}\prod_{i=1}^m  \bar{\phi}(t-s_i,x_{i}) \Esp\left[ \lambda_{s_i}\right]^{\frac{1}{m+m'}}ds_{i}$ and $B_{k}(t) = \sum_{\mathbf{k}_{m} \in \mathcal{P}(\bar{p})} \sum_{\bar{x} \in E^m}{\bar{p}\choose{\mathbf{k}_m}} R_{m}(t)$. Similarly, we also have
\begin{align*}
\Esp[|r^{\infty^{n_1}}_t|^{\frac{p}{q}}] & \leq B_{k}(t)^{\frac{1}{q}}. \numberthis \label{Ineq:Termii21}
\end{align*}
Using Inequalities (\ref{Ineq:Termii1}) and (\ref{Ineq:Termii2}), we deduce that (ii) verifies 
\begin{equation}
\Esp \big[|1 + r^{n_1}_t + r^{\infty^{n_1}}|^{\gamma} \big]^{\frac{1}{\gamma}} \leq  3^{\frac{\gamma - 1}{\gamma}} (1 + 2 \sup_{0 \leq k \leq n_{\psi}-1}\bigg[ B_{k}(t)\bigg]^{\frac{1}{q\gamma}}).
\label{Eq:IneqTermiiF}
\end{equation}
By combining inequalities (\ref{Ineq:Termi1}) and (\ref{Eq:IneqTermiiF}), we deduce that 
\begin{equation}
(2) \leq 3^{\frac{\gamma - 1}{\gamma}}  \bigg[ c_1 \Esp \big[ \|U_{t} - U^{\infty}_{t}\|^{\beta p}\big]^{\frac{1}{\beta p}} + c_2 \big[\int_{0}^t g(t-s)f(s)\, ds\big]^{\frac{1}{\beta}}\bigg] \bigg[1  + 2 \sup_{0 \leq k \leq n_{\psi}-1}\bigg[ B_{k}(t)\bigg]^{\frac{1}{q\gamma}}\bigg].
\label{Eq:IneqTermF1}
\end{equation}
Using Theorem \ref{lem:ExistenceProcHks}, we have $\sup_{e,t}\Esp[\sup_{u}\psi(e,u,t,r_t)]$ is finite. Thus, there exists $K$ such that  
\begin{equation}
(1) \leq c_3 \Esp[|| U_{t} - U^{\infty}_t ||].
\label{Eq:IneqTermF2}
\end{equation}
Thus using Equations (\ref{Eq:IneqTermF1}) and (\ref{Eq:IneqTermF2}), we prove (\ref{Eq:IneqTermF3}).
\paragraph{Step (ii):} By a density argument, there exist continous sequences of functions $(u^p)_{p \geq 1}$, $(g^p_1)_{p \geq 1}$ and $(\bar{h}^p)_{p \geq 1} $ such that $u^p (t) \underset{p \rightarrow \infty}{\rightarrow} u(t)$ and $u \leq u^p$, $g_1^p (t) \underset{p \rightarrow \infty}{\rightarrow} g_1(t)$ and $g_1 \leq g_1^p$ and $\bar{h}^p  \underset{p \rightarrow \infty}{\overset{L^1}{\rightarrow}} \bar{h}$ and $\bar{h} \leq \bar{h}^p$. Thus, we have 
\begin{align*}
f(t) \leq u^p(t) + g^p_1(t)G\big(\int_{0}^t \bar{h}^p(s) f(s)\, ds\big).
\end{align*}
Using a density argument again, we can find a sequence of functions $(f^k)_{k \geq 1}$ converges uniformly towards $f$. By affording ourselves to use sub-sequences, we can always consider that 
\begin{align*}
f^p(t) \leq \tilde{u}^p(t) + g^p_1(t)G\big(\int_{0}^t \bar{h}^p(s) f^p(s)\, ds\big),
\end{align*}
with $\tilde{u}^p(t) = u^p(t) + |f-f^p|_{\infty}$. Using Theorem 3 in \cite{beesack1984some} and Inequality (\ref{Eq:IneqTermF3}), we have 
\begin{align*}
f^p(t) \leq v^p(t)F^p(t)\left\{1 + G\bigg[H^{-1}\big(\int_{0}^t \bar{h}^p(s) g^p_1(s)\, ds\big) \bigg] \right\},
\end{align*}
with $H(s) = \int_{0}^s \frac{dt}{1 + G(t)}$, $v^p(t) = \max(G_1(\tilde{u}^p)(t),1)$, $ F^p(t)= \max(G_1(g_1^p)(t),1)$ and 
\begin{align*}
G_1(w)(t) = w(t) \left(1 + \int_0^t w(s) \bar{h}^p(s) e^{\int_{s}^t \bar{h}^p g^p_1\,du}\, ds \right).
\end{align*}
By sending $p$ to infinity, we deduce that  
\begin{align}
f(t) \leq v(t)F(t)\left\{1 + G\bigg[H^{-1}\big(\int_{0}^t \bar{h}(s) g_1(s)\, ds\big) \bigg] \right\},
\label{Ineq:Finalf}
\end{align} 
with  $v(t) = \max(G_1(u)(t),1)$ and $ F(t)= \max(G_1(g_1)(t),1)$.
\paragraph{Step (iii):} Let us prove that $\int_{\mathbb{R}_+} u(t)\, dt < \infty $. Since $B(t)$ is uniformly bounded, we only need to prove that 
\begin{equation*}
\left\{
\begin{array}{l}
\int_{\mathbb{R}_+} \Esp \big[\| U_{t} - U^{\infty}_t\|\big] \, dt < \infty\\
\int_{\mathbb{R}_+} \Esp \big[ \|U_{t} - U^{\infty}_{t}\|^{\beta p}\big]^{\frac{1}{\beta p}}\, dt < \infty.
\end{array}
\right.
\end{equation*}
Since $0 < \underline{\psi} = \inf_{u,t,r}\sup_{e}\psi(e,u,t,r) \leq \lambda_n$, we have 
\begin{align*}
\Esp \big[ \int_{\mathbb{R}_+} \| U_{t} - U^{\infty}_t\| \, dt \big] & = \leq  \Esp \big[ \sum_{n \geq 0} \| U_{n} - U^{\infty}_{n}\|\int_{T_n}^{T_{n+1}} \, dt \big]  \\
													               & = \Esp \big[ \sum_{n \geq 0} \|U_n - U^{\infty}_{n}\| \Esp[T_{n+1}-  T_n | \mathcal{F}_n] \big] \\
                      &  \leq \Esp \big[ \sum_{n \geq 0} \|U_n - U^{\infty}_{n}\| \frac{1}{\underline{\psi}} \big].
\end{align*}
Using Lemma \ref{lem:BoundednesSum}, we have $\Esp \big[ \sum_{n \geq 0} \|U_n - U^{\infty}_{n}\| \big] < \infty$ which ensures that $\Esp \big[ \int_{\mathbb{R}_+} \| U_{t} - U^{\infty}_t\| \, dt \big] < \infty $. By using a similar methodology and the fact that $ \sum_{n \geq 0} \Esp \big[\|U_n - U^{\infty}_{n}\|^{\beta p} \big] r^n < \infty $ with $r> 1$, see Lemma \ref{lem:BoundednesSum}, we also have $\int_{\mathbb{R}_+} \Esp \big[ \|U_{t} - U^{\infty}_{t}\|^{\beta p}\big]^{\frac{1}{\beta p}}\, dt < \infty$.
\paragraph{Step (iv):} Since $g_1$ is bounded and $\int_{0}^t \bar{h}(s) \, ds < \infty$, the functions $F(t)$ and \newline $\left\{1 + G\bigg[H^{-1}\big(\int_{0}^t \bar{h}(s) g_1(s)\, ds\big) \bigg] \right\}$ are bounded as well. Moreover, $\int_{\mathbb{R}_+} u(t)\, dt < \infty  $ thanks to the previous step. Thus, by applying Inequality (\ref{Ineq:Finalf}), we have that $ \int_{\mathbb{R}_+} f(t)\, dt < \infty $ which completes the proof.
\end{proof}
\subsection{Speed of convergence}
\begin{lem}We have the following error estimate:
\begin{equation*}
||P_t(w,.) - \bar{\pi}||_{TV} \leq K_1 e^{- K_2 t}, \quad \forall w \in W,
\end{equation*}
with $K_3 > 0$ and $K_2 > 0$.
\label{Lemma:SpeedofConvergence}
\end{lem}
\begin{proof}[Proof of Lemma \ref{Lemma:SpeedofConvergence}] We forget the dependence of $\Esp_{\mathbf{x}}[X]$ on the initial state $\mathbf{x}$ for any random variable $X$. We have 
\begin{align*}
||P_t(w,.) - \bar{\pi}||_{TV} & \leq \mathbb{P}[\sup_{e}|N_s - N^{\infty}_s| \ne 0,  \quad \forall s \in (t,\infty)] \\
							 & = \left(1 - \mathbb{P}[\sup_{e}N_s = N^{\infty}_s,  \quad \forall s \in (t,\infty)]\right) = (i).
\end{align*}
Using Lemma \ref{Lem:ComputPP} and Jensen's Inequality, we have 
\begin{align*}
(i) \leq 1 - e^{-\int_{t}^{\infty} f(s) \, ds},
\end{align*}
with $f(t) =  \sup_e\Esp \big[|\lambda_t(e) - \lambda_t^{\infty}(e) | \big]$ for any $t \geq 0$. Using Inequality (\ref{Ineq:Finalf}) and the boundedness of $F$ and $\left\{1 + G\bigg[H^{-1}\big(\int_{0}^t \bar{h}(s) g_1(s)\, ds\big) \bigg] \right\}$, we have 
\begin{align}
(i) \leq c_1\int_{t}^{\infty} u(t) \, dt,
\label{Ineq:SpeedConvergence1}
\end{align}
with $c_1$ a positive constant. Let us now prove that 
\begin{align}
u(t) \leq c_1 e^{- \alpha t},
\label{Ineq:SpeedConvergence1_error1}
\end{align}
with $\alpha$ a positive constant. We have 
\begin{align*}
\Esp[\|U_t - U^{\infty}_t\|] = \Esp[\|U_{N(t)} - U^{\infty}_{N^{\infty}(t)} \|] & \leq \Esp[\|U_{N(t)} - U^{\infty}_{N(t)} \|] +\Esp[\|U^{\infty}_{N(t)} - U^{\infty}_{N^{\infty}(t)}\|]. 
\end{align*}
Using the fact that $\sum_{n \geq 1} \Esp[\|U_n - U^{\infty}_n \|]r^n < \infty$, there exists $\alpha > 0$ such that  $\Esp[\|U_n - U^{\infty}_n \|] \leq Ae^{-\alpha n}$. Let us denote by $U^{\infty, \delta}_t$ the $\delta $-translated process defined such that $U^{\infty, \delta}_t = U^{\infty}_{t + \delta} $. By applying Lemma \ref{Lemma:ErgodicityIncreasing} to the process $U^{\infty, \delta}$, we also have $\sup_{\delta}\big(\sum_{n \geq 1} \Esp[\|U^{\infty,\delta}_n - U^{\infty}_n \|]r^n\big) < \infty$ which ensures that $\Esp[\|U^{\infty, \delta}_n - U^{\infty}_n \|] \leq Ae^{-\alpha n} $. Using Lemma \ref{Lem:TransTimeEvtTimereal} below and the uniqueness of the stationary distribution, we have $\frac{N(t)}{t} \underset{t \rightarrow}{\rightarrow} \frac{1}{\Esp_{\pi}[\Delta T_1]}$ and $\frac{N^{\infty}(t)}{t} \underset{t \rightarrow}{\rightarrow} \frac{1}{\Esp_{\pi}[\Delta T_1]},$ $a.s$. Thus, we deduce that 
\begin{align}
\Esp[\|U_t - U^{\infty}_t\|] \leq c_1 e^{-\alpha t}.
\label{Ineq:SpeedConvergence1_sub}
\end{align}
Using the same lines of argument, we also have 
\begin{align}
\Esp[\|U_t - U^{\infty}_t\|^{\beta p}]^{\frac{1}{\beta p}} \leq c_1 e^{-\alpha t}.
\label{Ineq:SpeedConvergence2_sub}
\end{align}
By combining Inequalities (\ref{Ineq:SpeedConvergence1_sub}) and (\ref{Ineq:SpeedConvergence2_sub}) and using the expression of $u(t)$, we recover Inequality (\ref{Ineq:SpeedConvergence1_error1}) which ensures that 
\begin{align*}
(i) \leq c_1 e^{-\alpha t}.
\end{align*}
This completes the proof.
\end{proof}
\begin{lem} For any initial state $u \in \mathbb{U}$, the process $\Delta T_n$ satisfies
\begin{align*}
\frac{\sum_{i = 1}^n \Delta T_i}{n} \underset{n \rightarrow \infty}{\rightarrow} \Esp_{\mu}[\Delta T_1] \qquad a.s,
\end{align*} 
with $\mu$ the unique stationary distribution of the point process $N$.
\label{Lem:TransTimeEvtTimereal}
\end{lem}
\begin{proof}
Since there exists $\underline{\lambda}> 0$ such that $\inf_{t,u,r} \sum_{e \in E} \lambda_t(e,u,r) > \underline{\lambda}$, we have $\Esp [\Delta T_n ] \leq \frac{1}{\underline{\lambda}}$ for any $n \geq 1$. Thus, $\Delta T_n$ admits a finite stationary distribution. Using the Theorem 17.1.2 in \cite{meyn2012markov}, we complete the proof.
\end{proof}

\section{Proof of Propositions \ref{Prop:LimitThMrkv} and \ref{Prop:LimitThMrkv2_1}}
\label{sec:Prooflth}
\begin{proof}[Proof of Proposition \ref{Prop:LimitThMrkv}]
The proof of Equation (\ref{Eq:CvgResult1}) is a direct application of Theorem 2 in \cite{jldoob}.\\
Since $(U_n)$ is $f$-geometrically ergodic, see Lemma \ref{lem:BoundednesSum}, $(Y_n)$ is $g$-geometrically ergodic and $U_n$ and $Y_n$ are independent, the process $(U_n,Y_n)$ is $\tilde{f}$-geometrically ergodic with $\tilde{f}(u,y) = f(u) + g(y)$. Let $g$ and $h$ be two functions such that $g^2, h^2 \leq \tilde{f}$, $\mu$ the stationary distribution of $(U,Y)$ and $\bar{v} = v - \Esp_{\mu} [v]$ for any function $v$. By following the same lines of argument of Lemma 16.1.5 in \cite{meyn2012markov}, we have 
\begin{align*}
 | \Esp_{\pi} [\bar{h}(Z_{n})\bar{g}(Z_{n+k}) ] | & \leq R \Esp_{\pi} [\bar{\tilde{f}}(Z_0)] r^k,
\end{align*}
with $Z_n =(U_n,Y_n)$, $r < 1$ and $R$ a positive constant. The quantity $\Esp_{\pi} [\bar{\tilde{f}}(Z_0)]$ is bounded by Lemma \ref{Lem:negDriftCondition}. Thus $Z$ is a geometric mixing and Theorems 19.1 and 19.2 in \cite{billing} give the result.
\end{proof}

\begin{proof}[Proof of Proposition \ref{Prop:LimitThMrkv2_1}]
Using Lemma \ref{Lem:TransTimeEvtTimereal} and Proposition \ref{Prop:LimitThMrkv2_1}, the proof of this result is analagous to the proof of Theorem 4.2 in \cite{huang2017ergodicity}.
\end{proof}
\section{Stationary distribution computation}
\label{sec:statdistrib}
\begin{proof}[Proof of Proposition \ref{Prop:GeneralProbStat}] Let $z \in Z$ and $z' \in Z$ such that $z \ne z'$. Since $\zeta$ is stationary under $\mu$, we have
\begin{align}
\sum_{z' \in Z} \int_{A^{z'}}\mu(dw) P_t(w,A^z) = \int_{W_0} \mu(dw) P_t(w,A^z)  =  \mu(A^z), \quad \forall t \geq 0,
\end{align}
with $P_t(w,.)$ the probability distribution of $\zeta^{0,w}_t$ starting from the initial condition $w$ and $A^{z} = \{(w_s)_{s \leq 0} \in W_0; \, \zeta^{0,w}_0 = z\}$. Since $\int_{A^{z'}}\mu(dw) P_t(w,A^z) = \mathbb{P}_{\mu} [\zeta_t = z,\, \zeta_0 = z'] = \mathbb{P}_{\mu} [ \zeta_0 = z']\mathbb{P}_{\mu} [ \zeta_t = z|\zeta_0 = z']$ and $\mu(A^z) =  \mathbb{P}_{\mu} [ \zeta_0 = z]$, the quantity $\pi(z) = \mu(A^z)$ defined in Section \ref{subsec:Stationaryprob} satisfies 
\begin{align*}
\sum_{z' \in Z} \pi(z') \mathbb{P}_{\mu} [ \zeta_t = z|\zeta_0 = z']  =  \pi(z), \quad \forall t \geq 0, 
\end{align*}
which also leads to the following equation:
\begin{align*}
\sum_{z' \in Z} \pi(z') \tilde{Q}(z,z') = 0,\quad \sum_{z' \in Z} \pi(z') = 1,
\end{align*}
with $\tilde{Q}(z,z') = \underset{\delta \rightarrow 0}{\lim} \cfrac{\mathbb{P}_{\mu} [ \zeta_{\delta} = z' | \zeta_{0} = z ]}{\delta}$. The quantity $\tilde{Q}(z,z')$ satisfies
\begin{align*}
\tilde{Q}(z,z') = \underset{\delta \rightarrow 0}{\lim} \cfrac{\mathbb{P}_{\mu} [ U_{\delta} = z' | U_{0} = z ]}{\delta} & = \underset{\delta \rightarrow 0}{\lim} \cfrac{\mathbb{P}_{\mu} [ \{T_{1} \leq \delta, \, e_1 \in E(z,z') \} | \zeta_{0} = z ] + \epsilon}{\delta} \\
		& = \underset{\delta \rightarrow 0}{\lim} \cfrac{\Esp_{\mu} [ \mathbb{P} [ \{T_{1} \leq \delta, \, e_1 \in E(z,z') \} | \mathcal{F}_0]| \zeta_{0} = z ] + \epsilon}{\delta}\\
		& = \Esp_{\mu}[\underset{\delta \rightarrow 0}{\lim}\cfrac{ \mathbb{P} [ \{T_{1} \leq \delta, \, e_1 \in E(z,z') \} | \mathcal{F}_0]}{\delta}| \zeta_{0} = z ] + \underset{\delta \rightarrow 0}{\lim}\cfrac{\epsilon}{\delta}\\
				& = \Esp_{\mu}[ \sum_{e_1 \in E(z,z')} \lambda_{0}(e_1)| \zeta_{0} = z ] + \underset{\delta \rightarrow 0}{\lim}\cfrac{\epsilon}{\delta}\\
		& = \sum_{e_1 \in E(z,z')} \Esp_{\mu}[\lambda_{0}(e_1)| \zeta_{0} = z] + \underset{\delta \rightarrow 0}{\lim}\cfrac{\epsilon}{\delta} ,
\end{align*}
where $\epsilon$ is an error term associated to the cases when at least two events happen in the interval $[0,\delta]$. Since $\sum_{e_1 \in E}\Esp_{\mu}[\lambda_0 (e_1)]$ is finite, we have $\epsilon \leq c_1 \delta^2 $ with $c_1$ a positive constant. We deduce that
\begin{align}
\tilde{Q}(z,z') = \sum_{e_1 \in E(z,z')} \Esp_{\mu}[\lambda_{0}(e_1,u)] = Q(z,z'). \label{Eq:proportQtildeQ}
\end{align}
This completes the proof.
\end{proof}
\section{Proof of Proposition \ref{Prop:Estimateintens}}
\label{sec:AppendixLawLargeNumbersNtdivt}
\begin{proof}[Proof of Proposition \ref{Prop:Estimateintens}] We write $\lambda^{u,u'}_s = \sum_{e \in E(u,u')} \lambda_s(e)$ and $E(u,u')$ the set of events that moves the order book from the state $u$ to $u'$. We have 
\begin{align}
\frac{N^{u,u'}_t}{t} = \frac{\int_{0}^t \lambda_s \delta^{s}_{u,u'} \, ds}{t} + \big(\frac{N^{u,u'}_t - \int_0^t \lambda_s \delta^{s}_{u,u'} \, ds}{t} \big). \label{Eq:Ntdivt1}
\end{align}
Since $(\lambda_s)_{s \geq 0}$ is stationary under $\bar{\pi}$ and $\Esp_{\bar{\pi}}[\lambda_s] < \infty$, the Theorem 2.1-chapter X in \cite{jldoob} ensures that 
\begin{align*}
\frac{\int_{0}^t \lambda_s \delta^{s}_{u,u'}  \, ds}{t} \underset{t \rightarrow \infty}{\rightarrow} \Esp_{\bar{\pi}}[\lambda_0 \delta^{0}_{u,u'} ] = \sum_{e \in E (u,u')} \Esp_{\bar{\pi}}[\lambda_0(e)\delta^{0}_{u,u'}]  & = \sum_{e \in E (u,u')} \Esp_{\bar{\pi}}[\lambda_0(e)\delta^{0}_{u,u'}|U_0 = u] \mathbb{P}_{\bar{\pi}}[U_0 = u]\\
				   & = \mathbb{P}_{\bar{\pi}}[U_0 = u] \sum_{e \in E (u,u')} \Esp_{\bar{\pi}}[\lambda_0(e)|U_0 = u] \\
				   & = \mathbb{P}_{\bar{\pi}}[U_0 = u]  Q(u,u') , \quad a.s. \numberthis \label{Eq:Ntdivt2}
\end{align*}
Moreover, since $N^{u,u'}_t - \int_0^t \lambda^{u,u'}_s \, ds$ is a martingale and $\sup_{s \geq 0,u,u'}\Esp[\lambda^{u,u'}_s ]< \infty$, we have 
\begin{align}
\frac{N^{u,u'}_t - \int_0^t \lambda_s \delta^{s}_{u,u'} \, ds}{t}  \underset{t \rightarrow \infty}{\rightarrow} 0, \quad a.s. \label{Eq:Ntdivt3}
\end{align}
Hence, by combining (\ref{Eq:Ntdivt1}), (\ref{Eq:Ntdivt2})  and (\ref{Eq:Ntdivt3}), we prove $\cfrac{N^{u,u'}_t}{t} \underset{t \rightarrow \infty}{\rightarrow} \mathbb{P}_{\bar{\pi}}[U_0 = u]  Q(u,u') ,$ $a.s.$ On the other hand, we have 
\begin{align}
\frac{t^u}{t} = \frac{\int_{0}^t \delta^{s}_{u} \, ds}{t}. \label{Eq:Ntdivt4_}
\end{align}
Since $(U_s)_{s \geq 0}$ is stationary under $\bar{\pi}$ and $\Esp_{\bar{\pi}}[\delta^{s}_{u}] < \infty$, the Theorem 2.1-chapter X in \cite{jldoob} ensures that 
\begin{align*}
\frac{\int_{0}^t \delta^{s}_{u}  \, ds}{t} \underset{t \rightarrow \infty}{\rightarrow} \Esp_{\bar{\pi}}[ \delta^{0}_{u} ] = \mathbb{P}_{\bar{\pi}}[U_0 = u] \quad a.s. \numberthis \label{Eq:Ntdivt5_}
\end{align*}
Thus, we deduce that 
\begin{align*}
\cfrac{N^{u,u'}_t}{t^u} = \cfrac{\cfrac{N^{u,u'}_t}{t}}{\cfrac{t^u}{t}} \underset{t \rightarrow \infty}{\rightarrow} Q(u,u'), \quad a.s,
\end{align*}
which completes the proof.
\end{proof}

\begin{proof}[Proof of confidence interval computation] By applying Theorem \ref{Prop:LimitThMrkv} to the sequence of $\eta_s = \lambda_s \delta^s_{u,u'}$ and use basic inequalities to approximate $t$ by its integer part $\lfloor t \rfloor $, we have 
\begin{align*}
\sqrt{t}\big(\cfrac{N^{u,u'}_t}{t} - \mathbb{P}_{\bar{\pi}}[U_0 = u]  Q(u,u')\big)\overset{\mathcal{L}}{\rightarrow} \sigma^1 W_1, \numberthis \label{Eq:Ntdivterror_1_}
\end{align*}
with $\sigma^2_1 = \Esp_{\mu} [(\lambda_0\delta^0_{u,u'})^2] + 2 \sum_{k \geq 1} \Esp_{\mu} [\lambda_0\delta^0_{u,u'}\lambda_k\delta^k_{u,u'}]$ and $W_t$ a standard brownian motion. Similarly, by using the same arguments, we also have
\begin{align*}
\sqrt{t}\big(\cfrac{t^u}{t} - \mathbb{P}_{\bar{\pi}}[U_0 = u] \big)\overset{\mathcal{L}}{\rightarrow} \sigma^2 W_1, \numberthis \label{Eq:Ntdivterror_2_}
\end{align*}
with $\sigma^2_2 = \Esp_{\mu} [(\delta^0_{u})^2] + 2 \sum_{k \geq 1} \Esp_{\mu} [\delta^0_{u}\delta^k_{u}]$. Using (\ref{Eq:Ntdivterror_1_}) and (\ref{Eq:Ntdivterror_2_}), we have with asymptotic probability $95\%$ that 
\begin{align*}
\begin{array}{lcl}
 \mathbb{P}_{\bar{\pi}}[U_0 = u]  Q(u,u')& \in & [\cfrac{N^{u,u'}_t}{t} +\cfrac{1.96 \sigma_1}{\sqrt{t}}, \cfrac{N^{u,u'}_t}{t} -\cfrac{1.96 \sigma_1}{\sqrt{t}}  ]\\
\mathbb{P}_{\bar{\pi}}[U_0 = u]^{-1}  & \in & [\cfrac{t}{t^u} +\cfrac{1.96 \sigma_2}{\sqrt{t}}\times \cfrac{t}{t^u}, \cfrac{t}{t^u} -\cfrac{1.96 \sigma_2}{\sqrt{t}} \times \cfrac{t}{t^u} ]. \\
\end{array}
\numberthis \label{Eq:Ntdivterror_3_}
\end{align*}
Equation (\ref{Eq:Ntdivterror_3_}) ensures that we have with probability $90\% $
\begin{align*}
Q(u,u') \in [(\cfrac{N^{u,u'}_t}{t} +\cfrac{1.96 \sigma_1}{\sqrt{t}})(\cfrac{t}{t^u} +\cfrac{1.96 \sigma_2}{\sqrt{t}}\times \cfrac{t}{t^u}), (\cfrac{N^{u,u'}_t}{t} -\cfrac{1.96 \sigma_1}{\sqrt{t}})(\cfrac{t}{t^u} -\cfrac{1.96 \sigma_2}{\sqrt{t}} \times \cfrac{t}{t^u}) ].
\end{align*}
\end{proof}

\section{Proof of Remark \ref{Rem:Vol_ratio_cons_insert}}
\label{sec:proofRemVol_ratio_cons_insert}
\begin{proof}We assume that the insertion (resp. consumption) intensity $\lambda^{+}$ (resp. $\lambda^{-}$) is constant and focus on the best bid limit $Q^1$. The stationary distribution $\pi^{old}$ of $Q^1$ verifies
\begin{equation}
\pi^{old}(q)  =  \pi^{old}(0) (\rho^{old})^q,\quad
\pi^{old}(0)  =  ( 1 + \sum_{q=1}^{\infty}(\rho^{old})^q)^{-1},\quad
\rho^{old}  =  \cfrac{\lambda^{+}  }{\lambda^{-} },
\label{Eq:Ex0}
\end{equation}
with $q \geq 1$ the size of $Q^1$. We add to the market a new agent whose insertion (resp. consumption) intensity $\lambda^{+,a}$ (resp. $\lambda^{-,a}$) is also constant. The stationary distribution $\pi^{new}$  of $Q^1$ in the new market satisfies 
\begin{equation}
\pi^{new}(q)  =  \pi^{new}(0) (\rho^{new})^q ,\quad
\pi^{new}(0)  =  ( 1 + \sum_{q=1}^{\infty}(\rho^{new})^q )^{-1},\quad
\rho^{new}  =  \cfrac{\lambda^{+} + \lambda^{+,a}}{\lambda^{-} + \lambda^{-,a}},
\label{Eq:Ex1}
\end{equation}
with $q \geq 1$ the size of $Q^1$. Using Equations (\ref{Eq:Ex0}) and (\ref{Eq:Ex1}), we can write 
\begin{equation}
\rho^{new}  =  \rho^{old} (1 + R(\mathbf{\lambda},\mathbf{\lambda}^{a})),\quad
\pi^{new}(0)  =  \big( 1 + \sum_{q=1}^{\infty}(\rho^{old})^q(1 + R (\mathbf{\lambda},\mathbf{\lambda}^{a}))^q \big)^{-1},
\label{Eq:Ex00}
\end{equation} 
with $\mathbf{\lambda} = (\lambda^+,\lambda^-)$, $\mathbf{\lambda}^a = (\lambda^{+,a},\lambda^{-,a}) $ and $R(\mathbf{\lambda},\mathbf{\lambda}^{a}) = (1 + \frac{\lambda^{+,a}}{\lambda^{+}})/(1 + \frac{\lambda^{-,a}}{\lambda^{-}}) - 1$. We want the new introduced agent to reduce the volatility of the old market which at the first order reads
\begin{equation}
\Esp_{\pi^{new}} [\eta_0^{2}] \leq \Esp_{\pi^{old}} [\eta_0^{2}].
\label{Eq:IneqIntuition1}
\end{equation}
Using Equation (\ref{Eq:Ex00}), we can reformulate Inequality (\ref{Eq:IneqIntuition1}) in the following way:
\begin{equation}
\sum_q \frac{(\rho^{old})^q\bigg(1 + R(\mathbf{\lambda},\mathbf{\lambda}^{a})\bigg)^q}{\bigg(1 + \sum_{j=1}^{\infty} (\rho^{old})^j (1 + R(\mathbf{\lambda},\mathbf{\lambda}^{a}))^j \bigg)} \eta_0^{2}(q)  \leq \sum_q \frac{(\rho^{old})^q}{\bigg( 1 + \sum_{j=1}^{\infty}(\rho^{old})^j \bigg)}\eta_0^{2}(q), 
\label{Eq:IneqIntuition2}
\end{equation}
for any function $\eta_0$. To satisify Inequality (\ref{Eq:IneqIntuition2}) we need $
R(\mathbf{\lambda},\mathbf{\lambda}^{a}) \geq 0 $ which leads to
\begin{align*}
\cfrac{\lambda^{+,a}}{\lambda^{-,a}} \geq \cfrac{\lambda^{+}}{\lambda^{-}},
\end{align*}
This condition is a well-known result which ensures that the new agent needs to have an insertion/consumption ratio greater than the one of the market.  
\end{proof}

\section{Supplementary numerical results}
\label{sec: appendix stationary measure}

The three next figures show the liquidity consumption and provision intensities at the first limit relative to the whole market according to the queue size, the corresponding stationary measure and the long term volatility, respectively for EssilorLuxottica, Michelin and Orange. 

\begin{figure}[H]

\centering
    \hfill (a) Intensity of the market \hfill (b) Stationary measure $Q^1$ \hfill ~\\
        \includegraphics[width=0.42\linewidth]{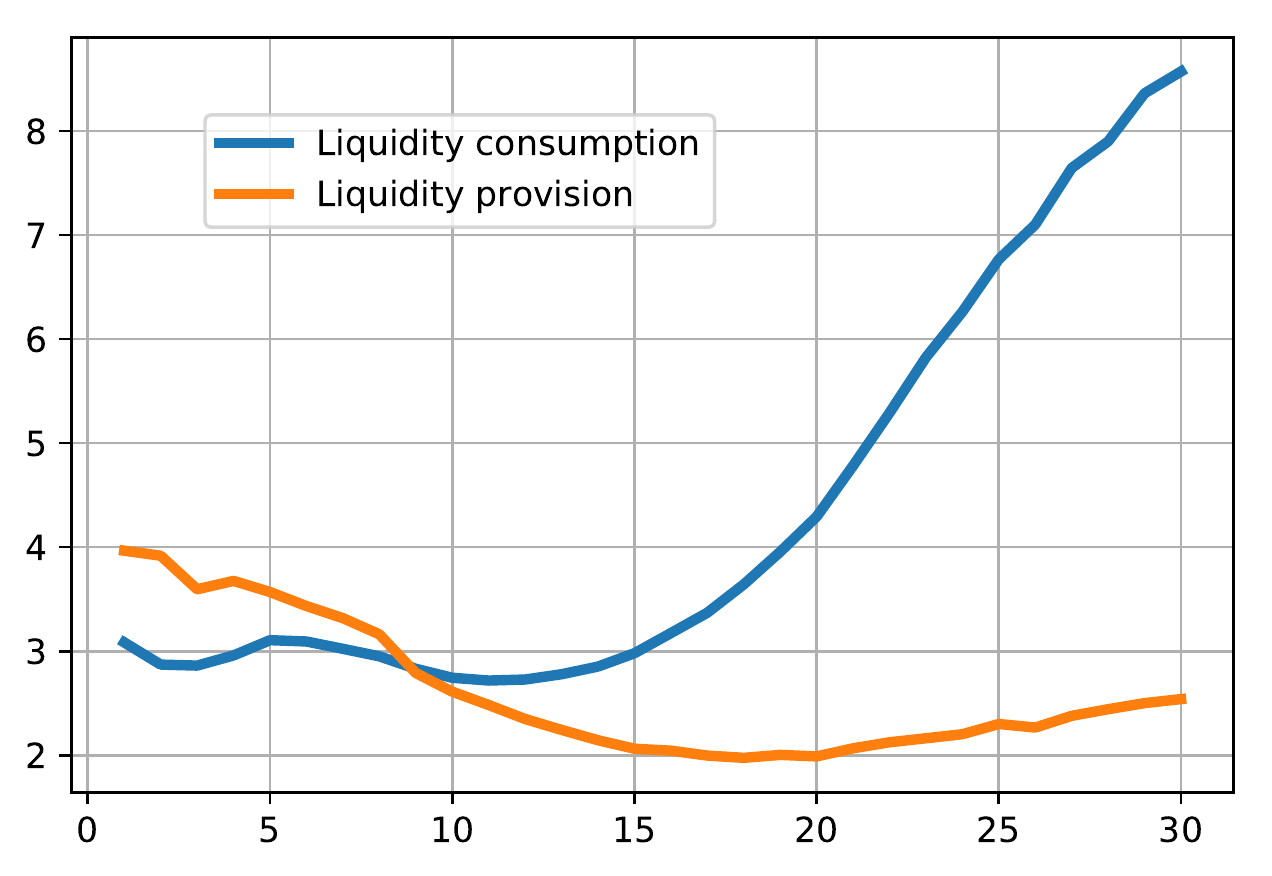}
	    \includegraphics[width=0.47\linewidth]{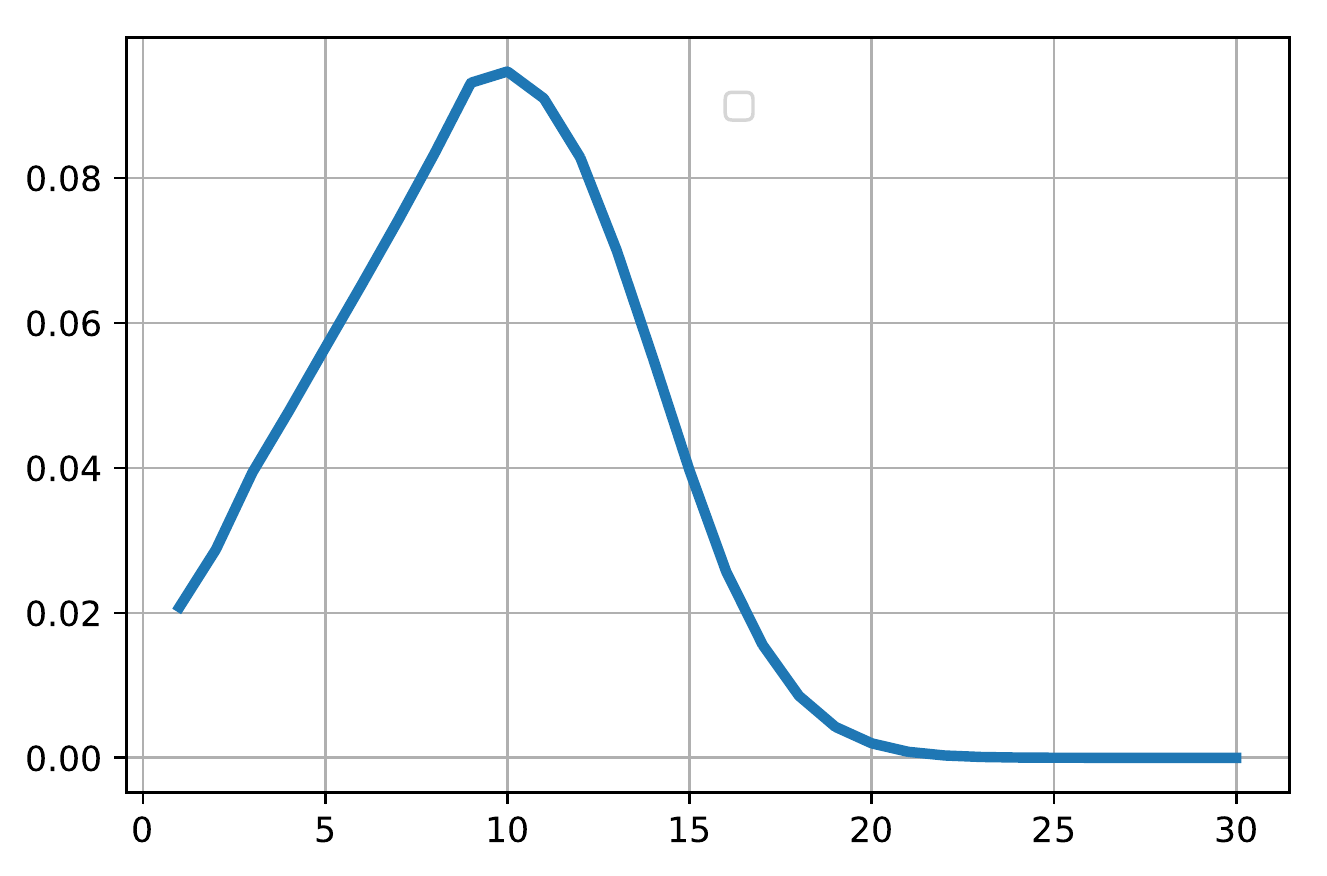}\\
	    \hfill Long term price volatility $ \sigma^{2,G} =  0.038$, $\sigma^{2,M}_{10} = 0.26$. \hfill ~\\
	    \caption{(a) Liquidity insertion and consumption intensities (in orders per second) with respect to the queue size (in AES) and (b) the corresponding stationary distribution of $Q^1$ with respect to the queue size (in AES), proper to ExilorLuxottica.}
	    \label{Exilor}
\end{figure}

\begin{figure}[H]

\centering
    \hfill (a) Intensity of the market \hfill (b) Stationary measure $Q^1$ \hfill ~\\
        \includegraphics[width=0.42\linewidth]{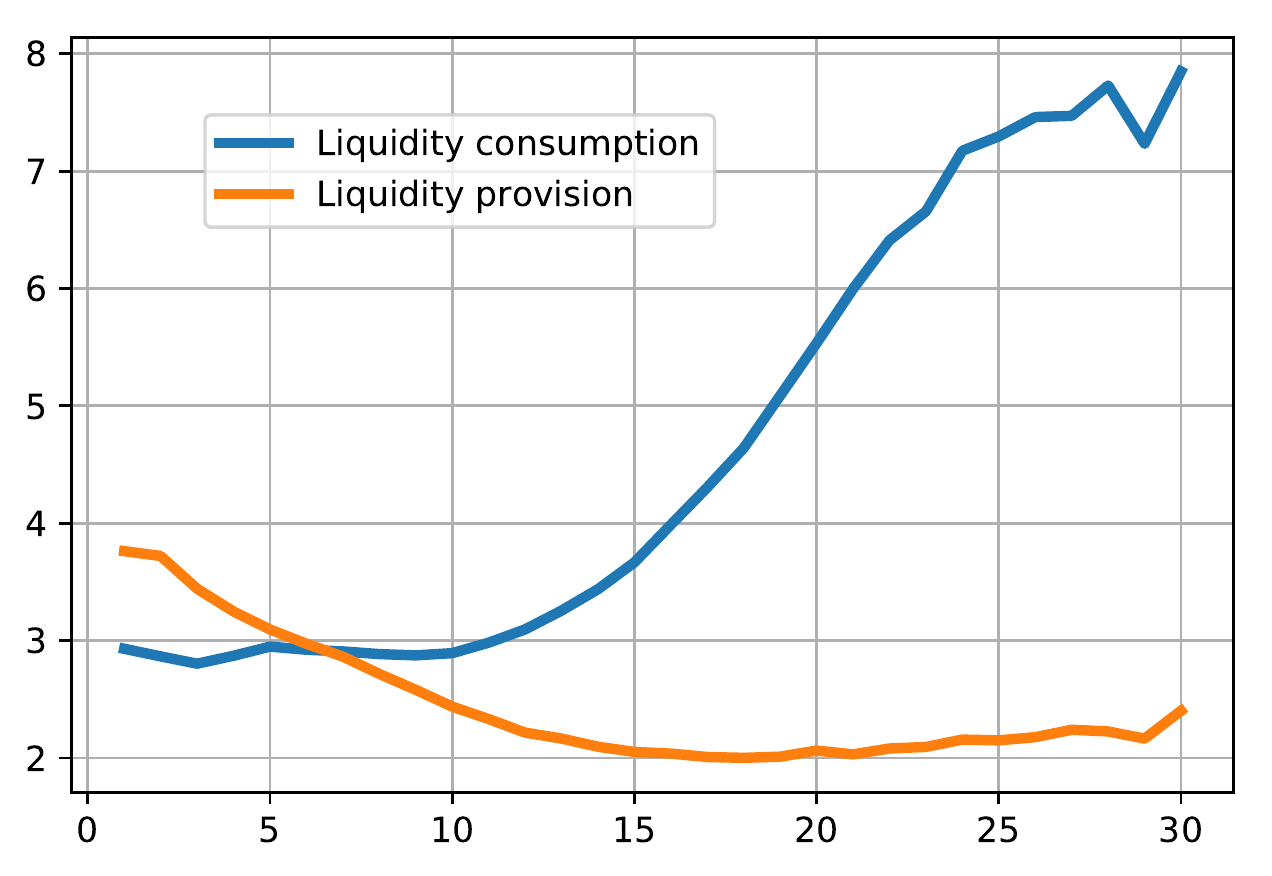}
	    \includegraphics[width=0.47\linewidth]{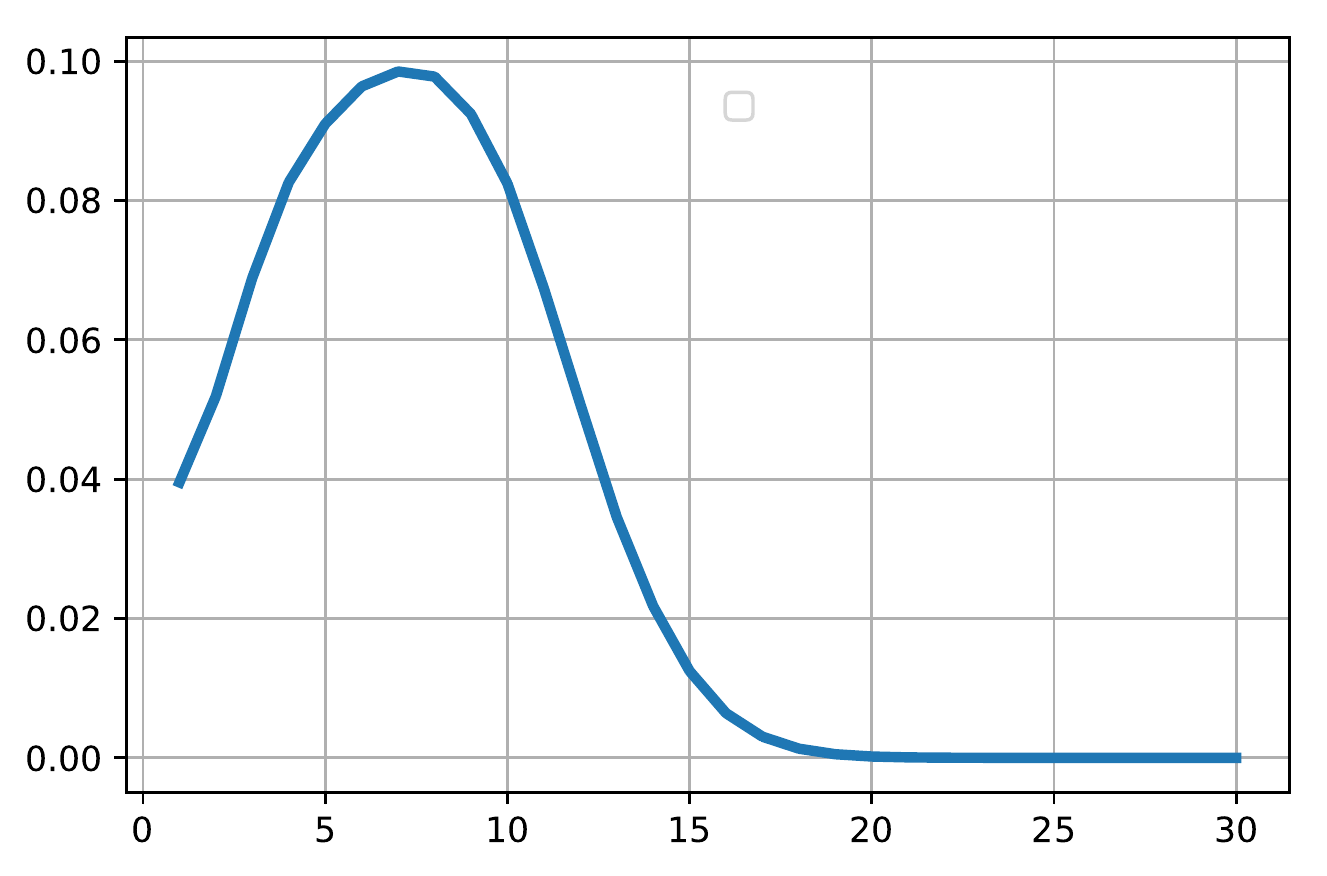}\\
	    \hfill Long term price volatility $\sigma^{2,G} =  0.075$, $\sigma^{2,M}_{10} = 0.490$. \hfill ~\\
	    \caption{(a) Liquidity insertion and consumption intensities (in orders per second) with respect to the queue size (in AES) and (b) the corresponding stationary distribution of $Q^1$ with respect to the queue size (in AES), proper to Michelin.}
	    \label{Michelin}
\end{figure}

\begin{figure}[H]

\centering
    \hfill (a) Intensity of the market \hfill (b) Stationary measure $Q^1$ \hfill ~\\
        \includegraphics[width=0.42\linewidth]{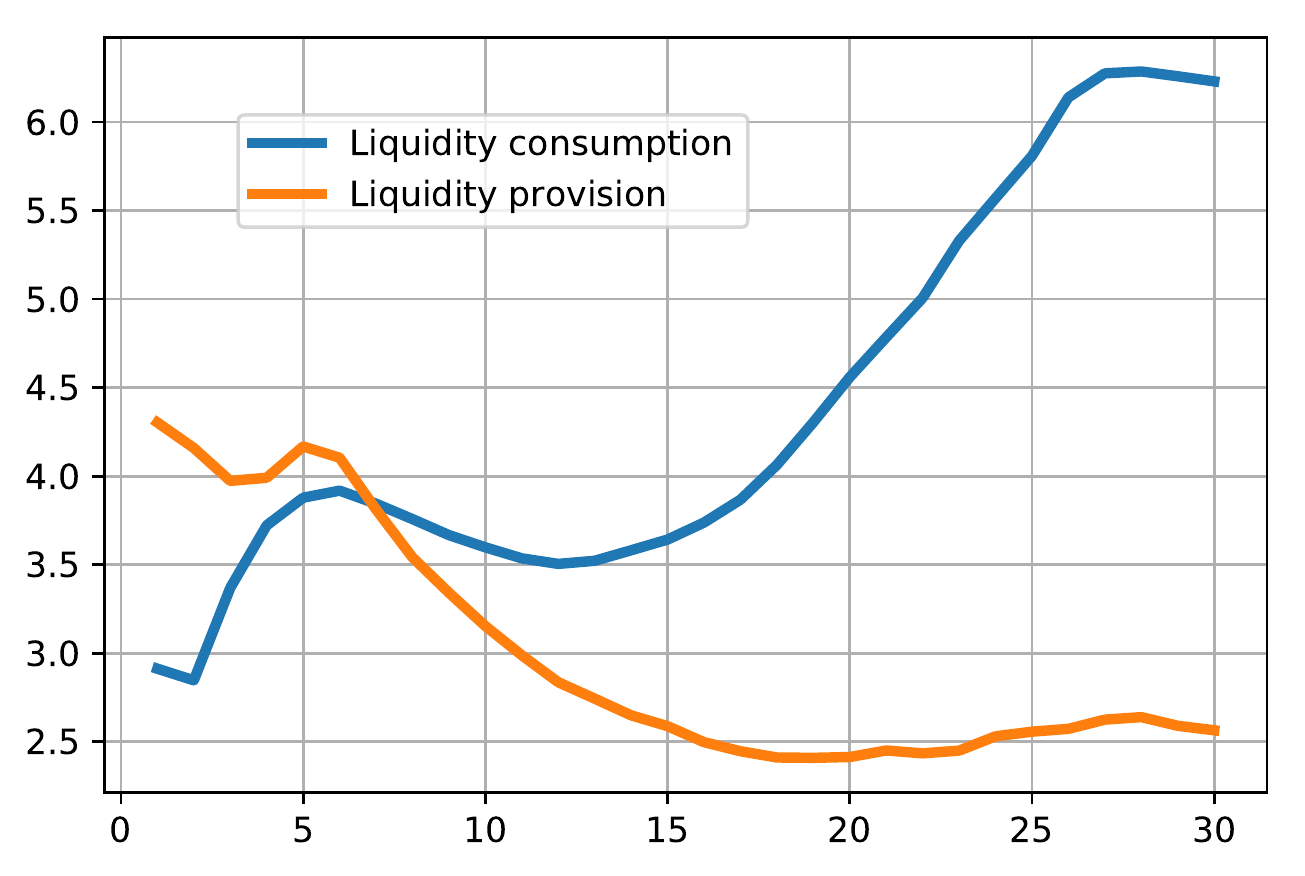}
	    \includegraphics[width=0.47\linewidth]{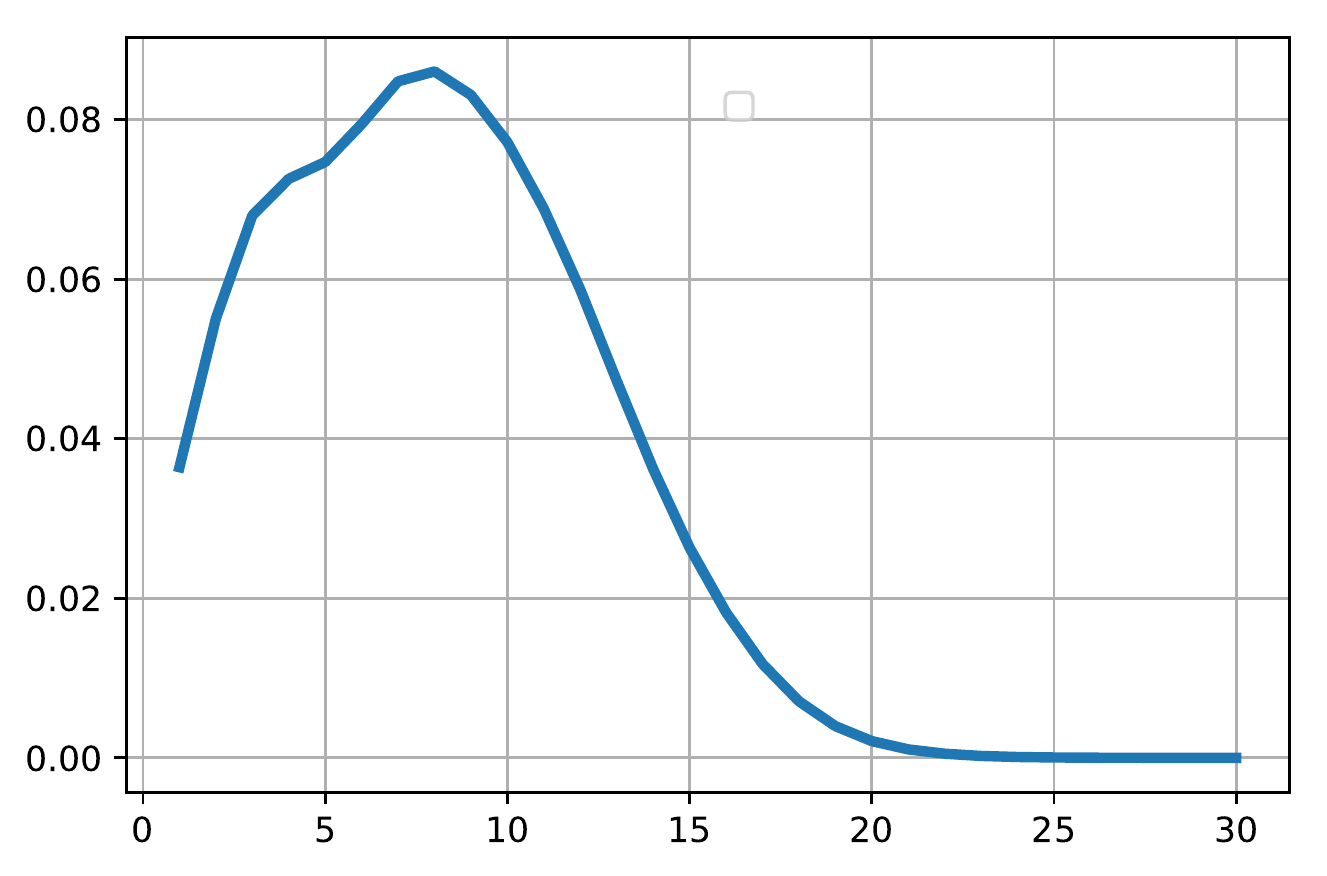}\\
	    \hfill Long term price volatility $\sigma^{2,G} =  0.065$, $\sigma^{2,M}_{10} = 0.453$. \hfill ~\\
	    \caption{(a) Liquidity insertion and consumption intensities (in orders per second) with respect to the queue size (in AES) and (b) the corresponding stationary distribution of $Q^1$ with respect to the queue size (in AES), proper to Orange.}
	    \label{Orange}
\end{figure}
\bibliographystyle{plain}

For each of the market makers, we compute the liquidity consumption and provision intensities, and the corresponding stationary measure that we would obtain in a situation where the studied market maker withdraws from the market and the other market participants do not change their behaviour. We show respectively the results relative to EssilorLuxottica, Michelin and Orange. 

\begin{figure}[H]
\centering
    \hfill Intensities and $\sigma^{2,M}_{10}$ when one market maker leaves the market : stock EssilorLuxottica\hfill ~\\
      \adjustbox{max height=\dimexpr\textheight-.5cm\relax,
           max width=\dimexpr\textwidth-1.cm\relax}{
        \includegraphics[width=1\linewidth]{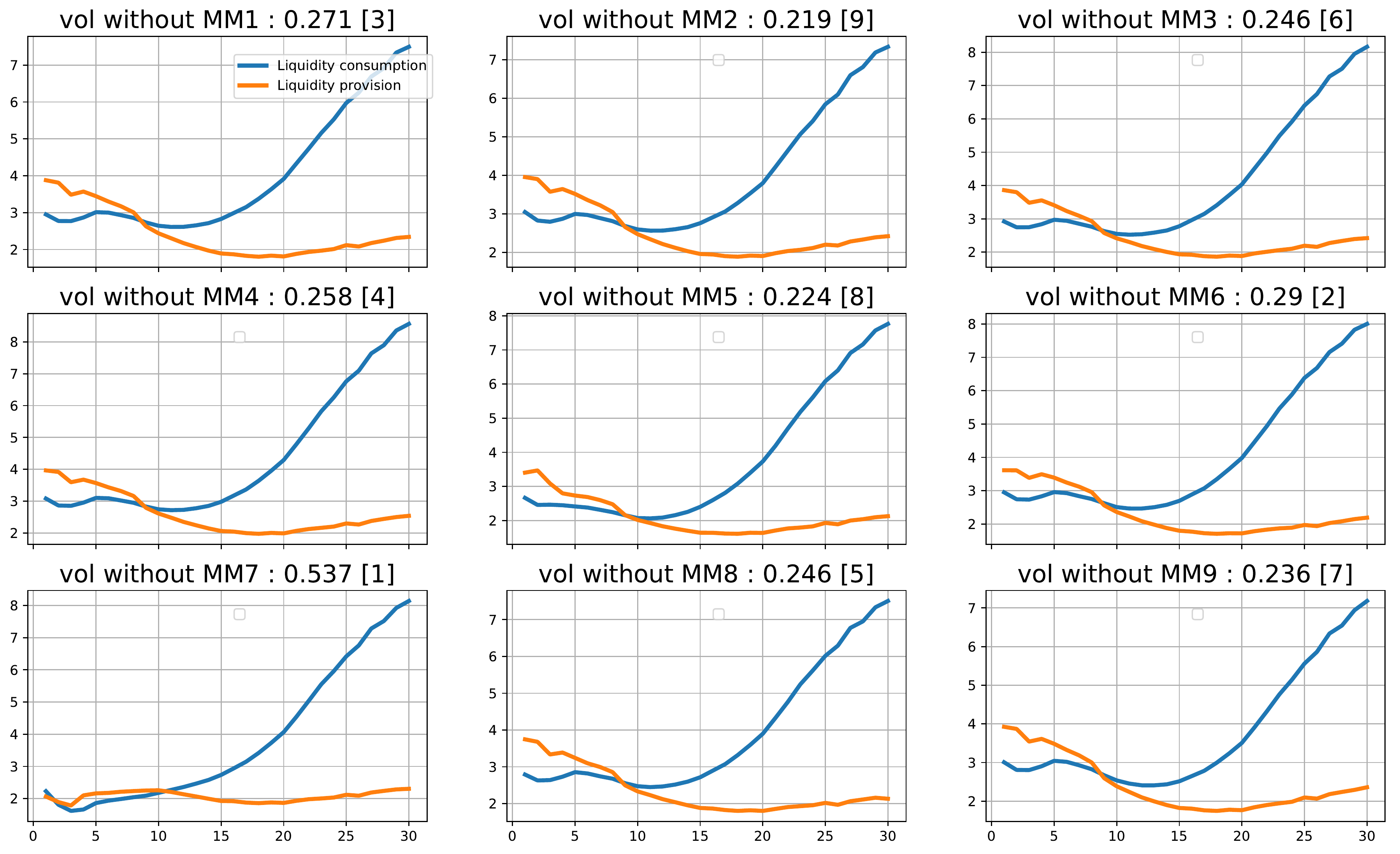}
        }
	    \caption{Liquidity insertion and consumption intensities (in orders per second) with respect to the queue size (in AES) and $\sigma^{2,M}_{10}$ when one market maker is ejected from the market for the stock EssilorLuxottica.}
\end{figure} 

\begin{figure}[H]
\centering
    \hfill Intensities and $\sigma^{2,M}_{10}$ when one market maker leaves the market : stock Michelin\hfill ~\\
      \adjustbox{max height=\dimexpr\textheight-.5cm\relax,
           max width=\dimexpr\textwidth-1.cm\relax}{
        \includegraphics[width=1\linewidth]{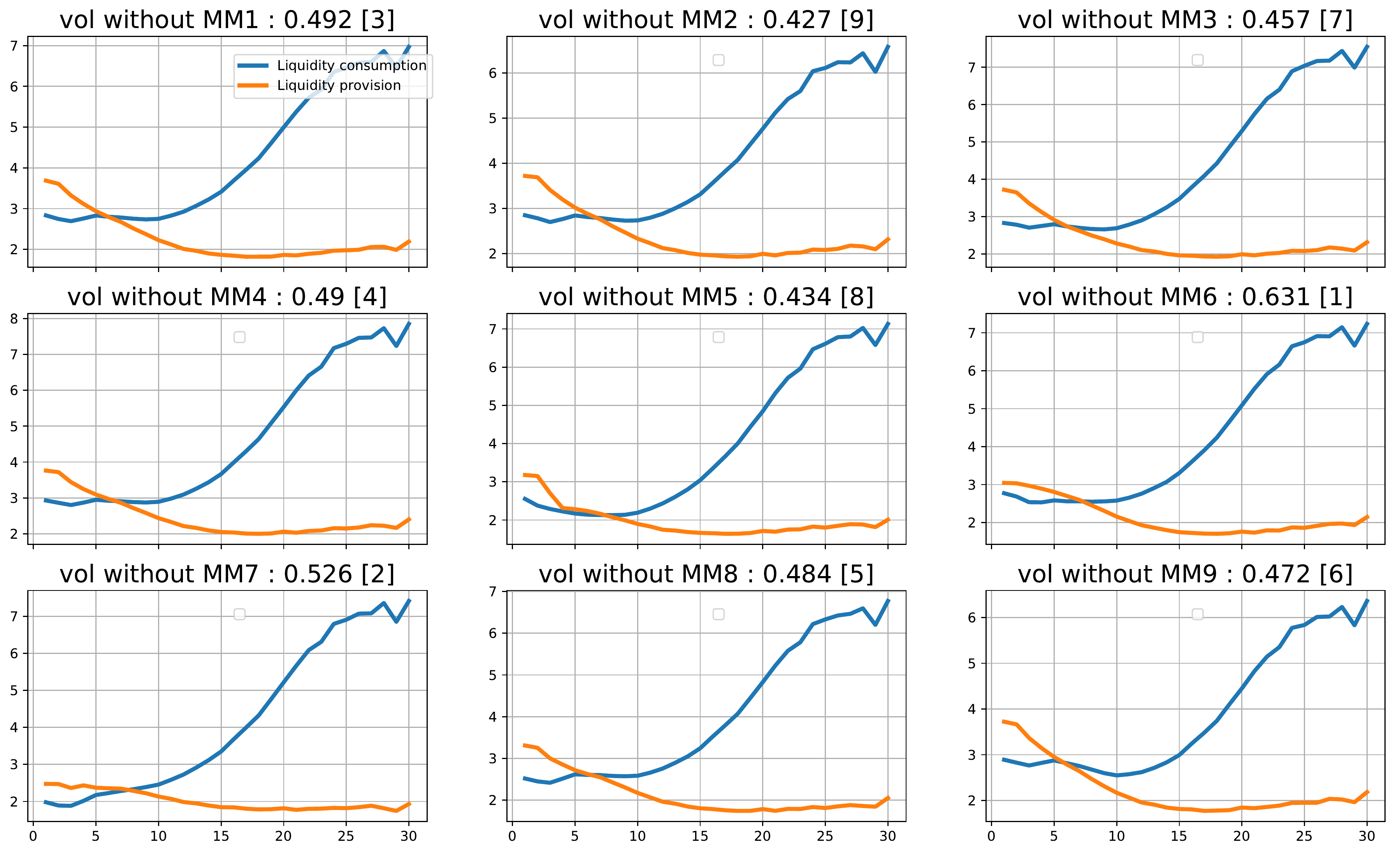}
        }
	    \caption{Liquidity insertion and consumption intensities (in orders per second) with respect to the queue size (in AES) and $\sigma^{2,M}_{10}$ when one market maker is ejected from the market for the stock Michelin.}
\end{figure} 

\begin{figure}[H]
\centering
    \hfill Intensities and $\sigma^{2,M}_{10}$ when one market maker leaves the market : stock Orange\hfill ~\\
      \adjustbox{max height=\dimexpr\textheight-.5cm\relax,
           max width=\dimexpr\textwidth-1.cm\relax}{
        \includegraphics[width=1\linewidth]{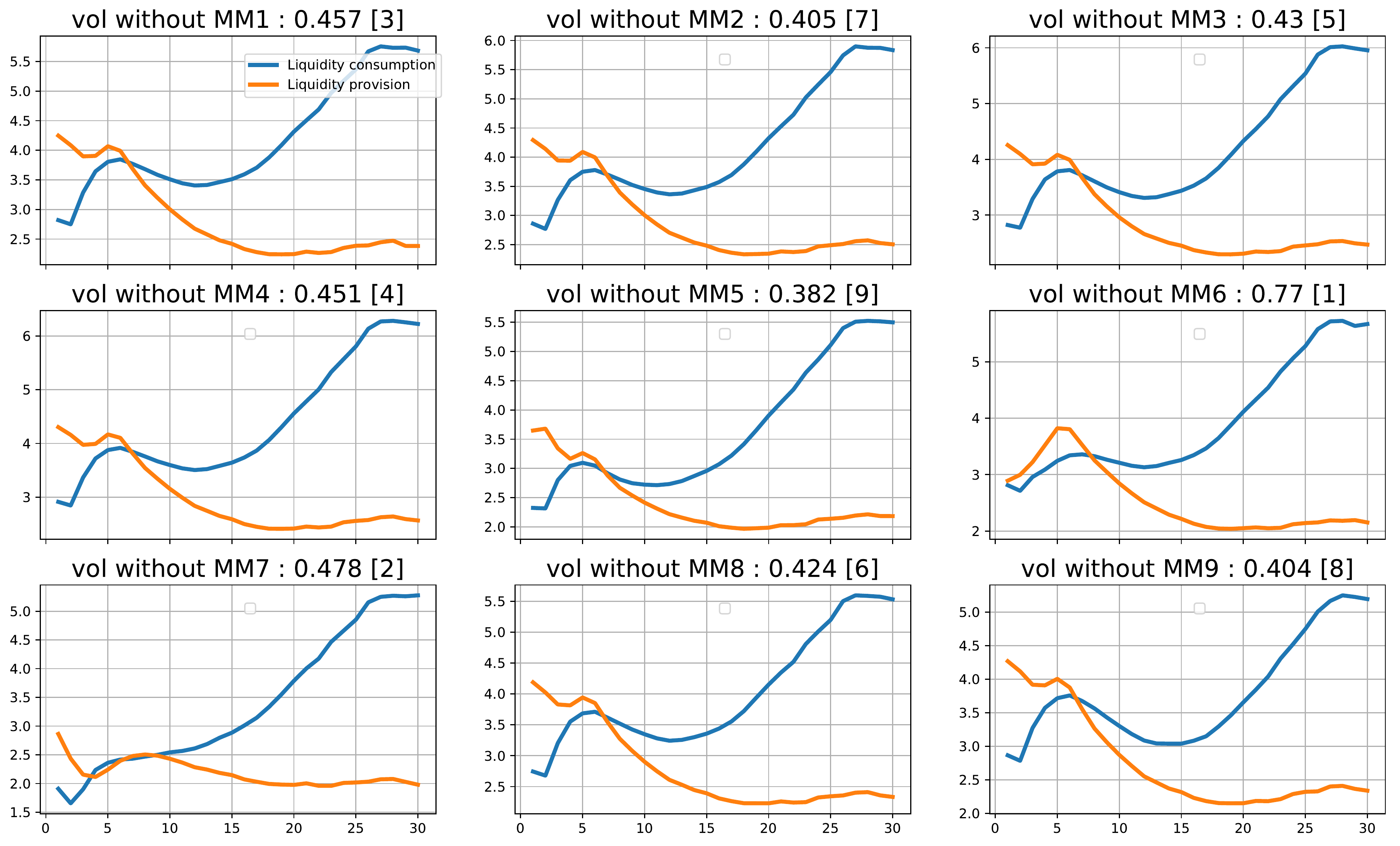}
        }
	    \caption{Liquidity insertion and consumption intensities (in orders per second) with respect to the queue size (in AES) and $\sigma^{2,M}_{10}$ when one market maker is ejected from the market for the stock Orange.}
\end{figure} 
\bibliography{biblio_market_making}
\end{document}